\documentclass[acmtoplas,pdftex]{acmsmall}

\usepackage{amsmath,amssymb,amsbsy,latexsym,stmaryrd}
\usepackage[algoruled,vlined,english]{algorithm2e}

\usepackage{graphicx}
\usepackage[utf8x]{inputenc}
\usepackage{galois}
\usepackage[backgroundcolor=green!40!white]{todonotes}

\usepackage[pdftex,%
 	colorlinks,%
 	pdfauthor={Pierre Ganty, Rupak Majumdar},%
 	hypertexnames=false,%
	pdftitle={Algorithmic Verification of Asynchronous Programs}]{hyperref}

\def\qed{\rule{0.4em}{1.4ex}}

\def\@envspa{\hspace{0.3em}}
\def\@sa{\hspace{-0.2em}}
\def\@sb{\hspace{0.5em}}
\def\@sc{\hspace{-0.1em}}

\newtheorem{construction}{Construction}{\bfseries\upshape}{\itshape}
\newtheorem{assumption}{Assumption}{\bfseries\upshape}{\itshape}

\def\mynote#1{\textcolor{red}{\sf $\clubsuit$ #1$\clubsuit$}}

\makeatletter
\begingroup \catcode `|=0 \catcode `[= 1
\catcode`]=2 \catcode `\{=12 \catcode `\}=12
\catcode`\\=12 |gdef|@xcomment#1\end{comment}[|end[comment]]
|endgroup
\def\@comment{\let\do\@makeother \dospecials\catcode`\^^M=10\def\par{}}
\def\begincomment{\@comment\@xcomment}
\makeatother
\newenvironment{comment}{\begincomment}{}

\def\set#1{{\{ #1 \}}}
\def\tuple#1{{\langle #1 \rangle }}
\def\multi#1{{\llbracket #1 \rrbracket}}
\def\nats{{\mathbb{N}}}

\def\parikh{{\mathsf{Parikh}}}

\def\mmap{\mathbf{m}}

\def\card#1{\lvert {#1} \rvert}
\def\tmenc#1{\lVert {#1} \rVert}
\def\cmap{\mathbf{c}}

\newcommand{\dotcup}{\ensuremath{\mathaccent\cdot\cup}}
\newcommand{\proj}{\mathit{Proj}}

\newcommand{\dcl}[1]{\left\downarrow {#1}\right.}

\newcommand{\ucl}[1]{\left\uparrow {#1}\right.}

\newcommand{\fire}[1]{\left[ {#1}\right\rangle}

\newcommand{\multiset}[1]{{\mathbb{M}[ #1 ]}}

\def\loc{\mathit{loc}}
\def\instr{\mathsf{Instr}}
\def\inc{\mathsf{inc}}
\def\dec{\mathsf{dec}}
\def\zerotest{\mathsf{zerotest}}
\def\typeinst{\mathsf{TypeInst}}
\def\twocm{\mathsf{2CM}}
\def\threecm{\mathsf{3CM}}
\def\ncm{n\mathsf{CM}}
\def\assign{\mathrel{\mathop:}=}
\def\ap{\mathfrak{P}}
\def\pn{\mathsf{PN}}
\def\pnfamily{\mathcal{N}}

\def\pnr{\mathsf{PN+R}}

\def\pnri{\mathsf{PN+\!R+\,!\,}}
\def\cfl{\mathsf{CFL}}
\def\cfg{\mathsf{CFG}}
\def\prod{\mathcal{P}}
\def\cnt{\mathit{cnt}}

\def\tocov{\mathit{2cover}}
\def\Sigmabar{\overline{\Sigma}}
\def\ttn{{\tt n}}
\def\ttw{{\tt w}}
\def\ttsent{{\tt sent}}
\def\ttrecv{{\tt recv}}
\def\ttrpccall{{\tt rpccall}}
\def\ttwrpc{{\tt wrpc}}
\def\stmts{\mathsf{stmts}}

\begin{document}

\title{Algorithmic Verification of Asynchronous Programs}

\date{\today}
\author{PIERRE GANTY
\affil{\textsc{imdea} Software, Madrid, Spain}
{RUPAK MAJUMDAR}
\affil{\textsc{mpi-sws}, Kaiserslautern, Germany}
}

\begin{abstract}
Asynchronous programming is a ubiquitous systems programming idiom to manage
concurrent interactions with the environment. In this style, instead of waiting
for time-consuming operations to complete, the programmer makes a non-blocking
call to the operation and posts a callback task to a task buffer that is
executed later when the time-consuming operation completes. A co-operative
scheduler mediates the interaction by picking and executing callback tasks from
the task buffer to completion (and these callbacks can post further callbacks
to be executed later). Writing correct asynchronous programs is hard because
the use of callbacks, while efficient, obscures program control flow.

We provide a formal model underlying asynchronous programs and study
verification problems for this model.  We show that the safety verification
problem for finite-data asynchronous programs is \textsc{expspace}-complete.
We show that liveness verification for finite-data asynchronous programs is
decidable and polynomial-time equivalent to Petri Net reachability.
Decidability is not obvious, since even if the data is finite-state,
asynchronous programs constitute infinite-state transition systems: both the
program stack and the task buffer of pending asynchronous calls can be
potentially unbounded. 

Our main technical construction is a polynomial-time semantics-preserving
reduction from asynchronous programs to Petri Nets and conversely. The
reduction allows the use of algorithmic techniques on Petri Nets to the
verification of asynchronous programs.

We also study several extensions to the basic models of asynchronous programs
that are inspired by additional capabilities provided by implementations of
asynchronous libraries, and classify the decidability and undecidability of
verification questions on these extensions.

\end{abstract}




\acmformat{Ganty, P. and Majumdar, R. 2011.
Algorithmic verification of asynchronous programs.}

\begin{bottomstuff}
This work is supported by the National Science Foundation, under grants
CCF-0546170, CCF-0702743, and CNS-0720881.  Pierre Ganty was sponsored by the
Comunidad de Madrid's Program {\sc prometidos-cm}  (S2009TIC-1465), by the {\sc
people-cofund}'s program {\sc amarout} (PCOFUND-2008-229599), and by the
Spanish Ministry of Science and Innovation (TIN2010-20639).
A preliminary version of this paper appeared in the ACM-SIGPLAN Symposium
on the Principles of Programming Languages, 2009.
Author's addresses: 
Pierre Ganty,
\textsc{imdea} Software, Madrid, Spain, pierre.ganty@imdea.org;
and
Rupak Majumdar, MPI-SWS, Kaiserslautern, Germany,
rupak@mpi-sws.org.
\end{bottomstuff}

\maketitle

\section{Introduction}
\label{sec:intro}

\emph{Asynchronous programming} is a ubiquitous idiom to manage concurrent
interactions with the environment with low overhead.  In this style of
programming, rather than waiting for a time-consuming operation to complete,
the programmer can make \emph{asynchronous} procedure calls which are stored in
a \emph{task buffer} pending for later execution, instead of being executed
right away. We call \emph{handlers} those procedures that are asynchronously
called by the program.  In addition, the programmer can also make the usual
\emph{synchronous} procedure calls where the caller blocks until the callee
finishes. A co-operative \emph{scheduler} repeatedly picks \emph{pending
handler instances} from the task buffer and executes them atomically to
completion. Execution of the handler instance can lead to further handler being
\emph{posted}. We say that handler \(p\) is posted whenever an instance of
\(p\) is added to the task buffer. The posting of a handler is done using the
asynchronous call mechanism.  The interleaving of different picks-and-executes
of pending handler instances (a pick-and-execute is often referred to as a
\emph{dispatch}) hides latency in the system.  Asynchronous programming has
been used to build fast servers and routers \cite{Flash,Click}, embedded
systems and sensor networks \cite{TinyOS}, and forms the basis of web
programming using Ajax.

Writing correct asynchronous programs is hard.  The loose coupling between
asynchronous calls obscures the control and data flow, and makes it harder to
reason about them.  The programmer must keep track of concurrent interactions,
manage data flow between posted handlers (including saving and passing
appropriate state between dispatches), and ensure progress.  Since the
scheduling and resource management is co-operative and performed by the
programmer, one mis-behaving procedure (e.g., one that does not terminate, or
takes up too many system resources) can bring down the entire system.  

We study the problem of algorithmic verification of {\em safety} and {\em
liveness} properties of asynchronous programs.  Informally, safety properties
specify that ``something bad never happens,'' and liveness properties specify
that ``something good eventually happens.'' For example, a safety property can
state that a web server does not crash while handling a request, and a liveness
property can state that (under suitable fairness constraints) every request to
a server is eventually served.

For our results, we focus on {\em finite-data}
asynchronous programs in which data variables range over a finite domain of values.
Our main results show that the safety verification for finite-data asynchronous
programs is \textsc{expspace}-complete, and the liveness verification problem
is decidable and polynomial-time equivalent to Petri net reachability.
The finiteness assumption on the data is necessary
for decidability results, since all verification questions are already
undecidable for 2-counter machines \cite{Min67}.  However, since the depth of
the stack or the size of the task buffer could both be unbounded, even with
finitely many data values, asynchronous programs define transition systems
with possibly {\em infinitely many states}.

Specifically, we develop algorithms to check that an asynchronous program (1)
reaches a particular data value ({\em global state reachability}, to which
safety questions can be reduced) and (2) terminates under certain fairness
constraints on the scheduler and external events (fair termination, to which
liveness questions can be reduced \cite{Vardi91}).  For fair termination, the
fairness conditions on the scheduler rule out certain undesired paths, in which
for example the scheduler postpones some pending handler forever.  

For sequential programs with {\em synchronous} calls, both safety and liveness
verification problems have been studied extensively, and decidability results
are well known \cite{sp81,BurkartS94,RHS95,BEM97,wal01}.  One simple attempt is to
reduce reasoning about asynchronous programs to reasoning about synchronous
programs by explicitly modeling the task buffer and the scheduling. A way to
model an asynchronous program as a sequential one, is to add a counter
representing the number of pending instances for each handler, increment the
appropriate counter each time a handler is posted, and model the scheduler as a {\em
dispatch loop} which picks a non-zero counter, decrements it, and executes the
corresponding handler code. While the reduction is sound, the resulting system
is infinite state, as the counters modeling the pending handler instances can
be unbounded, and it is not immediate that existing safety and liveness
checkers will be complete in this case (indeed, checking safety and liveness
for recursive counter programs is undecidable in general).

Instead, our decidability proofs rely on a connection between asynchronous
programs and Petri nets \cite{Rei86}, an infinite state concurrency model with
many decidable properties.  In particular, we show an encoding of asynchronous
programs into Petri nets and vice versa.  This enables the reduction of
decision problems on asynchronous programs to problems on Petri nets.  As noted
in \cite{ChaV07,jmpopl07,SenV06}, the connection to Petri nets uses the fact
that the two sources of unboundedness ---unbounded program stack from recursive
synchronous calls and unbounded counters from pending asynchronous calls--- can
be decoupled: while a (possibly recursive) procedure is executing, the number
of pending handler instances can only increase, and the number of pending
handler instances decreases precisely when the program stack is empty.
Accordingly, our proof of decidability proceeds as follows.

First, we note that the change to the state of the task buffer before and after
the dispatch of a handler depends only on the number of times each handler is
posted. Therefore the ordering in which handler have been posted can be simply
ignored.  Thus, while the execution of the handler (in general) defines a
context-free language over the alphabet of handlers, what is important from the
analysis perspective is the Parikh image \cite{Parikh66} of this language.
(Recall that the Parikh image of a word counts the number of occurrences of
each letter in the word, and the Parikh image of a language is the set of
Parikh images of each of its words.) We show that the effect of each handler
can be encoded by a Petri net which is linear in the size of the grammar
representation of the handler. Our Petri net construction builts upon
\cite{Esp97} but extends it so as to satisfy one additional property of crucial
importance for correctness. Given the Petri net encoding of individual
handlers, we can then construct a Petri net that strings together the handlers
according to the semantics of asynchronous programs.  This Petri net is linear
in the size of the asynchronous program and captures in a precise sense the
computations of the asynchronous system.  Moreover, given a Petri net, we can
conversely construct an asynchronous program polynomial in the size of the
Petri net that captures in a precise sense the behaviors of the net, a result
that is useful to prove lower bounds on asynchronous programs.

Safety verification then reduces to checking coverability of the Petri net for
which we can use known decidability results \cite{KarpM69,Rackoff78}.
Together, this gives a tight {\sc expspace}-complete decision procedure for
safety verification of asynchronous programs. (The lower bound follows from
known {\sc expspace}-hardness of Petri net  coverability \cite{Lipton} and an
encoding of an arbitrary Petri net as an asynchronous program that is linear in
the size of the Petri net.) Previous decidability proofs for safety
verification \cite{SenV06,jmpopl07} used backward reachability of
well-structured transition systems \cite{ACJT96} to argue decidability, and did
not yield any upper bound on the complexity of the problem.

An alternate route to safety verification \cite{SenV06} explicitly invokes
Parikh's theorem \cite{Parikh66} to construct, for each handler, a regular
language which has the same Parikh image.  Coupled with our construction of
Petri nets, this gives another algorithm for safety verification.
Unfortunately, this construction does not give a tight complexity bound.  It is
known that the automaton representation of a regular set with the same Parikh
image as a context-free grammar can be at least exponential in the size of the
grammar.  Thus, the Petri net obtained using the methods of \cite{SenV06} can
be exponential in the size of the original asynchronous program.  This only
gives a {\sc 2expspace} upper bound on safety verification (using the {\sc
expspace} upper bound for Petri net coverability \cite{Rackoff78}).

For fair termination, we proceed in two steps.  An asynchronous program can
fail to terminate in two ways.  First, a particular handler execution can loop
forever. Second, each dispatch can terminate, but there can be infinite
sequence of posted handler and dispatches.

For infinite runs of the first kind, the task buffer can be abstracted away (as
no dispatches occur from within a dispatched handler) and we can use a
combination of safety verification (checking that a particular handler can ever
be dispatched) and techniques for liveness checking for finite-state pushdown
systems \cite{BurkartS94,wal01} (checking that a handler loops forever).

The second case above is more interesting, and we focus on this problem.  For
infinite runs of the second form, we note that the Petri net constructed from
an asynchronous program preserves all infinite behaviors, and we can reduce
fair termination of the asynchronous program (assuming each individual
dispatched handler terminates) to an analogous property on the Petri net.  We show
that this property can be encoded in a logic on Petri nets \cite{yen:paths},
which can be reduced to checking certain reachability properties of Petri nets
\cite{RP09}.  Conversely, we show that the Petri net reachability problem can
be reduced in polynomial time to a fair termination question on asynchronous
programs.  Together, we show that the fair termination problem for asynchronous
programs is polynomial-time equivalent to the Petri net reachability problem.
Again, this gives an {\sc expspace}-hard lower bound on the problem
\cite{Lipton}.  On the other hand, the best known upper bounds for Petri net
reachability take non-primitive recursive space
\cite{kosaraju82,Lambert92,Mayr81,Meyer}.  (In the absence of fairness, i.e., for the
termination problem, we get an {\sc expspace}-complete algorithm. Previously,
\cite{cv-tcs09} gave a decision procedure for this problem, but the complexity
of their procedure is not apparent.)
 
The reduction to Petri nets also enables us to provide decision procedures for
related verification questions on asynchronous programs.  First, we show a
decision procedure for {\em boundedness}, a safety property that asserts there
exists some finite $N$ such that the maximum possible size of the task buffer
at any point in any execution is at most $N$.  For the boundedness property we
again use a known result on Petri nets which allows to decide the existence of
an upper bound \(D\) on the size of the task buffer at any point in any
execution (or return infinity, if the task buffer is unbounded).  Since the
task buffer is often implemented as a finite buffer, let us say of size $d$, if
$D>d$ holds then there is an execution of the system that leads to an overflow
of the buffer, and to a possible crash.  Our decision procedure for the
boundedness problem uses the above reduction to Petri nets, and checks
boundedness of Petri nets using standard algorithms in {\sc expspace}.  Second,
the {\em fair non-starvation} question asks, given an asynchronous program and
a fairness condition on executions, whether every pending handler instance is
eventually dispatched (i.e., no pending handler instance waits forever). Fair
non-starvation is practically relevant to ensure that an asynchronous program
(such as a server) is responsive.  We show fair non-starvation is decidable by
showing a reduction to Petri nets. 

We also study safety and liveness verification for natural extensions to
asynchronous programs inspired by features supported in common asynchronous
programming languages and libraries.  For the model of asynchronous programs
where a handler can cancel all pending instances of a handler, we show that
safety is decidable, but boundedness and termination are not.  If in addition,
a handler can test (at most once in every execution) the absence of pending
instances for a specific handler, safety becomes undecidable as well. The
decidability result uses decidability of coverability of Petri nets extended
with reset arcs \cite{ACJT96}. The undecidability results are based
on undecidability of boundedness or reachability of Petri nets with reset
arcs, or the undecidability of reachability of two-counter machines.

%


\section{Informal examples}

We start by giving informal examples of asynchronous programs using, for
readability, a simple imperative language.  We use C-like syntax with an
additional construct {\tt post} $f(e)$ which is the syntax for an
\emph{asynchronous call} to procedure $f$ with arguments $e$.  Operationally,
the execution of {\tt post} $f(e)$ posts handler \(f(e)\): an \emph{instance}
of handler \(f(e)\) is added to the task buffer. 

In the initial state of an asynchronous program, the task buffer is specified
by the programmer and the program stack is empty. Whenever the program stack is
empty, the scheduler dispatches a pending handler instance, if any. The program
\emph{stops} when the scheduler has no pending handler instances to dispatch.

In our formal development, we use a more abstract language acceptor based
model. Compiling our imperative programs to the formal model (assuming all data
variables range over finite types) is straightforward although laborious. 

\subsection{Safety Properties}

\begin{figure}[t]
\centering
\begin{minipage}[t]{.55\textwidth}
\begin{verbatim}
server() {
1: client *c = alloc_client();
2: if (c != 0) {
3:   c->state = TO_READ;
4:   post process_client(c);
   }
5: post server();
}

process_client(c) {
1: if (c->state == TO_READ) {
2:   post read(c); return;
   }
3: if (c->state == DONE_READ) {
4:   post send(c); return;
   }
E: assert(false);
}

read(c) {
1: if (*) {
2:   disconnect(c);   //ERROR: should return here
3: } else {
4:   if (*) { c->state = TO_READ;   }
5:   else   { c->state = DONE_READ; }
6: }
7: process_client(c);
}

send(c) {
1: assert(c->state == DONE_READ);
2: disconnect(c);  //done processing
}

disconnect(c) { // close connection
1: c->state = CLOSED;
2: return;
}

Initially:  server();
\end{verbatim}
\end{minipage}
\caption{Server example with bug}
\label{fig:server}
\end{figure}

Figure~\ref{fig:server} shows an abstracted example of a server that runs in a
loop (procedure {\tt server}) responding to external events to connect.
When a client connects to the server, the server loop allocates a data structure
for the connection, reads data asynchronously, sends data back to the client, and
disconnects. If there is an error reading data, the connection is disconnected.

The implementation uses asynchronous calls to procedures {\tt read} and {\tt
send}.  The server allocates data specific to a connection ({\tt
alloc\_client}), sets the state of the connection to {\tt TO\_READ} and posts
handler {\tt process\_client} to process the connection and posts itself to
wait for the next connection.

The handler {\tt process\_client} performs data read and data send.  It looks
at the state of the connection and posts {\tt read} or {\tt
send} based on the state.  It is an error to execute {\tt process\_client} if
the connection is in any other state (and the code is expected never to reach
the label {\tt E}).

The handler {\tt read} can disconnect a connection based on some error (lines
1,2), or read data.
If the data has not been read completely (modeled by the then-branch of the
non-deterministic conditional on line 4), the state is kept at {\tt TO\_READ}.
If the data has been read completely (modeled by the else-branch of the
non-deterministic conditional on line 4), the state is changed to {\tt DONE\_READ}.
In both cases, the procedure {\tt process\_client} is called (synchronously)
which, in turn, posts {\tt read} or {\tt send}.

The handler {\tt send} closes the connection by calling disconnect.  It expects
a connection whose state {\tt DONE\_READ} denotes data has been read (the
assertion on line 1), and the state is marked {\tt CLOSED}.

The example is representative of many server implementations, and demonstrates
the difficulty of writing asynchronous programs.  The sequential flow of
control, in which a connection is accepted, data is read, data is sent to the
client, and the connection is closed, gets broken into individual handlers and
the control flow is obscured.  Moreover, the state space can be unbounded as an
arbitrary number of connections can be in flight at the same time.

For correct behavior of the server, the programmer expects the connection is in
specific states at various stages of processing.  These are demonstrated by the
assertions in the code.

In this example, the assertion in {\tt send} holds for all program executions,
but the assertion in {\tt process\_client} does {\em not}.  The assertion in
{\tt send} holds because the condition is checked in {\tt process\_client}
(line 3) before {\tt send} is posted.  However, there can be an arbitrary delay
between the check and the execution of {\tt send} for this connection, with any
number of other connections executing in the middle.

The assertion in {\tt process\_client} can be violated in an execution which
{\tt read} terminates a connection on line 2 by calling {\tt disconnect} (which
sets the state to {\tt CLOSED}), and subsequently {\tt process\_client} is
called on line 7.  The bug occurs because the author forgot a {\tt return} on
line 2 after the {\tt disconnect}.

Our first goal is to get a sound and complete algorithm which can automatically
check an asynchronous program for safety properties such as assertions.

\begin{figure}[t]
\centering
\begin{minipage}[t]{.55\textwidth}
\begin{verbatim}
global int sent = 0, recv = 0;
global int n, w;
wrpc() {
  if (recv < n) {
    if (sent < n && sent - recv < w) {
      post rpccall();
      sent++;
    }
    post wrpc();
  } else {
    return;
  }
}
rpccall() { recv ++; }
Initially: wrpc();
\end{verbatim}
\end{minipage}
\caption{Windowed RPC implementation \label{fig-windowed-rpc} }%
\end{figure}

\subsection{Liveness Properties}

Figure~\ref{fig-windowed-rpc}  shows a simplified asynchronous implementation
of {\em windowed RPC}, in which a client makes $\ttn$ asynchronous procedure
calls in all, of which at most $\ttw\leq \ttn$ are pending at any one time.
(Assume that $\ttn$ and $\ttw$ are fixed constants.)
Windowed RPC is a common systems programming idiom which enables concurrent
interaction with a server without overloading it.

The windowed RPC client is implemented in the procedure $\ttwrpc$.  Two global
counters, $\ttsent$ and $\ttrecv$, respectively track the number times
\(\ttrpccall\) has been posted and the number pending instances of
\(\ttrpccall\) that have completed. The server is abstracted by the procedure
$\ttrpccall$ which increments \(\ttrecv\).  The procedure $\ttwrpc$ first
checks how many instances of \(\ttrpccall\) have completed.  If the number is
$\ttn$ or more, it terminates.  Otherwise if fewer than \(\ttn\) instances to
\(\ttrpccall\) have been posted and the number of pending instances (equal to
$\ttsent - \ttrecv$) is lower than the window size \(\ttw\) then \(\ttwrpc\)
posts \(\ttrpccall\). Finally, $\ttwrpc$ posts itself (this is done by an
asynchronous recursive call), either to further post handlers or to wait for
pending instances of \(\ttrpccall\)  to complete.

As mentioned in \cite{KKK07}, already in this simple case, asynchronous code
with windowed control flow is quite complex as the control decisions are spread
across multiple pieces of code.

Consider the desirable property that the windowed RPC fairly terminates, which
implies that, at some point in time, every pending instances of \(\ttrpccall\)
completed and the task buffer is empty. Informally, this property is true
because $\ttwrpc$ posts $\ttrpccall$ at most $\ttn$ times, and posts itself
only as long as $\ttrecv$ is less than $\ttn$.  Each execution of $\ttrpccall$
increments $\ttrecv$, so that after $\ttn$ dispatches of $\ttrpccall$, the
value of $\ttrecv$ reaches $\ttn$, and from this point, each dispatch to
$\ttwrpc$ does not post new handler. Thus, eventually, the task buffer becomes
empty.

Notice that we need the assumption that the scheduler \emph{fairly} dispatches
pending handlers: a post to $q$ is followed by a dispatch of $q$.  Without that
assumption the program does not terminate: consider the infinite run where
the scheduler always picks $\ttwrpc$ in preference to $\ttrpccall$.

\smallskip %
\noindent %
{\bf Fair Termination.} An asynchronous program \emph{fairly
terminates} if $(i)$ every time a procedure is called (synchronously or
asynchronously), it eventually returns; and $(ii)$ there is no infinite run
that is fair.  An infinite run is said to be {\em fair} if for every handler
\(q\) and for every step along the run, a pending instance to \(q\)
is followed by a dispatch of \(q\).  The fairness constraint is expressible as
a $\omega$-regular property.

Of course, for most server applications, the asynchronous program implementing
the server should {\em not} terminate (indeed, termination of a server points
to a bug).


\begin{figure}[t]
	\centering
	\begin{minipage}[t]{.2\textwidth}
\begin{verbatim}
global bit = 0;
h1() {
  if (bit == 0) {
    post h1();
    post h2();
  }
}

h2() {
  bit = 1;
}
Initially: h1();
\end{verbatim}
	\end{minipage}
\caption{A fairly terminating asynchronous program\label{fig-ex-unbdd}}%
\end{figure}

\smallskip %
\noindent %
{\bf Fair Non-starvation.}
A second ``progress condition'' is fair non-starvation.
When an asynchronous program does not terminate, we can still require that
$(i)$ every execution of a procedure that is called
(synchronous or asynchronous) eventually returns; and $(ii)$ along every fair infinite
run no handler is starved.  A starving handler corresponds to a particular
pending handler instance which is never dispatched, and hence which waits
forever to be executed. Consider a handler {\tt h} that posts itself twice. A
fair infinite execution dispatches {\tt h} infinitely often, even though a
particular pending instance to {\tt h} may never get to run.

Our second goal is to provide sound and complete algorithms to check fair termination
and fair non-starvation properties of asynchronous programs.

\smallskip
Proving safety and liveness properties for asynchronous
programs is difficult for several reasons.
First, as the server and the windowed RPC example suggests,
reasoning about termination may require reasoning about the dataflow
facts (e.g., the fact that the state is checked to be {\tt DONE\_READ} before
posting {\tt send} in server or that $\ttrecv$ eventually reaches $\ttn$ in RPC).
Second, at each point, there can be an unbounded number of pending handler instances.
This is illustrated by the program in Fig.~\ref{fig-ex-unbdd}, which terminates on
each fair execution, but in which the task buffer contains unboundedly many
pending instances (to {\tt h2}).  Third, each handler can potentially be recursive,
so the program stack can be unbounded as well.

We remark that if the finite dataflow
domain induces a sound abstraction of a concrete asynchronous program in which
data variables range over infinite domains, that is, if the finite abstraction has
more behaviors, then our analysis is sound:
if the analysis with the finite dataflow domains shows the asynchronous program fairly terminates (resp.\ is
fair non-starving) then the original asynchronous program fairly
terminates (resp.\ is fair non-starving).

\section{Preliminaries}

\subsection{Basics}

An \emph{alphabet} is a finite non-empty set of \emph{symbols}.  For an
alphabet $\Sigma$, we write $\Sigma^*$ for the set of finite sequences of
symbols (also called \emph{words}) over $\Sigma$.  A set $L\subseteq\Sigma^*$
of words defines a \emph{language}.  The length of a word $w\in\Sigma^*$,
denoted $\card{w}$, is defined as expected.  An infinite word \(\omega\)
alphabet $\Sigma$ is an infinite sequence of symbols.  For a finite non-empty
word $w\in \Sigma^*\setminus\set{\varepsilon}$, we write $w^{\omega}$ for the
infinite word given by the infinite repetition of $w$, that is, $w\cdot w \cdot
w \cdots $.  The projection of word \(w\) onto some alphabet $\Sigma'$, written
$\proj_{\Sigma'}(w)$, is the word obtained by erasing from \(w\) each symbol
which does not belong to \(\Sigma'\).  For a language $L$, define
\(\proj_{\Sigma'}(L)=\set{\proj_{\Sigma'}(w)\mid w \in L}\).

A \emph{multiset} $\mmap\colon \Sigma\rightarrow\nats$ over $\Sigma$ maps each
symbol of $\Sigma$ to a natural number.
Let $\multiset{\Sigma}$ be the set of all multisets over $\Sigma$.
We treat sets as a special case of multisets 
where each element is mapped onto $0$ or $1$.

We sometimes write
$\mmap=\multi{q_1,q_1,q_3}$ for the multiset
$\mmap\in\multiset{\set{q_1,q_2,q_3,q_4}}$ such that $\mmap(q_1)=2$,
$\mmap(q_2)=\mmap(q_4)=0$, and $\mmap(q_3)=1$. The empty multiset \(\multi{}\) is denoted
\(\varnothing\).  The size of a multiset $\mmap$, denoted $\card{\mmap}$, is
given by $\sum_{\gamma\in\Sigma}\mmap(\gamma)$.
Note that this definition applies to sets as well.

Given two multisets $\mmap,\mmap'\in\multiset{\Sigma}$ we define $\mmap\oplus
\mmap'\in\multiset{\Sigma}$ to be multiset such that $\forall a\in\Sigma\colon
(\mmap\oplus \mmap')(a)=\mmap(a)+\mmap'(a)$, we also define the natural order
$\preceq$ on $\multiset{\Sigma}$ as follows: $\mmap\preceq\mmap'$ if{}f there
exists $\mmap^{\Delta}\in\multiset{\Sigma}$ such that
$\mmap\oplus\mmap^{\Delta}=\mmap'$.

Given $\mmap$, we define $\dcl{\mmap}$ and
$\ucl{\mmap}$ to be the \emph{downward closure} and \emph{upward closure of}
$\mmap$, defined by $\set{\mmap'\in\multiset{\Sigma} \mid \mmap'\preceq \mmap}$
and $\set{\mmap'\in\multiset{\Sigma} \mid \mmap\preceq \mmap'}$, respectively.
The downward and upward closure are naturally extended to sets of multisets.

For $\Sigma\subseteq \Sigma'$ we regard $\mmap\in\multiset{\Sigma}$ as a
multiset of $\multiset{\Sigma'}$ where undefined values are sent to $0$.
We define the projection of
$\mmap'\in\multiset{\Sigma'}$ onto $\Sigma\subseteq\Sigma'$ as the multiset
$\mmap\in\multiset{\Sigma}$ such that $\forall \sigma\in\Sigma\colon
\mmap(\sigma)=\mmap'(\sigma)$.  
We write this as follows \(\proj_{\Sigma}(\mmap')\).

The {\em Parikh image} $\parikh\colon\Sigma^* \rightarrow \multiset{\Sigma}$
maps a word $w\in \Sigma^*$ to a multiset $\parikh(w)$ such that
$\parikh(w)(a)$ is the number of occurrences of $a$ in $w$.  For example,
$\parikh(abbab)(a)=2$, $\parikh(abbab)(b)=3$ and \(\parikh(\varepsilon)=\varnothing\).  For a language $L$, we define
$\parikh(L) = \set{\parikh(w)\mid w\in L}$.  Given an alphabet \(\Sigma'\),
define $\parikh_{\Sigma'}$ to be the function $\parikh\comp\proj_{\Sigma'}$
where $\comp$ denotes the function composition.

\subsection{Formal Languages}\label{sec:_formal_languages}

A \emph{context-free grammar} (\(\cfg\) for short) $G$ is a tuple
$(\mathcal{X},\Sigma,\prod)$ where $\mathcal{X}$ is a finite set of variables
(non-terminal letters), $\Sigma$ is an alphabet of terminal letters and $\prod
\subseteq \mathcal{X}\times (\Sigma \cup \mathcal{X})^*$ 
a finite set of
productions (the production $(X,w)$ may also be noted $X\rightarrow w$).  Given
two strings $u,v \in (\Sigma \cup \mathcal{X})^*$ we define the relation $u
\underset{G}\Rightarrow v$, if there exists a production $(X, w)\in\prod$ and
some words $y,z \in (\Sigma \cup \mathcal{X})^*$ such that $u=yXz$ and $v=ywz$.
We use $\underset{G}\Rightarrow^*$ for the reflexive transitive closure of
$\underset{G}\Rightarrow$.  A word $w\in \Sigma^*$ is {\em recognized} (we also
say \emph{accepted}) from the state $X\in \mathcal{X}$ if
$X\underset{G}\Rightarrow^* w$.  We sometimes simply write \(\Rightarrow\)
instead of \(\underset{G}\Rightarrow\) if \(G\) is clear from the context.

An \emph{initialized context-free grammar} \(G\) is given by a tuple
\((\mathcal{X},\Sigma,\prod,X_0)\) where \((\mathcal{X},\Sigma,\prod)\) is a
\(\cfg\) and \(X_0\in\mathcal{X}\) is the \emph{initial variable}. 
When the initial variable is clear from the context, we simply say context-free grammar.

We define the language of an initialized \(\cfg\) \(G\), denoted $L(G)$, as
$\set{w\in\Sigma^*\mid X_0\Rightarrow^* w}$. A language $L$ is
\emph{context-free} (written \(\cfl\)) if there exists an initialized \(\cfg\)
$G$ such that $L=L(G)$.

A \emph{regular grammar} $R$ is a context-free grammar such that each production is in
$\mathcal{X}\times \bigl((\Sigma\cdot\mathcal{X})\cup \set{\varepsilon}\bigr)$.  It is
known that a language $L$ is \emph{regular} if{}f $L=L(R)$ for some initialized
regular grammar $R$.

We usually use the letters $G$ and $R$ to denote grammars and regular grammars,
respectively.  Given a \(\cfg\) $G=(\mathcal{X},\Sigma,\prod)$ its \emph{size},
denoted $\tmenc{G}$, is given by $\card{\mathcal{X}}+\card{\Sigma}+\sum\set{
\card{Xw} \mid (X,w)\in\prod}$.

We will use the following result from language theory in our proofs.
\begin{lemma}{(Parikh's Lemma \cite{Parikh66})}
For any context free language $L$ there is an effectively computable regular
language $L'$ such that $\parikh(L) = \parikh(L')$.
\label{lem:parikh}
\end{lemma}

Any two languages $L$ and $L'$ such that $\parikh(L)=\parikh(L')$ are said to
be \emph{Parikh-equivalent}.

Throughout the paper, we make the following assumption without loss of generality.
\begin{assumption}
\(\prod\subseteq \bigl(\mathcal{X} \times (\mathcal{X}^2 \cup \Sigma \cup \set{\varepsilon})\bigr)\) for every \(\cfg\) \(G=(\mathcal{X},\Sigma,\prod)\). 
\label{ass:normalform}
\end{assumption}
It has been shown, see for instance in \cite{LL10}, that every \(\cfg\) can be
transformed, in polynomial time, into an equivalent grammar of the above form.

\section{Formal Model}

As noted in the informal example, our formal model consists of three ingredients:
a global store of data values, a set of potentially recursive handlers, and a task buffer that maintains
a multiset of pending handler instances.
We formalize the representation using asynchronous programs.

\subsection{Asynchronous Programs}

An asynchronous program \(\ap = (D, \Sigma, \Sigma_i, G, R, d_0, \mmap_0)\) consists of a
finite set of {\em global states} $D$, 
an alphabet $\Sigma$ of \emph{handler names}, 
an alphabet $\Sigma_i$ of {\em internal actions} disjoint from $\Sigma$, 
a \(\cfg\) $G = (\mathcal{X},\Sigma\cup\Sigma_i,
\prod)$,
a regular grammar \(R=(D,\Sigma\cup\Sigma_i,\delta)\), 
a multiset $\mmap_0\in\multiset{\Sigma}$ of initial pending
handler instances, and an initial state $d_0 \in D$.  
We assume that for each $\sigma\in\Sigma$, there is a non-terminal $X_\sigma\in\mathcal{X}$ of $G$.

A {\em configuration} $(d, \mmap) \in D\times \multiset{\Sigma}$ of \(\ap\)
consists of a global state $d$ and a multiset $\mmap$ of pending handler instances.  For a
configuration $c$, we write $c.d$ and $c.\mmap$ for the global state and the
multiset in the configuration respectively.  The \emph{initial} configuration
\(c_0\) of \(\ap\) is given by \(c_0.d=d_0\) and \(c_0.\mmap=\mmap_0\).

The semantics of an asynchronous program is given as a labeled transition
system over the set of configurations, with a transition relation $\rightarrow
\subseteq (D\times\multiset{\Sigma})\times \Sigma \times
(D\times\multiset{\Sigma})$ defined as follows: let \(\mmap,\mmap'\in\multiset{\Sigma}\), \(d,d'\in D\) and \(\sigma\in\Sigma\)
\begin{gather*}
	(d,\mmap\oplus\multi{\sigma})\overset{\sigma}\rightarrow (d',\mmap\oplus\mmap')\\
\text{ if{}f }\\
\exists w\in (\Sigma\cup\Sigma_i)^* \colon d \underset{R}\Rightarrow^* w \cdot d' \land X_{\sigma} \underset{G}\Rightarrow^* w \land \mmap' = \parikh_{\Sigma}(w)\enspace .
\end{gather*}

Intuitively, we model the (potentially recursive) code of a handler using a
context-free grammar.  The code of a handler does two things: first, it can
change the global state (through $R$), and second, it can add new pending handler instances
(through derivation of a word in $\Sigma^*$).  Together, the transition
relation $\rightarrow$ states that there is a transition from configuration
$(d, \mmap\oplus\multi{\sigma})$ to $(d',\mmap\oplus\mmap')$ if there is an execution of
handler $\sigma$ that changes the global state from $d$ to $d'$ and adds to the task buffer
the handler instances given by $\mmap'$. Note that the multiset
$\mmap$ (the current content of the task buffer minus the pending handler
instance \(\sigma\)) is unchanged while \(\sigma\) executes, and that the order
in which the handler instances are added to the task buffer is immaterial
(hence, in our definition, we take the Parikh image of $w$).

Finally, we conclude from the definition of their semantics that asynchronous
programs satisfy the following form of \emph{monotonicity}. Let us first define
the ordering \(\sqsubseteq\subseteq
(D\times\multiset{\Sigma})\times(D\times\multiset{\Sigma})\) such that
\(c\sqsubseteq c'\) if{}f \(c.d=c'.d \land c.\mmap\preceq c'.\mmap\). Also we have:
\[\forall \sigma\in\Sigma\,\forall c_1\,\forall c_2\,\forall c_3\,\exists
c_4\colon c_1\overset{\sigma}\rightarrow c_2 \land c_1\sqsubseteq c_3 \text{ implies }
c_3\overset{\sigma}\rightarrow c_4 \land c_2\sqsubseteq c_4 \enspace .\]
Therefore, as already pointed in \cite{SenV06,cv-tcs09}, the transitions system
\(\bigl((D\times\multiset{\Sigma},\sqsubseteq),\rightarrow,c_0\bigr)\) defined
by asynchronous programs are \emph{well-structured transition systems} as given
in \cite{ACJT96,FS01}.

A {\em run} of an asynchronous program is a finite or infinite sequence
\[
c_0 \overset{\sigma_0}\rightarrow c_1 \cdots c_k\overset{\sigma_{k}}\rightarrow c_{k+1}  \cdots
\]
of configurations $c_i$ starting from the initial configuration \(c_0\).  A
configuration $c$ is {\em reachable} if there is a finite run \( c_0
\overset{\sigma_0}\rightarrow \cdots \overset{\sigma_{k-1}}\rightarrow c_k \)
with $c_k=c$.

A handler $\sigma\in\Sigma$ is {\em pending} at a configuration \(c\) if $c.\mmap(\sigma) > 0$.
The handler $\sigma$ is said to be {\em dispatched} in the transition $c \overset{\sigma}\rightarrow c'$.

An infinite run
\( c_0 \overset{\sigma_0}\rightarrow \cdots c_k\overset{\sigma_{k}}\rightarrow  \cdots \)
is {\em fair} if for every $\sigma\in\Sigma$, if $\sigma$ is dispatched only
finitely many times along the run, then $\sigma$ is not pending at $c_j$ for
infinitely many $j$'s.  Intuitively, an infinite run is unfair if at some point
some handler is pending and is never dispatched.

For complexity considerations, we encode an asynchronous program as follows.
The grammar \(G\) and \(R\) are encoded as given in
Sect.~\ref{sec:_formal_languages}.  The initial multiset is encoded as a list
of pairs $(\sigma, \mmap_0(\sigma))$, and using a binary representation for
$\mmap_0(\sigma)$.  The {\em size} of an asynchronous program $A$ encoded as
above is denoted $\tmenc{A}$.

\subsection{From Program Flow Graphs to Asynchronous Programs}

We briefly describe how program flow graphs can be represented formally as
asynchronous programs.

We represent programs using control flow graphs \cite{ASU86}, one for each
procedure. The set of procedure names is denoted \(\Sigma\).  The \emph{control flow
graph} for a procedure $\sigma\in\Sigma$ consists of a labeled, directed graph
$(V_\sigma,E_\sigma)$, together with a unique entry node $v^e_\sigma\in
V_\sigma$, a unique exit node $v^x_\sigma\in V_\sigma$, and an edge labeling
which labels each edge with either a statement (such as assignments or
conditionals) taken from a set $\stmts$, or a {\em synchronous} procedure call
(that gets executed immediately) or an {\em asynchronous} procedure call (that
gets added to the task buffer). The nodes of the control flow graph correspond
to control points in the procedure, the entry and exit nodes represent the
point where execution begins and ends, respectively.  Moreover, control flow
graphs are well-formed: every node of \(V_{\sigma}\) is reachable from
\(v^e_{\sigma}\) and co-reachable from \(v^x_{\sigma}\).  We allow arbitrary
recursion.

Let $D$ be a fixed finite set of dataflow values.  We assume that there is an
abstract transfer function $M \colon D \times (\Sigma\cup\stmts) \rightarrow D$
which maps dataflow values and statements to a dataflow value, and captures the
abstract semantics of the program.

Let us now define an asynchronous program \(\ap=(D,\Sigma,\stmts, G,R,d_0,\mmap_0)\).
The reasoning underlying the definition of \(\ap\) is to map the control flow
graphs to \(G\) and the abstract transfer function to \(R\).

We define the \(\cfg\) $G=(\mathcal{X},\Sigma\cup\stmts,\prod)$ where the set
of nonterminals $\mathcal{X}$ is the set of all nodes in all
control flow graphs.

The set of productions $\prod$ is defined as the smallest set such that:
\begin{itemize}
	\item \((X \rightarrow \sigma \cdot Y)\in\prod\) if the edge $(X,Y)$ in the
		control flow graph is labeled with an asynchronous call to procedure
		$\sigma\in\Sigma$;
	\item \((X \rightarrow \mathit{st}\cdot Y)\in\prod\) if the edge $(X,Y)$ is
		labeled with a statement $\mathit{st}\in\stmts$;
	\item \( (X \rightarrow v^e_\sigma \cdot Y)\in\prod\) if the edge $(X,Y)$ is
		labeled with a synchronous call to procedure $\sigma\in\Sigma$;
   \item \((v^x_\sigma\rightarrow \varepsilon)\in\prod\) for each procedure $\sigma\in\Sigma$.
\end{itemize}

Assumption~\ref{ass:normalform} does not hold on \(G\). However it can be
enforced easily (in this case in linear time) by replacing productions of the
form \(X\rightarrow \gamma\cdot Y\) (\(\gamma\in(\Sigma\cup\stmts)\)) by
\(X\rightarrow G \cdot Y\) and \(G\rightarrow \gamma\)) where \(G\) is a 
fresh variable.

We define the regular grammar \(R=(D,\Sigma\cup\stmts,\delta)\) where
\(\delta=\set{ d\rightarrow \mathit{st}\cdot d'\mid d,d'\in D\land
\mathit{st}\in \Sigma\cup\stmts \land M(d,\mathit{st})=d'}\).

Let \(\sigma_0\in\Sigma\) be the main procedure. Intuitively, a leftmost
derivation in the grammar \(G\) starting from \(v^e_{\sigma_0}\) corresponds to
an interprocedurally valid path in the program. The derived word is the
sequence of asynchronous calls to procedures of \(\Sigma\) and statements of
\(\stmts\) made along that path. The global state is given by executing the
program along the path with the abstract semantics specified by $M$ on the
domain $D$ starting from an initial dataflow value \(d_{\imath}\).  Therefore,
\(\ap\) is such that \(\mmap_0=\multi{\sigma_0}\) and \(d_0=d_{\imath}\).

\begin{remark}
Observe that by modelling handlers using language acceptors we are
abstracting away the non terminating executions within a handler.
\label{rmk:finiterunsonly}
\end{remark}

\subsection{A Technical Construction}\label{sec:technical}

Given an asynchronous program $\ap=(D,\Sigma,\Sigma_i, G,R,d_0,\mmap_0)$, we define a
``product grammar'' $G^R$ which synchronizes derivations in $G$ and $R$.
The \(\cfg\) $G^R$ simplifies some subsequent constructions on asynchronous programs.



\begin{definition}
Given a \(\cfg\) $G=(\mathcal{X},\Sigma\cup\Sigma_i,\prod)$ 
and a regular grammar $R=(D,\Sigma\cup\Sigma_i,\delta)$,
define the \(\cfg\) $G^R=(\mathcal{X}^R, \Sigma, \prod^R)$
where $\mathcal{X}^R = \set{ [dXd'] \mid d,d'\in D, X\in\mathcal{X}}$,
and $\prod^R$ is the least set such that each of the following holds:
\begin{itemize}
	\item if \( (X\rightarrow \varepsilon)\in\prod\) and \(d\in D\) then \( ([d X d]\rightarrow \varepsilon)\in\prod^R\).
	\item if \( (X\rightarrow a)\in\prod\) and \( (d\rightarrow a\cdot d')\in\delta\) then
		\( ([dXd']\rightarrow \proj_{\Sigma}(a))\in\prod^R\).
	\item if \([d_0 A d_1],[d_1 B d_2]\in\mathcal{X}^R\) and \( (X\rightarrow AB)\in\prod\)
		then \( ([d_0 X d_2]\rightarrow [d_0 Ad_1][d_1 Bd_2])\in\prod^R \).
\end{itemize}
\label{def:GR}
\end{definition}

\begin{lemma}
Let \(G\), \(R\) and \(G^R\) as in Def.~\ref{def:GR}.
For every \(d,d'\in D\), \(X\in\mathcal{X}\), \(w_1\in\Sigma^*\) and \(w\in (\Sigma\dotcup \Sigma_i)^{*}\) we have:

\begin{gather}
	[dXd']{\underset{G^R}\Rightarrow^*}w_1
\text{ implies } \exists w_2\in (\Sigma\dotcup\Sigma_i)^*
\colon \proj_{\Sigma}(w_2)=w_1
\land d{\underset{R}\Rightarrow^*}w_2\cdot d' \land
X{\underset{G}\Rightarrow^*} w_2 \\
d{\underset{R}\Rightarrow^*}w\cdot d' \land  
X{\underset{G}\Rightarrow^*} w \text{ implies }
[dXd']{\underset{G^R}\Rightarrow^*}\proj_{\Sigma}(w)\enspace .\label{eq:correctnessofconstruction}
\end{gather}
Moreover, \(G^{R}\) can be computed in time polynomial in the size of \(G\) and \(R\).
	\label{lem:gr_correctness}
\end{lemma}
\begin{proof}
See Sect.~\ref{sec:synsemproduct} for a proof of (1) and (2). 
Given def.~\ref{def:GR}, it is routine to check that the time complexity bound holds.
\end{proof}


Lem.~\ref{lem:clear} below makes clear the purpose of this section:
it gives an equivalent but simpler definition for
the semantics of an asynchronous program.

\begin{definition}
Let \(\ap=(D,\Sigma,\Sigma_i,G,R,d_0,\mmap_0)\) be an asynchronous program.  We define a
\emph{context} to be  an element of $D\times \Sigma\times D$.
We also introduce the abbreviation 
$\mathfrak{C}=D\times\Sigma\times D$ for the set of all contexts.  
Let $c=(d_i,\sigma,d_f)\in\mathfrak{C}$, define \(G^c\) 
to be an initialized \(\cfg\) which is given by \(G^R\) with the
initial symbol \([d_iX_{\sigma}d_f]\), that is
\(G^c=(\mathcal{X}^R,\Sigma,\prod^R,[d_iX_{\sigma}d_f])\).
\label{def:gc}
\end{definition}

\begin{lemma}
	Let $c=(d_1,\sigma,d_2)\in\mathfrak{C}$ and \(\mmap\in\multiset{\Sigma}\), we have:
\[
(d_1,\multi{\sigma})\overset{\sigma}\rightarrow (d_2,\mmap)
\quad\text{ if{}f }\quad \mmap\in\parikh(L(G^c))\enspace. \]
\label{lem:clear}
\end{lemma}
\begin{proof}
	The definition of \(\rightarrow\) shows that 
\begin{align*}
	& (d_1,\multi{\sigma})\overset{\sigma}\rightarrow (d_2,\mmap)\\
	\text{if{}f } &\exists w\in (\Sigma\cup\Sigma_i)^* \colon d_1 \underset{R}\Rightarrow^* w \cdot d_2 \land X_{\sigma} \underset{G}\Rightarrow^* w \land \mmap = \parikh_{\Sigma}(w) &\text{def.\ of }\rightarrow\\
	\text{if{}f } &\exists w\in (\Sigma\cup\Sigma_i)^* \colon [d_1 X_{\sigma}d_2]{\underset{G^R}\Rightarrow^*}\proj_{\Sigma}(w) \land \mmap = \parikh_{\Sigma}(w) &\text{Lem.~\ref{lem:gr_correctness}}\\
	\text{if{}f } &\exists w\in (\Sigma\cup\Sigma_i)^* \colon [d_1 X_{\sigma}d_2]{\underset{G^R}\Rightarrow^*}\proj_{\Sigma}(w) \land \mmap = \parikh\comp\proj_{\Sigma}(w) &\text{def.\ of }\parikh_{\Sigma}\\
	\text{if{}f } &\exists w'\in \Sigma^* \colon [d_1 X_{\sigma}d_2]{\underset{G^R}\Rightarrow^*} w' \land \mmap = \parikh(w') &\text{elim.\ }\proj_{\Sigma}\\
	\text{if{}f } & \mmap\in\parikh(L(G^c)) &\text{def.\ of }G^c,\parikh
\end{align*}
	
\end{proof}

Observe that this equivalent semantics completely ignores the ordering in which
handlers are posted.
Using the above constructions, we have eliminated the need to explicitly carry around the
internal actions $\Sigma_i$.
Consequently, in what follows, we shall omit the internal actions from our description of
asynchronous programs.

\subsection{Properties of Asynchronous Programs}

In this paper, we study the following decision problems for asynchronous programs.
The first set of problems relate to properties of finite runs.
\begin{definition}
	\hspace{0pt}
\begin{itemize}
	\item {\bf Safety (Global state reachability)}:\\
		{Instance:} An asynchronous program $\ap$ and a global state $d_f\in D$\\
		{Question:} Is there a reachable configuration \(c\) such that \(c.d=d_f\) ?\\
		If so $d_f$ is said to be \emph{reachable} (in \(\ap\)); otherwise \emph{unreachable}.
	\item {\bf Boundedness (of the task buffer)}:\\
		{Instance:} An asynchronous program $\ap$\\
		{Question:} Is there an $N\in\mathbb{N}$ such that
		for every reachable configuration \(c\) we have $\vert c.\mmap \vert \leq N$?\\
		If so the asynchronous program $\ap$ is \emph{bounded}; otherwise
		\emph{unbounded}.
        \item {\bf Configuration reachability}:\\
		{Instance:} An asynchronous program $\ap$ and a configuration \(c\)\\
		{Question:} Is \(c\) reachable?
\end{itemize}%
\label{def:finite-aaruns}
\end{definition}

The next set of problems relate to properties of infinite runs.

\begin{definition}
All the following problems have a common input given by an asynchronous
program $\ap$
\begin{itemize}
	\item {\bf Non Termination}:
		Is there an infinite run?
	\item {\bf Fair Non Termination}: Is there a fair infinite run?
   \item {\bf Fair Starvation}:
				 Is there a fair infinite run \(c_0,c_1,\ldots,c_i,\ldots\), a handler
				 $\sigma\in\Sigma$ and some index \(J\geq 0\) such that for each
				 \(j\geq J\) we have \linebreak (i) \(c_j.\mmap(\sigma)\geq 1\), and
				 (ii) if $c_j \overset{\sigma}\rightarrow c_{j+1}$ then
				 $c_j.\mmap(\sigma) \geq 2$?
\end{itemize}
\label{def:infinite-aaruns}
\end{definition}

We provide some intuition on the fair starvation property.
A  run could be fair, but a specific pending handler instance may never get chosen in the run.
We say that the handler instance is {\em starved} in the run.
Of course, the desired property for a program is the complement: that
no handler is starved on any run (i.e., that every infinite fair run
does not starve any handler).

\section{Petri net semantics}

In this section we show how asynchronous programs can be modelled by Petri nets.
We review a reduction from asynchronous programs to Petri nets
and sharpen the reduction to get optimal complexity bounds.

\subsection{Petri nets}

A {\em Petri net} ($\pn$ for short) $N=(S,T,F=\tuple{I,O})$ consists of a
finite non-empty set $S$ of \emph{places}, a finite set $T$ of \emph{transitions}
disjoint from $S$, and a pair $F=\tuple{I,O}$ of functions
$I\colon T\rightarrow\multiset{S}$ and $O\colon T\rightarrow\multiset{S}$.

To define the semantics of a $\pn$ we introduce the definition of
\emph{marking}.  Given a $\pn$ $N=(S,T,F)$, a marking $\mmap\in\multiset{S}$ is
a multiset which maps each $p\in S$ to a non-negative integer. 
For a marking $\mmap$,
we say that $\mmap(p)$ gives the number of \emph{tokens} contained in place $p$.

A transition $t\in T$ is \emph{enabled at} marking $\mmap$, written
$\mmap\fire{t}$, if $I(t)\preceq\mmap$. A transition $t$ that is enabled at
$\mmap$ can \emph{fire}, yielding a marking $\mmap'$ such that $\mmap'\oplus
I(t)=\mmap\oplus O(t)$. We write this fact as follows: $\mmap\fire{t}\mmap'$.


We extend enabledness and firing inductively to finite sequences of transitions
as follows.  Let $w\in T^*$.  If $w=\varepsilon$ we define
$\mmap\fire{w}\mmap'$ if{}f $\mmap'=\mmap$; else if $w=u\cdot v$ we have
$\mmap\fire{w}\mmap'$ if{}f there exists $\mmap_{1}$ such that
$\mmap\fire{u}\mmap_{1}$ and $\mmap_1\fire{v}\mmap'$.  

Let $w_{\infty}=t_0,t_1,\ldots$ be an infinite sequence of transitions.
We write \(\mmap\fire{w_{\infty}}\) if{}f there exist
markings \(\mmap_0,\mmap_1,\ldots\) such that \(\mmap_0=\mmap\) and
\(\mmap_i\fire{t_i}\mmap_{i+1}\).

An \emph{initialized} \(\pn\) is given by a pair $(N,\mmap_{\imath})$ where
$N=(S,T,F)$ is a Petri net and $\mmap_{\imath}\in\multiset{S}$ is called the
\emph{initial marking} of $N$.

A marking $\mmap$ is reachable from
$\mmap_{0}$ if{}f there exists $w\in T^*$ such that
$\mmap_{0}\fire{w}\mmap$.
The \emph{set of reachable states from $\mmap_{0}$}, 
written $\fire{\mmap_{0}}$, is thus
$\set{\mmap\mid\exists w\in T^*\colon \mmap_{0}\fire{w}\mmap}$.
When the starting marking is omitted, it is assumed to be $\mmap_{\imath}$.

We now define the size of the encoding of a \(\pn\) and of
their markings.  First, let us recall the encoding of a multiset
$\mmap\in\multiset{S}$.  It is encoded as a list of pairs $(p,\mmap(p))$
symbol/value for each symbol $p\in S$.
The size of the encoding, noted $\tmenc{\mmap}$, is given by the number of bits
needed to write down the list of pairs, where we assume $\mmap(p)$ is encoded
in binary.  The encoding of a $\pn$ $N$ is given by a list of lists.  Each
transition $t\in T$ is encoded by two lists corresponding to $I(t)$ and $O(t)$.
The size of $N$, written $\tmenc{N}$, is thus defined as $\sum_{t\in
T}\tmenc{I(t)}+\sum_{t\in T}\tmenc{O(t)}$.

We now define the boundedness, the reachability and the coverability problem
for Petri nets.  Let $(N,\mmap_{\imath})$ be a initialized $\pn$.  The
\emph{boundedness problem} asks if $\fire{\mmap_{\imath}}$ is finite set. Let
$\mmap\in\multiset{S}$, the \emph{reachability problem} (resp.
\emph{coverability problem}) asks if $\mmap\in\fire{\mmap_{\imath}}$ (resp.
$\ucl{\mmap}\cap\fire{\mmap_{\imath}}\neq\emptyset$) and if so $\mmap$ is said
to be \emph{reachable} (resp. \emph{coverable}).  In each of the above problem,
the \emph{size} of an instance is given by the
\(\tmenc{N}+\tmenc{\mmap_{\imath}}\) plus \(\tmenc{\mmap}\), if
any.

A marking $\mmap$ is Boolean if for each place $p \in S$, we have $\mmap(p) \in\set{0,1}$.
An initialized Petri net is {\em Boolean} if $\mmap_{\imath}$ is Boolean and  
for each $t\in T$, both $I(t)$ and $O(t)$ are Boolean.
The following technical lemma shows that for any initalized Petri net, one
can compute in polynomial time a
Boolean initialized Petri net that is equivalent w.r.t.\ 
the boundedness problem (i.e., the original Petri net is bounded iff the Boolean Petri net is).
Similarly, for an initialized Petri net and a marking, one can compute a Boolean
initialized Petri net and a Boolean marking that is equivalent
w.r.t.\ the coverability and reachability problems.

\begin{lemma}
	(1) Let \((N,\mmap_{\imath})\) be an initialized \(\pn\).  There exists
        a Boolean initialized \(\pn\) 
	\((N',\mmap'_{\imath})\) computable in polynomial time in the size of $(N,\mmap_{\imath})$
        such that $(N,\mmap_{\imath})$ is bounded
        iff $(N',\mmap'_{\imath})$ is bounded. 

	(2) Let \((N,\mmap_{\imath},\mmap_f)\) be an instance of the
	reachability (respectively, coverability) problem.  There exists a Boolean initialized Petri net 
	\((N',\mmap'_{\imath})\) and a Boolean marking \(\mmap'_f\) computable in polynomial
	time such that $\mmap_f$ is reachable (respectively, coverable) in \((N,\mmap_{\imath})\) iff 
        $\mmap'_f$ is reachable (respectively, coverable) in \((N',\mmap'_{\imath})\).
	\label{lem:binaryisnotaproblem}
\end{lemma}
Lem.~\ref{lem:binaryisnotaproblem} which proof is in the appendix shows that
lower bounds for Petri nets already hold for Boolean Petri nets. This will be
useful in the next sections to get lower bounds on asynchronous programs.

The following results are known from the $\pn$ literature.

\begin{theorem}
\hspace{0pt}
\begin{enumerate}
\item \cite{Rackoff78}
	The boundedness and coverability problems for $\pn$ are
        {\sc expspace}-complete.
\item \cite{kosaraju82,Lipton} The reachability problem for $\pn$ is decidable and {\sc
	expspace}-hard.
	\label{prop-cpltxy}
\end{enumerate}
\end{theorem}

While the best known lower bound for Petri net reachability is {\sc
expspace}-hard, the best known upper bounds take non-primitive recursive space \cite{kosaraju82,Meyer,Mayr81,Lambert92}.  Moreover,
Lem.~\ref{lem:binaryisnotaproblem} shows that the lower bounds hold already for
Boolean Petri nets.

\subsection{Petri net semantics of asynchronous programs}

We now show how to model an asynchronous program \(\ap=(D,\Sigma,G,R,d_0,\mmap_0)\) as an initialized
\(\pn\) \((N_{\ap},\mmap_{\imath})\), parameterized by a family
of {\em widgets} $\pnfamily^\clubsuit = \set{N^\clubsuit_c \mid c \in D\times\Sigma\times D}$. 
Each widget $N^\clubsuit_{(d,a,d')}$ 
is a Petri net, intuitively capturing the effect of 
executing a handler $a$ taking the system from global state $d$ 
to global state $d'$. 

Fix an asynchronous program $\ap = (D,\Sigma, G, R, d_0,\mmap_0)$.
Let $\pnfamily^\clubsuit = \set{N^\clubsuit_c \mid c \in \mathfrak{C}}$
be a family of Petri nets, called widgets, one for each context in $\mathfrak{C}$.
We say that the family $\pnfamily^\clubsuit$ is {\em adequate} if the following conditions
hold.
For each $c = (d_1,a,d_2)\in\mathfrak{C}$, the widget 
\(N^{\clubsuit}_c=(S^{\clubsuit}_c,T^{\clubsuit}_c,F^{\clubsuit}_c)\) is a 
\(\pn\) with a distinguished \emph{entry} place
\((\mathit{begin},c)\in S^{\clubsuit}_c\) and a distinct \emph{exit} place
\((\mathit{end},c)\in S^{\clubsuit}_c\). 
Moreover for every \(\mmap\in\multiset{\Sigma}\) we have:
	\begin{equation}
		\exists w\in (T^{\clubsuit}_c)^*\colon
		\multi{(\mathit{begin},c)}\fire{w}(\multi{(\mathit{end},c)}\oplus\mmap)
		\text{ if{}f }(d_1,\multi{a})\overset{a}\rightarrow (d_2,\mmap)\enspace .
		\label{eq:correctmodel}
	\end{equation}

Construction~\ref{constr:pn} below shows how an adequate family of widgets
is ``stitched together'' to give a Petri net model for an asynchronous program.

\begin{construction}
Let \(\ap=(D,\Sigma,G,R,d_0,\mmap_0)\) be an asynchronous program and $\pnfamily^{\clubsuit}$
an adequate family of widgets for $\ap$.  
Define \(
(N_{\ap}(\pnfamily^{\clubsuit}),\mmap_{\imath}) \) to be an initialized \(\pn\) where (1)
\(N_{\ap}(\pnfamily^{\clubsuit})=(S_{\ap},T_{\ap},F_{\ap})\)
is given as follows:
		\begin{itemize}
			\item the set $S_{\ap}$ of places is given by $ D\cup \Sigma\cup\bigcup_{c\in\mathfrak{C}} S_c^{\clubsuit}$;
			\item the set $T_{\ap}$ of transitions is given by
				\( \bigcup_{c\in\mathfrak{C}} \bigl(\set{t_c^{<}}\cup T_c^{\clubsuit}\cup \set{t_c^{>}}\bigr)\);
			\item $F_{\ap}$ is such that for each $c=(d_1,a,d_2)\in\mathfrak{C}$ we have
				\begin{align*}
					F_{\ap}(t_c^{<})&=\tuple{\multi{d_1, a},\multi{(\mathit{begin},c)}}\\
					F_{\ap}(t)&=F_{c}^{\clubsuit}(t) & t\in T^{\clubsuit}_c\\
					F_{\ap}(t_c^{>})&=\tuple{\multi{(\mathit{end},c)}, \multi{d_2}}
				\end{align*}
		\end{itemize}
		and (2) $\mmap_{\imath}=\multi{d_0}\oplus \mmap_0$.
		\label{constr:pn}
\end{construction}

In what follows we use the notation \(N_{\ap}\) to denote
\(N_{\ap}(\pnfamily^{\clubsuit})\), which is parameterized by an 
adequate family \(\pnfamily^{\clubsuit}\). 

\begin{comment}
Let us now define a family of widgets and show it satisfies all the
requirements of Constr.~\ref{constr:pn}.
%
\begin{definition}\label{def:conceptualcontext}
	Let $c=(d_1,a,d_2)\in\mathfrak{C}$, the \(\pn\) $N_c^{\infty}=(S_c^{\infty},T_c^{\infty},F_c^{\infty})$ is given by:
	\begin{itemize}
		\item the set $S_c^{\infty}$ of places is given by
			$\Sigma\cup\set{(\mathit{begin},c),(\mathit{end},c)}$;
		\item \(T_c^{\infty}= \set{t_{\mmap} \mid \mmap\in\parikh(L(G^c)) } \); and
		\item Let $t_{\mmap}\in T_c^{\infty}$,
	\(F_c^{\infty}(t_{\mmap})=\tuple{\multi{(\mathit{begin},c)},\multi{(\mathit{end},c)}\oplus\mmap}\).
	\end{itemize}
	Finally, define \(\pnfamily^{\infty}=\set{N_c^{\infty}}_{c\in\mathfrak{C}}\).
\end{definition}

The definition of \(\pnfamily^{\infty}=\set{N^{\infty}_c}_{c\in\mathfrak{C}}\)
is such that each \(N_c^{\infty}\) trivially satisfies condition
\eqref{eq:correctmodel}. In fact by setting \(w=t_{\mmap}\) and using the
equivalence of Lem.~\ref{lem:clear}, condition \eqref{eq:correctmodel} holds.

Unfortunately, for some \(c\in\mathfrak{C}\) the \(\pn\) \(N_c^{\infty}\) may
have infinitely many transitions.  
Therefore the family 
\((N_{\ap}(\pnfamily^{\infty}),\mmap_{\imath})\), 
while a semantic characterization of the transition system of an $\ap$,
cannot be used algorithmically. 
\end{comment}

We show two constructions of adequate families.
First, we recall a simple definition of an adequate family of widgets,
inspired by a similar construction in \cite{SenV06}, that
leads to a Petri net $N_{\ap}$ which is exponential in the size of $\ap$.
Next, we give a new construction of an adequate family of widgets
that leads to a Petri net $N_{\ap}$ of
size {\em polynomial} in $\ap$. 
As we shall see later our
definition allows to infer the existence of optimal {\sc expspace} algorithms
for checking safety and boundedness properties.

\smallskip
\noindent
{\bf First construction of an adequate family.}
Let us now define the widgets
\(\pnfamily^{\star}=\set{N^{\star}_c}_{c\in\mathfrak{C}}\) using ideas from \cite{SenV06}. 
The central idea is to rely on the effective construction 
of Lem.~\ref{lem:parikh} which, given an initialized \(\cfg\) $G$, returns an
initialized regular grammar $A$ such that the languages $L(G)$ and $L(A)$ are
Parikh-equivalent.

\begin{definition}
Let $c=(d_1,a,d_2)\in\mathfrak{C}$. Let 
$A^c=(Q^c,\Sigma,\delta^c,q_0)$ be a regular grammar such that 
$\parikh(L(G^c)) = \parikh(L(A^c))$.
Define the Petri net $N^{\star}_c=(S^{\star}_c,T^{\star}_c,F^{\star}_c)$ given as follows:
\begin{itemize}
	\item the set $S^{\star}_c$ of places is given by $\set{(\mathit{begin},c),(\mathit{end},c)}\cup Q^c\cup\Sigma$;
        \item $T^{\star}_c = \delta^c \cup \set{t_i}$;
	\item the sets $F^{\star}_c$ are such that
		\begin{align*}
			F^{\star}_c(t_i)= & \tuple{\multi{(\mathit{begin},c)},\multi{q_0}} &\\
			F^{\star}_c(q\rightarrow \sigma\cdot q')= & \tuple{\multi{q},\multi{q',\sigma}}& \\
			F^{\star}_c(q \rightarrow \varepsilon)= & \tuple{\multi{q},\multi{(\mathit{end},c)}}
		\end{align*}
\end{itemize}
	Finally, define \(\pnfamily^{\star}=\set{N_c^{\star}}_{c\in\mathfrak{C}}\).
\label{def:parikh_construction}
\end{definition}

An invariant of $N^{\star}_c$ is that every reachable marking from
\(\multi{(\mathit{begin},c)}\) is such that the tokens in places \(Q^c\) never
exceed \(1\).

Lem.~\ref{lem:construction} shows that $\pnfamily^{\star}$ is an adequate family.
%
\begin{lemma}
	Let $c=(d_1,a,d_2)\in\mathfrak{C}$, and $\mmap\in\multiset{\Sigma}$ all the
	following statements are equivalent:
	\begin{enumerate}
		\item \( (d_1,\multi{a})\overset{a}\rightarrow(d_2,\mmap)\);%
		\item \(\mmap\in\parikh(L(G^c))\);%
		\item \(\mmap\in\parikh(L(A^c))\);%
		\item \(\exists w\in (T^{\star}_c)^* \colon
			\multi{(\mathit{begin},c)}\fire{w}_{N^{\star}_c}\bigl(\mmap\oplus\multi{(\mathit{end},c)}\bigr)\).
	\end{enumerate}
		\label{lem:construction}%
\end{lemma}
\begin{proof}
	(1) and (2) are equivalent by Lem.~\ref{lem:clear}. 
	(2) and (3) are equivalent by assumption on $A^c$. 
        Finally, (3) and (4) are equivalent
	by Def.~\ref{def:parikh_construction}.
\end{proof}

Note that for some \(c\in\mathfrak{C}\) the set \(S^{\star}_c\)
of places in \(N^{\star}_c\)  may be exponentially larger than the set
\(\mathcal{X}^c\) of variables of $G^c$.
As an example consider the following \(\cfg\)
\(G=(\set{A_n,\dots,A_0},\set{a},\prod, A_n)\) for some \(n\geq 0\) where
\(\prod=\set{A_k \rightarrow A_{k-1} A_{k-1} \colon 1\leq k\leq n}\cup\set{A_0
\rightarrow a}\). Clearly \(L(G)=\set{a^{2^n}}\) and therefore there is no
regular grammar with less than \(2^n\) variables which accepts the same
language.

\smallskip
\noindent
{\bf Second construction of an adequate family.}
We now define a new family 
\(\pnfamily=\set{N_{c}}_{c\in\mathfrak{C}}\)
of widgets 
which improves on \(\pnfamily^{\star}\) by providing more compact widgets,
in particular, widgets polynomial in the size of $G$.
Given  a context $c=(d_1,a,d_2)\in\mathfrak{C}$ and the
associated initialized \(\cfg\) $G^c=(\mathcal{X}^c,\Sigma,\prod^c,[d_1X_ad_2])$, 
the widget $N_{c}=(S_c,T_c,F_c)$ will be such that
$|S_c|=\mathcal{O}(|\mathcal{X}^c|)$ and $|T_c|=\mathcal{O}(|\prod^c|)$. 

Our construction combines two ingredients.

The first ingredient is the following construction of \cite{Esp97} which, given
an initialized \(\cfg\) \(G=(\mathcal{X},\Sigma,\prod,S)\), returns an
initialized \(\pn\) \( (N_G,\mmap_{\imath}) \) where (1) \(N_G=(S_G,T_G,F_G)\)
is given by
\begin{itemize}
	\item \(S_G=\mathcal{X}\cup\Sigma\) and \(T_G = \prod\);
	\item \(F_G(X\rightarrow\alpha)=\tuple{\multi{X},\parikh(\alpha)}\); 
\end{itemize}
and (2) \(\mmap_{\imath}=\multi{S}\).
Let \(\mathfrak{S}\) be the set of transition sequences that are enabled in \(\mmap_{\imath}\).
We conclude from \cite{Esp97} that there is a total surjective function \(f\)
from the set of derivations of \(G\) onto \(\mathfrak{S}\) such that for every
\(\alpha\in(\mathcal{X}\cup\Sigma)^*\) if \(S\underset{G}\Rightarrow^* \alpha\)
then \(\mmap_{\imath}\fire{f(S\Rightarrow^* \alpha)}\parikh(\alpha)\).

Unfortunately, the above construction cannot be used directly to build an adequate
family because of the following problem. 
Recall that for each \(c\in\mathfrak{C}\) the widget
\(N_c=(S_c,T_c,F_c)\) has an exit place \( (\mathit{end}, c)\in S_c\) and
 condition \eqref{eq:correctmodel} must hold.
Using Lem.~\ref{lem:clear} we
obtain that \eqref{eq:correctmodel} is equivalent to: 
\[\exists w\in (T_c)^* \colon \multi{(\mathit{begin},c)}\fire{w}(\multi{(\mathit{end},c)}\oplus\mmap) \text{ if{}f }\mmap\in \parikh(L(G^{c}))\enspace .
 \] 
 This means that widget \(N_c\) should put a token in \( (\mathit{end},c)\) only
 after some \(\mmap \in \parikh(L(G^{c}))\) has been generated, that is,
 it should check that the derivation 
 \(S\underset{G}\Rightarrow^* \alpha\) it is simulating cannot be further
 extended, i.e., \(\alpha\in\Sigma^*\).
 This is equivalent to checking that each place which
 corresponds to a variable in \(\mathcal{X}^{c}\) is empty. However the
 definition of \(\pn\) transitions do not allow for such a test for 0.
 Therefore we need an additional ingredient in the widget in order to ensure
 that \(N_c\) puts a token in \( (\mathit{end},c)\) only after some \(\mmap \in
 \parikh(L(G^{c}))\) has been generated.

The second ingredient in our construction is the observation from \cite{EKL10:JACM,egkl11-ipl}
that as long as we are interested in the Parikh image of
a context-free language, it suffices to only consider derivations of 
{\em bounded index}. 
Let us first introduce a few notions on derivations of \(\cfg\).
	Let \(G=(\mathcal{X},\Sigma,\prod,S)\) be an initialized \(\cfg\).
	Given a word \(w\in(\Sigma\cup\mathcal{X})^*\), we denote by \(\#_{\mathcal{X}}(w)\) the
	number of symbols of \(w\) that belongs to \(\mathcal{X}\).
	Formally, \(\#_{\mathcal{X}}(w)= |\parikh_{\mathcal{X}}(w)|\).
	A derivation $S=\alpha_0 \Rightarrow \cdots \Rightarrow \alpha_m$ of \(G\)
	has index $k$ if $\#_{\mathcal{X}}(\alpha_i)\leq k$ for each
	\(i\in\set{0,\ldots,m}\).  The set of words of \(\Sigma^*\) derivable through
	derivations of index $k$ is denoted by $L^{(k)}(G)$.

\begin{lemma}{(from \cite{egkl11-ipl})}
		Let \(G=(\mathcal{X},\Sigma,\prod,S)\) be an initialized \(\cfg\), and
		let \(k=\card{\mathcal{X}}\), we have:
		\(\parikh(L(G))=\parikh(L^{(k+1)}(G))\).%
	\label{lem:ipl}%
\end{lemma}%

Our next widget construction is directly based on this result.
In the following, our widget definition only differs from
the construction of \cite{Esp97} 
by our use of an incidental budget place \(\$\).

In the construction of \((N_G,\mmap_{\imath})\) above,
define \(\mathfrak{s}\) to be the subset of \(\mathfrak{S}\) such that every marking
reachable through a sequence in \(\mathfrak{s}\) has no more than \(k\) tokens
in the places \(\mathcal{X}\). It is routine to check that \(\set{f^{-1}(w)\mid
w\in\mathfrak{s}}\) corresponds the set of derivations of index \(k\).

Let us define \(N^k_G\) which adds an extra place \(\$\) to \(N_G\) in order
to allow exactly the sequences of transitions of \(\mathfrak{s}\). 

\begin{definition}
	Let \(G=(\mathcal{X},\Sigma,\prod,S)\) be an initialized \(\cfg\) and let
	\(k>0\), we define \( (N^k_G,\mmap_{\imath}) \) to be an initialized \(\pn\)
	where (1) \(N^k_G=(S_G,T_G,F_G)\) is given by
\begin{itemize}
	\item \(S_G=\mathcal{X}\cup\Sigma\cup\set{\$}\);
        \item \(T_G = \prod\); and
	\item \(F_G\) is such that \\
                \begin{align*}
		F_G(X \rightarrow Z\cdot Y)= \tuple{\multi{X,\$},\multi{Z,Y}}\quad\quad\mbox{and}\quad\quad
		F_G(X \rightarrow \sigma)= \tuple{\multi{X},\parikh(\sigma)\oplus\multi{\$}} 
                \end{align*}
\end{itemize}
and (2) \(\mmap_{\imath}=\multi{S}\oplus\multi{\$^{k-1}}\).
\label{def:indexboundedfiringsequence}
\end{definition}

The set of enabled transition sequences of \(N^k_G\) coincides with the set of
derivation of index \(k\).  In fact, every reachable marking has exactly \(k\)
tokens in places \(\mathcal{X}\cup\set{\$}\). Therefore no reachable marking
puts more than \(k\) tokens in places \(\mathcal{X}\) which coincides with the
condition imposed on derivations of index \(k\).

\begin{lemma}
	Let \(G=(\mathcal{X},\Sigma,\prod,S)\) be an initialized \(\cfg\), let $k>0$,
	and let \( (N^k_G=(S_G,T_G,F_G),\mmap_{\imath}) \). For every \(\mmap\in\multiset{S_G}\)
	\[ (\mmap\oplus\multi{\$^k})\in\fire{\mmap_{\imath}}_{N^k_G} \text{ if{}f } \mmap\in\parikh(L^{(k)}(G))\enspace .\]
	\label{lem:smartencoding}
\end{lemma}
\begin{proof}
	We prove that for every \(\alpha_1,\alpha_2\in(\mathcal{X}\cup\Sigma)^*\) where
	both \(\#_{\mathcal{X}}(\alpha_1)\leq k\) 
	and \(\#_{\mathcal{X}}(\alpha_2)\leq k\),
	we have \(\alpha_1{\Rightarrow} \alpha_2\) if{}f there exists \(t\in T_G\)
	such that
	\(\parikh(\alpha_{1})\oplus\multi{\$^{k-\#_{\mathcal{X}}(\alpha_1)}}\fire{t}_{N_G^{k}}\parikh(\alpha_{2})\oplus\multi{\$^{k-\#_{\mathcal{X}}(\alpha_2)}}\).
	This holds by definition of \(F_G\).
	Also observe that \(\mmap_{\imath}=\parikh(S)\oplus\multi{\$^{k-\#_{\mathcal{X}}(S)}}\).

	Now note that the right hand side is equivalent to saying that there exist
	\(\alpha_1,\ldots,\alpha_{n+1}\in (\mathcal{X}\cup\Sigma)^*\) where each
	\(\alpha_i\) is such that \(\#_{\mathcal{X}}(\alpha_i)\leq k\), $S = \alpha_1$, 
        $\alpha_{n+1} \in\Sigma^*$, \(
	\alpha_1{\Rightarrow} \alpha_2 \cdots
	\alpha_{n}{\Rightarrow} \alpha_{n+1} \) and \( \mmap=\parikh(\alpha_{n+1})\), and
        use induction on $n$.
\end{proof}

Let us now turn to our widget definition which directly relies on the above results.
\begin{definition}
Let $c=(d_1,a,d_2)\in\mathfrak{C}$, and let
$G^c=(\mathcal{X}^c,\Sigma,\prod^c,[d_1X_ad_2])$ its associated initialized \(\cfg\).
Define $k=|\mathcal{X}^c|$ and $N_c=(S_c,T_c,F_c)$ such that:
\begin{itemize}
	\item the set $S_c$ of places is given by
		$\set{(\mathit{begin},c),(\mathit{end},c)}\cup \mathcal{X}^c\cup \set{(\$,c)}
		\cup\Sigma$;
        \item $T_c = \set{t_i, t_e} \cup \prod^c$; and
	\item the set $F_c$ is such that
        \begin{align*}
	F_c(t_i)= & \tuple{\multi{(\mathit{begin},c)},\multi{[d_1X_ad_2]}\oplus\multi{(\$,c)^{k}}}\\
	F_c(X \rightarrow Z\cdot Y)= & \tuple{\multi{X,(\$,c)},\multi{Z,Y}} \\
	F_c(X \rightarrow \sigma)= & \tuple{\multi{X},\parikh(\sigma)\oplus\multi{(\$,c)}}\\ 
	F_c(t_e)= & \tuple{\multi{(\$,c)^{k+1}},\multi{(\mathit{end},c)}}
        \end{align*}
\end{itemize}
Define \(\pnfamily=\set{N_c}_{c\in\mathfrak{C}}\).
\label{def:the_construction}
\end{definition}

The following lemma shows that the family constructed above is adequate.

\begin{lemma}
	Let $c=(d_1,a,d_2)\in\mathfrak{C}$, and $\mmap\in\multiset{\Sigma}$ we have:
	\[(d_1,\multi{a})\overset{a}\rightarrow(d_2,\mmap)\quad\text{
	if{}f }\quad \exists w\in (T_c)^* \colon
	\multi{(\mathit{begin},c)}\fire{w}_{N_c}\bigl(\multi{(\mathit{end},c)}\oplus\mmap\bigr)\enspace
	.\]
	Moreover, \(N_c\) is computable in time polynomial in the size of \(G^c\).
	\label{lem:construction_final}
	\end{lemma}
	\begin{proof}
		Lem.~\ref{lem:clear} shows that the left hand side of the
		equivalence can be replaced by \(\mmap\in\parikh(L(G^c))\).  Moreover,
		Lem.~\ref{lem:ipl} shows that \(L^{(k+1)}(G^c)\) and \(L(G^c)\) are
		Parikh-equivalent, hence that the left hand side of the equivalence can be
		replaced by \(\mmap\in\parikh(L^{(k+1)}({G^c}))\). Finally,
		we conclude from Lem.~\ref{lem:smartencoding} and Def.~\ref{def:the_construction} that
		\(\mmap\in\parikh(L^{(k+1)}({G^c}))\) if{}f \(\exists w\in (T_c)^* \colon
		\multi{(\mathit{begin},c)}\fire{w}_{N_c}\bigl(\multi{(\mathit{end},c)}\oplus\mmap\bigr)\)
		and we are done. Fianlly, given def.~\ref{def:the_construction}, it is
		routine to check that the polynomial time upper bound holds.
	\end{proof}


\section{Model Checking}

\subsection{Safety and Boundedness}\label{sec:safetyandboundedness}

In this section, we provide algorithms for checking safety (global state
reachability), boundedness, and configuration reachability for asynchronous
programs by reduction to equivalent problems on \(\pn\). Conversely, we show
that any \(\pn\) can be simulated by an asynchronous program with no recursion.

\begin{lemma}
	Let $\ap$ be an asynchronous program and let \((N_{\ap},\mmap_{\imath})\) be
	an initialized $\pn$ as given in Constr.~\ref{constr:pn}. We have
	\begin{itemize}
		\item $(N_{\ap},\mmap_{\imath})$ is bounded if{}f $\ap$ is bounded.
		\item \((d,\mmap)\) is reachable in $\ap$ if{}f
			$\multi{d}\oplus\mmap$ is reachable in $N_{\ap}$ from \(\mmap_{\imath}\).
	\end{itemize}
	Moreover, \( (N_{\ap}(\pnfamily),\mmap_{\imath}) \) can be computed in
	polynomial time from \(\ap\).
	\label{lem:equivpn}
\end{lemma}
\begin{proof}
	The results for boundedness and reachability essentially follows from requirement
	\eqref{eq:correctmodel} in the definition of adequacy. 
	For the polynomial time algorithm  we first show that polynomial time is
	sufficient to compute \(\pnfamily\) given \(\ap\). In fact, the polynomial
	time upper bounds follows from the fact that
	\(\pnfamily=\set{N_c}_{c\in\mathfrak{C}}\) contains polynomially many
	widgets, that each \(N_c\) is computable in polynomial time given \(G^c\)
	(Lem.~\ref{lem:equivpn}) and that each \(G^c\) is computable in polynomial
	time given \(\ap\) (basically \(G\), \(R\)) and \(c\) (Lem.~\ref{lem:gr_correctness}).
	Then given \(\pnfamily\), it is routine to check from Constr.~\ref{constr:pn} that 
	\((N_{\ap}(\pnfamily),\mmap_{\imath}) \) can be computed in polynomial time from \(\ap\).
\end{proof}

Let us now consider the boundedness, the safety and the configuration
reachability problems for asynchronous programs. 
Lem.~\ref{lem:equivpn} shows that for the boundedness, the safety and 
the configuration reachability problem
for asynchronous programs there is an equivalent instance of, respectively, the
boundedness, the coverability and the reachability problem for \(\pn\).
Moreover each of the reduction can be carried out in polynomial time.  In
\cite{Rackoff78} Rackoff gives {\sc expspace} algorithms to solve the
coverability and boundedness problem for \(\pn\), therefore we obtain an
exponential space upper bound for the safety and boundedness problems for
asynchronous programs. For the reachability problem, the best known upper bounds
take non-primitive recursive space \cite{EN94}.

\begin{figure*}[h]
  \centering
  \hspace{\stretch{1}}
    \begin{minipage}{.5\textwidth}
\begin{tabbing}
\ \ \ \ \=\ \ \ \=\ \ \ \=\\
{\it global} {\tt st} = $(\varepsilon,\varepsilon)$;\\
\\
{\tt runPN} () \{\\
\>	if {\tt st} \(\in (T\cup\set{\varepsilon})\times\set{\varepsilon}\) \{\\
\>\>    pick $t\in T$ non det.;\\
\>\>    {\tt st} {\tt =} $(t,\hat{I}(t))$;\\
\>  \}\\
\>  {\tt post} {\tt runPN}();\\
\}\\
\\
Initially: $\mmap_{\imath} \oplus\multi{\mathtt{runPN}}$
\end{tabbing}
    \end{minipage}%
  \hspace{\stretch{1}}
    \begin{minipage}{.45\textwidth}
\begin{tabbing}
\ \ \ \ \=\ \ \ \=\ \ \ \=\ \ \ \=\ \ \ \=\\
$p$() \{\\
\>  if {\tt st} {\tt ==} $(t, p\cdot w)$ \{\\
\>\>    {\tt st} {\tt =} $(t, w)$;\\
\>\>  if $w$ {\tt ==} $\varepsilon$ \{\\
\>\>\>      for each $p\in S$ do \{ \\
\>\>\>\>       if \(O(t)(p)>0\) \{\\
\>\>\>\>\>        {\tt post} $p$();\\
\>\>\>\>       \}\\
\>\>\>       \}\\
\>\>    \}\\
\>  \} else \{\\
\>\>    {\tt post} $p$();\\
\>  \}\\
\}
\end{tabbing}
    \end{minipage}%
  \hspace{\stretch{1}}
	\caption{Let \( (N=(S,T,F=\tuple{I,O}),\mmap_{\imath})\) be an initialized
	Boolean
	\(\pn\).  We assume that \(N\) is such that \(\forall t\in
	T\colon\card{I(t)}>0\). The encoding of \(N\) is given by an asynchronous
	program with \(|S|+1\) handlers.  }
  \label{fig-pnbound2aabound}
\end{figure*}

We now give the reverse reductions in order to derive lower complexity
bounds.  
In fact, we show how to reduce instances of the boundedness, the
coverability and the reachability problem for Boolean \(\pn\) into equivalent instances
of, respectively, the boundedness, the safety and the configuration
reachability problem for asynchronous programs.  
Each of those reduction is carried out in polynomial time in the size of the given instance.  
From known {\sc expspace} lower bounds
for Petri nets, and the construction in Lem.~\ref{lem:binaryisnotaproblem},
we get {\sc expspace} lower bounds for the boundedness, the safety, and
the configuration reachability problems for asynchronous programs.

Fix a Boolean initialized $\pn$ $(N,\mmap_{\imath})$.  
The encoding of a \(\pn\) as an asynchronous program given in Fig.~\ref{fig-pnbound2aabound} 
is the main ingredient of our reductions. 

For readability, we describe the asynchronous program in pseudocode syntax. It
is easy to convert the pseudocode to a formal asynchronous program.

Let us fix an (arbitrary) linear ordering $<$ on the places in $S$.
For each $t\in T$, let $\hat{I}(t)$ be the sequence obtained by ordering the set $I(t)$
according to the ordering $<$ on $S$, and let $\mathit{suffix}(\hat{I}(t))$ be the set
of suffixes of $\hat{I}(t)$.
Clearly, for any $t\in T$, there are at most $\card{S} + 1$ elements in $\mathit{suffix}(\hat{I}(t))$.

The intuition behind the construction of Fig.~\ref{fig-pnbound2aabound} is the
following.  The asynchronous program has $\card{S} + 1$ procedures, one
procedure $p$ for each place $p \in S$, and a procedure $\mathtt{runPN}$ that
simulates the transitions of $\pn$.  The content of the task buffer (roughly)
corresponds to the marking of the Petri net.

The procedure ${\tt runPN}$ initiates the simulation of $\pn$ by selecting nondeterministically
a transition to be fired.
A global variable ${\tt st}$ keeps track of the
transition $t$ selected and also the preconditions that have yet to be checked
in order for $t$ to be enabled.
The possible values for ${\tt st}$ are
$(\varepsilon,\varepsilon)$ (which holds initially),
and $(t, w)$ for transition $t\in T$ and $w\in\mathit{suffix}(\hat{I}(t))$
(encoding the fact that the current transition being simulated is $t$, and we need to
reduce the number of pending instances of each $p\in w$ by one in order to fire $t$).
Thus, the maximum number of possible values to ${\tt st}$ is $\card{T}\cdot (\card{S} + 1) + 1$.

The code for ${\tt runPN}$ works as follows.  If ${\tt
st}\in (T\cup\set{\varepsilon})\times\set{\varepsilon}$, it nondeterministically selects an arbitrary
transition $t$ of the \(\pn\) (not necessarily an enabled transition) to be
fired, sets ${\tt st}$ to $(t, \hat{I}(t))$, and reposts itself.  If ${\tt
st}\notin (T\cup\set{\varepsilon})\times\set{\varepsilon}$, it simply reposts itself.

We now describe how a transition is fired based on the global state ${\tt st}$.
When ${\tt runPN}$ sets ${\tt st}$ to $(t,\hat{I}(t))$, it means
that we must consume a token from each place in $I(t)$ in order to fire $t$.
Then the intuition is the following.
Each time a handler $p$ is dispatched it will check if it is the first
element in the precondition, i.e., if ${\tt st} = (t, p\cdot w)$ for some $w$.
If $p$ is not the first element in the precondition, it simply reposts itself,
so that the number of pending instances to each $p'\in S$ before and after the
dispatch of $p$ are equal.  However, if ${\tt st} = (t, p\cdot w)$, there are
two possibilities.  If $w\not = \epsilon$, then handler $p$ updates ${\tt st}$
to $(t, w)$, but does not repost itself. This ensures that after the execution
of $p$, the number of pending instances to $p$ is one fewer than before the
execution of $p$ (and thus, we make progress in firing the transition $t$ by
consuming a token from its precondition).  If $w = \epsilon$, then
additionally, handler $p$ posts $p'$ for each $p' \in O(t)$.  This ensures that
the execution of the transition $t$ is complete, and moreover, each place in
$O(t)$ now has a pending handler instance corresponding to the firing of $t$.

The initial task buffer is the multiset $\mmap_{\imath}\oplus\multi{\mathtt{runPN}}$ and the
initial value of ${\tt st}$ is $(\varepsilon,\varepsilon)$.

The following invariant is preserved by the program of
Fig.~\ref{fig-pnbound2aabound}, whenever \({\tt st}=(\mathit{tr},\varepsilon)\)
for \(\mathit{tr}\in T\cup\set{\varepsilon}\) we have that the multiset $\mmap$
given by the number of pending instances to procedure \(p\) for each $p\in S$
is such that $\mmap_{\imath}\fire{w\cdot \mathit{tr}}\mmap$ for $w\in T^*$.

Let us prove the invariant. Initially, we have ${\tt st} =
(\varepsilon,\varepsilon)$ and the task buffer is precisely
$\mmap_{\imath}$, so the invariant holds because we have
$\mmap_{\imath}\fire{\varepsilon}\mmap_{\imath}$.

By induction hypothesis, assume the invariant holds at some configuration of
the program in which \({\tt st}\in
(T\cup\set{\varepsilon})\times\set{\varepsilon}\). We show the invariant holds
the next time \({\tt st}\in(T\cup\set{\varepsilon})\times\set{\varepsilon}\).

Whenever ${\tt st}$ is of the form $(\mathit{tr},\varepsilon)$, each dispatch
to $p$ for $p\in S$ simply reposts itself.  When procedure ${\tt runPN}$ is
dispatched, it picks a transition $t$ to be fired.
Hence ${\tt st}$ is updated $(t,\hat{I}(t))$.
Suppose $\mmap\fire{t}$.
Then, for each $p\in I(t)$, the program configuration has a pending instance of $p$.
A sequence of dispatches corresponding to $\hat{I}(t)$  will
reduce ${\tt st}$ to $(t,p)$ for some $p\in S$, and at this point, the dispatch of $p$
will post as many calls as $O(t)$.
The configuration reached after this dispatch of $p$ sets ${\tt
st}=(t,\varepsilon)$ and the configuration of the program corresponds to a
marking $\mmap'\oplus I(t) = \mmap \oplus O(t)$.

Now suppose $t$ is not enabled in $\mmap$. Then in the simulation, ${\tt st}$ will
get to some value $(t, p\cdot w)$ such that there is no pending instance to $p$.
In this case, the state ${\tt st}$ will never be set to some value in
\((T\cup\set{\varepsilon})\times\set{\varepsilon}\),
and hence the invariant holds vacuously.

We conclude by establishing the {\sc expspace} lower bounds for boundedness,
safety and configuration reachability. 

\noindent %
{\bf Boundedness.}
Consider the reduction given at Fig.~\ref{fig-pnbound2aabound} which given an
initialized Boolean \(\pn\) \( (N,\mmap_{\imath})\) builds an asynchronous program \(\ap\). 
We deduce from above that \(\ap\) is bounded if{}f \( (N,\mmap_{\imath})\) is bounded.
Moreover, it is routine to check that \(\ap\) can be computed in time polynomial in the size
of \( (N,\mmap_{\imath})\).

\noindent %
{\bf Safety.}
Consider an instance of the coverability problem for Boolean \(\pn\). Because of the result of
Lem.~\ref{lem:binaryisnotaproblem} we can assume this instance has the following form: a
\(\pn\) \(N^{\flat}=(S\cup\set{p_i,p_c},T\cup\set{t_i,t_c}, F^{\flat})\), an 
initial marking \(\multi{p_i}\) and a marking to cover \(\multi{p_c}\).  Moreover
the only way to create a token in place \(p_c\) is by firing transition
\(t_c\). Observe that \(\ucl{\multi{p_c}}\cap\fire{\multi{p_i}}_{N^{\flat}}\),
namely \(p_c\) is coverable, if{}f there exists \(\mmap\) such that
\(\mmap\in\fire{\multi{p_i}}_{N^{\flat}}\) and \(\mmap\fire{t_c}\).

Then using the polynomial time construction given at
Fig.~\ref{fig-pnbound2aabound} we obtain a asynchronous program \(\ap\) which
satisfies the property that the global state \({\tt st}=(t_c,\varepsilon)
\) is reachable in \(\ap\) if{}f
\(\ucl{\multi{p_c}}\cap\fire{\multi{p_i}}_{N^{\flat}}\).

\noindent %
{\bf Configuration reachability.}
Consider an instance of the reachability problem for Boolean \(\pn\). Because of the
result of Lem.~\ref{lem:binaryisnotaproblem} we can assume this instance has
the following form: a \(\pn\)
\(N^{\flat}=(S\cup\set{p_i,p_r},T\cup\set{t_i,t_r}, F^{\flat})\) an initial
marking \(\multi{p_i,p_r}\) and \(\varnothing\) a marking to reach.  Moreover,
every transition sequence which reaches \(\varnothing\) ends with the firing of
\(t_r\).  Therefore using the polynomial time construction given at
Fig.~\ref{fig-pnbound2aabound} we obtain a asynchronous program \(\ap\) which
satisfies the property that the configuration \(c\) such that \(c.d\) is given
by \({\tt st}=(t_r,\varepsilon) \) and \(c.\mmap=\varnothing\) is reachable in
\(\ap\) if{}f \(\varnothing\in\fire{\multi{p_i,p_r}}_{N^{\flat}}\).

Hence we finally obtain the following results.
\begin{theorem}{}
\hspace{0pt}
\begin{enumerate}
\item The global state reachability and boundedness problems for asynchronous
	programs are {\sc expspace}-complete.
\item The configuration reachability problem for asynchronous programs is
	polynomial-time equivalent to the \(\pn\) reachability problem. 
      The configuration reachability problem is {\sc expspace}-hard.
\end{enumerate}
\end{theorem}

\subsection{Termination}

Since we now study properties of infinite runs of Petri nets modelling
asynchronous programs, there is a subset of transitions which becomes of
particular interest. This subset allows to distinguish the runs where some
widget enters a non terminating execution from those runs where each time a
widget runs, it eventually terminates.
Since our definition of asynchronous program does not allow for non-terminating
runs of a handler (see Rmk.~\ref{rmk:finiterunsonly}) we need a way to
discriminate non-terminating runs in the corresponding \(\pn\) widget.

\begin{definition}
	Let \(T^{d(a)}_{\ap}=\set{t_c^{>}\in T_{\ap}\mid c\in\mathfrak{C}\cap (D\times\set{a}\times D)}\) 
        for some $a\in\Sigma$ and let
\(T^d_{\ap}=\bigcup_{a\in\Sigma} T^{d(a)}_{\ap}\).
\end{definition}

\begin{definition}
	Let $\ap$ be an asynchronous program, and let \((N_{\ap},\mmap_{\imath})\) be
	an initialized \(\pn\) as given in Constr.~\ref{constr:pn}.
	Let \(\rho=\mmap_0\fire{t_0}\mmap_1\fire{t_1} \dots \mmap_n\fire{t_n} \dots  \)
	be an infinite run of \(N_{\ap}\) where $\mmap_0=\mmap_{\imath}$.
	\begin{itemize}
		\item $\rho$ is an infinite $\ap$-run if{}f $t_i\in T^d_{\ap}$ for infinitely many $i$'s;
		\item \(\rho\) is a fair infinite run if{}f
			\(
			\begin{cases}
				\rho \text{ is an infinite } \ap\text{-run, and}\\
				\text{for all } a\in\Sigma, \text{ if } t_i\in T^{d(a)}_{\ap} \text{ for finitely many }i\text{'s}\\
				\text{then }\mmap_{j}(a)=0 \text{ for infinitely many }j\text{'s}
			\end{cases}
			\);
		\item $\rho$ fairly starves $b (\in\Sigma)$ if{}f
			\(\begin{cases} \rho \text{ is a fair infinite run, and}\\
				\text{there is a } J\geq 0 \text{ such that for each } j\geq J\\
				\mmap_j(b)\geq 1\land (t_j\in T^{d(b)}_{\ap}\rightarrow \mmap_j(b)\geq 2)
				\end{cases}\).
	\end{itemize}
	\label{def:infinite-pnruns}
\end{definition}

\begin{lemma}
	Let $\ap$ be an asynchronous program and let \((N_{\ap},\mmap_{\imath})\) be
	an initialized \(\pn\) as given in Constr.~\ref{constr:pn}.
\begin{itemize}
	\item $\ap$ has an infinite run if{}f \((N_{\ap},\mmap_{\imath})\) has an infinite $\ap$-run;
	\item $\ap$ has a fair infinite run if{}f \((N_{\ap},\mmap_{\imath})\) has a fair infinite run;
	\item $\ap$ fairly starves $a$ if{}f \((N_{\ap},\mmap_{\imath})\) fairly starves $a$.
\end{itemize}
Moreover, \( (N_{\ap}(\pnfamily),\mmap_{\imath}) \) can be computed in
polynomial time from \(\ap\).%
\label{lem:equivpn_infinite}
\end{lemma}
\begin{proof}
	It suffices to observe that the Def.~\ref{def:infinite-pnruns} and
	Def.~\ref{def:infinite-aaruns} are equivalent using \eqref{eq:correctmodel}.
	The polynomial time construction was proven in Lem.~\ref{lem:equivpn}.
\end{proof}

We now give an {\sc expspace}-complete decision procedure for termination.
%
\begin{remark}
In what follows we assume a fixed linear order on the set of transitions $T$ (resp. places $S$)
which allow us to identify a multiset with a vector of $\nats^T$ (resp. $\nats^S$).
\end{remark}

We recall a class of path formulas
for which the model checking problem is decidable.
This class was originally defined in \cite{yen:paths}, but the model
checking procedure in that paper had an
error which was subsequently fixed in \cite{RP09}.
For simplicity, our definition below captures only a 
subset of the path formulas defined in \cite{yen:paths},
but this subset is sufficient to specify termination.  

Fix a
$\pn$ $N=(S,T,F,\mmap_{\imath})$.  
Let
$\mu_1,\mu_2,\dots$ be a family of \emph{marking variables} ranging over
\(\nats^{S}\) and $\sigma_1,\sigma_2,\dots$ a family of \emph{transition variables}
ranging over $T^*$.

\emph{Terms} are defined recursively as follows:
\begin{itemize}
	\item every $\cmap\in\nats^{S}$ is a term;
	\item for all $j>i$, and marking variables $\mu_j$ and $\mu_i$, we have $\mu_j-\mu_i$ is a term. 
	\item $\mathcal{T}_1 +\mathcal{T}_2$ and $\mathcal{T}_1-\mathcal{T}_2$ are
		terms if $\mathcal{T}_1$ and $\mathcal{T}_2$ are terms.
		(Consequently, every mapping $\cmap\in\mathbb{Z}^{S}$ is also a term)
\end{itemize}

\emph{Atomic predicates} are of two types: marking predicates and transition predicates.
\begin{description}
	\item[Marking predicates]
There are two types of marking predicates.
The first type consists in the forms 
$\mathcal{T}_1(p_1)=\mathcal{T}_2(p_2)$,
$\mathcal{T}_1(p_1)<\mathcal{T}_2(p_2)$,
and
$\mathcal{T}_1(p_1)>\mathcal{T}_2(p_2)$,
where $\mathcal{T}_1$ and $\mathcal{T}_2$ are terms and $p_1,p_2\in S$ are two places
of $N$. 
The second type consists in the forms $\mu(p)\geq z$ and $\mu(p)> z$, where $\mu$ is a marking variable,
$p\in S$, and $z\in\mathbb{Z}$.
\item[Transition predicates]
Define the \emph{inner product} $\otimes: \mathbb{Z}^T \times \mathbb{Z}^T \rightarrow \mathbb{Z}^T$ as
$\cmap_1\otimes\cmap_2 = \sum_{t\in T}\cmap_1(t)\cdot \cmap_2(t)$.
for $\cmap_1,\cmap_2\in\mathbb{Z}^{T}$.
A transition predicate is either of the form
$\parikh(\sigma_1)(t)\leq c$, where $c\in\nats$ and $t\in T$, or 
of the forms
	$\mathbf{y}\otimes\parikh(\sigma_i)\geq c$ 
and $\mathbf{y}\otimes\parikh(\sigma_i)\leq c$, where $i>1$, $c\in\nats$, 
$\mathbf{y}\in\mathbb{Z}^T$, and $\otimes$ denotes the inner
product as defined above.
\end{description}

A \emph{predicate} is a finite positive boolean combination of atomic
predicates.  A \emph{path formula} $\Lambda$ is a formula of the form:
\[
\exists\mu_1,\dots,\mu_m\exists \sigma_1,\dots,\sigma_m\colon
\bigl(\mmap_{\imath}\fire{\sigma_1}\mu_1\fire{\sigma_2}\dots\fire{\sigma_m}\mu_m\bigr)
\land \Phi(\mu_1,\dots,\mu_m,\sigma_1,\dots,\sigma_m)
\]
where $\Phi$ is a predicate.  A \emph{path formula} $\Lambda$ is increasing if
$\Phi$ implies $\mu_1\leq\mu_m$ (where $\mu_i\leq\mu_j$ for $i<j$ is an
abbreviation for $\bigwedge_{p\in S} (\mu_j-\mu_i)(p)> (-1^{S})(p)$) and
contains no transition predicate.
The size of a path formula is the number of symbols in the description of the formula,
where constants are encoded in binary.

The \emph{satisfiability problem} for a path formula $\Lambda$ asks if there exists
a run of $N$ of the form
$\mmap_{\imath}\fire{w_1}\mmap_1\fire{w_2}\dots\mmap_{m-1}\fire{w_m}\mmap_m$
for markings $\mmap_1,\ldots,\mmap_m$ and transition sequences
$w_1,\ldots,w_m\in T^*$, such that $\Phi(\mmap_1,\dots,\mmap_m,w_1,\dots,w_m)$
is true.  
If $\Lambda$ is satisfiable, we write $N\models\Lambda$.

\begin{theorem}{(from \cite{RP09})}
\begin{itemize}
	\item The satisfiability problem for a path formula is reducible in polynomial
		time to the reachability problem for Petri nets. Hence, the satisfiability
	problem is {\sc expspace}-hard. 
	\item The satisfiability problem for an increasing path formula is {\sc expspace}-complete.
\end{itemize}
	\label{thm-rp09}
\end{theorem}


We know define our reduction of the termination problem to the satisfiability
problem for an increasing path formula.

\begin{remark}
	Without loss of generality, we assume that in $\ap$, the set 
$\mmap_{0}$ of initial pending handler instances 
is given by the singleton
$\multi{a_0}$ for some $a_0\in\Sigma$ and $a_0$ is never
posted.
\label{rmk-wlog}
\end{remark}

\begin{lemma}
	Let $\ap$ be an asynchronous program and
	let \((N_{\ap},\mmap_{\imath})\) be
	an initialized \(\pn\) as given in Constr.~\ref{constr:pn}.
	Let $\Lambda_t$ be the path formula given by
\[
\exists\mu_1,\mu_2\colon\exists \sigma_1,\sigma_2\colon
\bigl(\mmap_{\imath}\fire{\sigma_1}\mu_1\fire{\sigma_2}\mu_2\bigr)
\land \mu_1\leq \mu_2\land T^{d}_{\ap}\otimes\parikh(\sigma_2)\geq 1 \enspace .  \]
	We have
	\vspace{0pt}\hspace{\stretch{1}}
	$(N_{\ap},\mmap_{\imath})\models\Lambda_t$ if{}f \((N_{\ap},\mmap_{\imath})\) has
	an infinite $\ap$-run.
	\hfill\vspace{0pt}
	\label{lem:pntermination}
\end{lemma}
\begin{proof}
	Let us first give a few facts about $\Lambda_{t}$:
\begin{itemize}
        \item {\it Fact 0:} $\Lambda_t$ is polynomial in the size of $(N_{\ap}, \mmap_{\imath})$.
	\item {\it Fact 1:} \(T^{d}_{\ap}\otimes\parikh(\sigma_2)\geq 1\) implies that
		\(\sigma_2\in T_{\ap}^*\cdot T^{d}_{\ap}\cdot T_{\ap}^*\)
		because it requires that some transition of \(T^d_{\ap}\) is fired along
		\(\sigma_2\);
	\item {\it Fact 2:} $\mu_1\leq \mu_2$ implies
		the sequence of transition given by $\sigma_2$ can be fired
		over and over.
\end{itemize}
Let us now turn to the proof.

{\bf Only if}:
Let $\mmap_1$, $\mmap_2$, $w_1$ and $w_2$ be a valuation of $\mu_1$, $\mu_2$,
$\sigma_1$ and $\sigma_2$ respectively such that $\Lambda_t$ is satisfied.
Fact 1 shows that $w_2\neq\varepsilon$ and $\parikh(w_2)(t)>0$ for some $t\in
T^d_{\ap}$.  Then Fact 2 shows that
$\mmap_{\imath}\fire{w_1}\mmap_1\fire{w_2^{\omega}}$ is an infinite
$\ap$-run of $N_{\ap}$ and we are done.

{\bf If}:
Let $\rho$ be an infinite $\ap$-run of $(N_{\ap},\mmap_{\imath})$. 
By definition  of infinite $\ap$-run, $\rho$ can be written as 
\(\mmap_{0}\fire{w_0}\mmap_1 \dots \mmap_n\fire{w_n}\ldots \)
where $\mmap_0=\mmap_{\imath}$ and for each $k\geq 0$, we have $w_k\in T_{\ap}^*\cdot T_{\ap}^{d}$.
By Dickson's Lemma \cite{dic13}, 
there exists two indices $i<j$ in the above infinite run such that
$\mmap_i\preceq \mmap_j$.
Let $\sigma_1=w_0 \dots w_{i-1}$,
$\sigma_2=w_{i}\dots w_j$, $\mu_1=\mmap_i$ and $\mu_2=\mmap_{j+1}$.
Clearly \(\mu_1\leq \mu_2\). 
Also we have that
$\sigma_2\neq\varepsilon$ because some transition of $T^{d}_{\ap}$
is in each $w_{k}$, and hence $T^d_{\ap}\otimes \parikh(\sigma_2) \geq 1$.
Thus, every conjunction of $\Lambda_t$ is satisfied.
\end{proof}

\begin{proposition}
	Given an asynchronous program $\ap$, determining the existence of
	an infinite run is {\sc expspace}-complete.
\label{prop-expdecisionproceduretermi}
\end{proposition}
\begin{proof}
	As expected our decision procedure relies on reductions to equivalent \(\pn\)
	problems.  We start by observing that the \(\pn\) \(N_{\ap}\) can
	be computed in time polynomial in the size of \(\ap\). 
	Lem.~\ref{lem:equivpn_infinite} shows that
	\(\ap\) has an infinite run if{}f \((N_{\ap},\mmap_{\imath})\) has an
	infinite \(\ap\)-run.  Next, Lem.~\ref{lem:pntermination} shows that
	determining whether \((N_{\ap},\mmap_{\imath})\) has an infinite \(\ap\)-run
	is equivalent to determining the satisfiability of
	\((N_{\ap},\mmap_{\imath})\models \Lambda_t\) where \(\Lambda_t\) can be computed
	in time polynomial in the size of \(N_{\ap}\). 
	The formula \(\Lambda_t\) is not an increasing path formula
	because it contains a transition predicate
	(\(T^{d}_{\ap}\otimes\parikh(\sigma_2)\geq 1\)). However the problem instance
	\((N_{\ap},\mmap_{\imath},\Lambda_t)\) can easily be turned into an
	equivalent instance \((N'_{\ap},\mmap'_{\imath},\Lambda'_t)\) that is
	computable in polynomial time and such that \(\Lambda'_t\) is a increasing
	path formula. This is accomplished by adding a place \(p_w\) to which a token
	is added each time some transitions of \(T^d_{\ap}\) is fired. Then it
	suffices to replace \(T^{d}_{\ap}\otimes\parikh(\sigma_2)\geq 1\) by
	\((\mu_2-\mu_1)(p_w)>0^{S}(p_w)\). It is routine to check that \(\Lambda'_t\)
	is a increasing path formula.

	Finally, the result of Thm.~\ref{thm-rp09} together with the fact that
	\(\Lambda'_t\) is an increasing path formula shows that the satisfiability of
	\((N'_{\ap},\mmap'_{\imath})\models \Lambda'_{t}\) can be determined in space
	exponential in the size of the input. Therefore we conclude that determining
	the existence of an infinite run in a given \(\ap\) has an {\sc expspace} upper bound.
	The {\sc expspace} lower bound follows by reduction from the termination of
	simple programs \cite{Lipton}.  Indeed, the construction of
	\cite{Lipton} (see also \cite{esparza-course}) shows how a deterministic
	$2^{2^n}$-bounded counter machine of size $O(n)$ can be simulated by a Petri
	net of size $O(n^2)$ such that the counter machine has an infinite
	computation if{}f the Petri net has an infinite execution and this construction
        is easily adapted to use asynchronous programs.
\end{proof}

\subsection{Fair Termination}

We now turn to fair termination.
\begin{lemma}
	Let $\ap$ be an asynchronous program and let \((N_{\ap},\mmap_{\imath})\) be
	an initialized \(\pn\) as given in Constr.~\ref{constr:pn}.
	Let $\Lambda_{\mathit{ft}}$ be the path formula given by
	\begin{gather*}
	\exists \mu_1,\mu_2,\mu_3\colon \exists \sigma_1,\sigma_2,\sigma_3\colon
	\mmap_{\imath}\fire{\sigma_1}\mu_1\fire{\sigma_2}\mu_2\fire{\sigma_3}\mu_3\\
		T_{\ap}\otimes\parikh(\sigma_1)\leq 0 \land \mu_2\leq\mu_3\land
		T^{d}_{\ap}\otimes\parikh(\sigma_3)\geq 1\\
		{\textstyle \bigwedge_{a\in\Sigma}}\Bigl(\mathbf{c}_a\otimes\parikh(\sigma_3)=0
	\rightarrow
	\bigl((\mathbf{p}_a-\mathbf{c}_a)\otimes\parikh(\sigma_2)=0 \land \mathbf{p}_a\otimes\parikh(\sigma_3)=0 \bigr)\Bigr)\notag
\end{gather*}
where $\mathbf{c}_a,\mathbf{p}_a\in\multiset{T_{\ap}}$
are s.t.
\(\mathbf{c}_a(t)=I(t)(a)\) and \(\mathbf{p}_a(t)=O(t)(a)\) for every \(t\in T_{\ap}\).
We have
\[
	(N_{\ap},\mmap_{\imath})\models\Lambda_{\mathit{ft}} \quad
	\mbox{if{}f} \quad (N_{\ap},\mmap_{\imath}) \mbox{ has a fair infinite run.}
\]
\label{lem:pnfairtermi}
\end{lemma}
\begin{proof}
As for termination (see Lem.~\ref{lem:pntermination}) we start with a 
few facts about $\Lambda_{\mathrm{ft}}$:
\begin{enumerate}
	\item For the sake of clarity we used an implication in $\Lambda_{\mathrm{ft}}$. However the equivalences
		$A\rightarrow B\equiv \neg A \lor B$ and $\mathbf{c}_a\otimes\parikh(\sigma_3)\neq0 \equiv
		\mathbf{c}_a\otimes\parikh(\sigma_3)>0$ shows that the above predicate is indeed a positive
		boolean combination of atomic predicates, hence $\Lambda_{\mathrm{ft}}$ is indeed
		a path formula.
        \item $\Lambda_{\mathrm{ft}}$ is polynomial in the size of the \(\pn\).
	\item \(T_{\ap}\otimes\parikh(\sigma_1)\leq 0\) ensures that $\sigma_1=\varepsilon$,
		hence that $\mu_1=\mmap_{\imath}$. The reason for this is to be able to use
		the more expressive transition predicate starting right from the initial marking.
	\item $\mu_2\leq \mu_3$ implies
		the sequence of transition given by $\sigma_3$ can be fired
		over and over (by monotonicity).
	\item \(T^{d}_{\ap}\otimes\parikh(\sigma_3)\geq 1\) ensures that
		\(\sigma_3\in T_{\ap}^*\cdot T^{d}_{\ap}\cdot T_{\ap}^*\) as for
		termination.
	\item The last conjunction ensures that each $a\in\Sigma$ is treated fairly.
		Intuitively, it says that if $\sigma_3$ does not dispatch $a\in\Sigma$
		(given by $\mathbf{c}_a\otimes\parikh(\sigma_3)=0$) then it must hold that $(i)$
		$a$ has been posted as many times as it has been dispatched along
		$\sigma_2$ (given by
		$(\mathbf{p}_a-\mathbf{c}_a)\otimes\parikh(\sigma_2)=0$),
                and $(ii)$
		$\sigma_3$ is not posting any call to $a$ (given by
		$\mathbf{p}_a\otimes\parikh(\sigma_3)=0$).
                Together, this means that there is no pending call to
                $a$ along the execution.
\end{enumerate}
We now turn to the proof.

{\bf Only if}:
Let $\mmap_{\mu_2}$, $\mmap_{\mu_3}$, $w_2$ and $w_3$ be a valuation of $\mu_2$, $\mu_3$,
$\sigma_2$ and $\sigma_3$ respectively such that $\Lambda_{\mathit{ft}}$ is
satisfied.  Note that by Fact (3) we know that since $\Lambda_{\mathit{ft}}$
holds we have $\sigma_1=\varepsilon$. Hence we find that
$\mmap_{\imath}\fire{w_2}\mmap_{\mu_2}\fire{w_3}\mmap_{\mu_3}$ where
$w_3\in T_{\ap}^*\cdot T^d_{\ap}\cdot
T_{\ap}^*$ by Fact (5).
Then Fact (4) shows that the run \(\rho\) given by
\(\mmap_{\imath}\fire{w_2}\mmap_{\mu_2}\fire{w_3^{\omega}}\) is an infinite
\(\ap\)-run of \((N_{\ap},\mmap_{\imath})\).

Let us now show $\rho$ is also a fair infinite run. We first rewrite $\rho$ as
\(\mmap_{0}\fire{t_0}\mmap_1\fire{t_1} \dots
\fire{t_{i-1}}\mmap_i\fire{t_i}\dots\) where \(\mmap_{0}=\mmap_{\imath}\),
\(w_2=t_0\dots t_{i-1}\) and \(w_3^{\omega}=t_i t_{i+1}\dots\) So we have that
$\mmap_i=\mmap_{\mu_2}$.

Our final step is to show that \(\rho\) matches a fair infinite
 run in \((N_{\ap},\mmap_{\imath})\). By hypothesis,
$\Lambda_{\mathrm{ft}}$ holds, so each implication holds.  
Fix $a\in\Sigma$.
We examine what the satisfaction of the implication entails.

(a) Assume that the left hand side of the implication does not hold.
This means that $w_3$ fires
some $t\in T^{d(a)}_{\ap}$, that is, some \(T^{d(a)}_{\ap}\)
occurs infinitely often along $w_3^{\omega}$, and the run is fair
w.r.t.\ $a$.

(b) If the left hand side of the implication holds, it means that
no \(T^{d(a)}_{\ap}\) is fired along $w_3$, hence
\(t_i\in T^{d(a)}_{\ap}\) 
holds for finitely many $i$'s in \(\rho\).
Because the implication is satisfied, Fact~(6) shows that, along $w_2$,
$a$ is posted as many times as it is dispatched.

We conclude from Remark~\ref{rmk-wlog} and $\mmap_{\imath}(a)=0$, that
\(\mmap_i(a)=\mmap_{\mu_2}(a)=0\), hence that for every position \(j\geq i\) we
have \(\mmap_{j}(a)=0\), namely \(\mmap_{j}(a)=0\) holds for infinitely many
$j$'s.

We conclude from the above cases that for every $a\in\Sigma$, we have
that if $t_i\in T^{d(a)}_{\ap}$ for finitely many $i$'s then
$\mmap_j(a)=0$ for infinitely many \(j\)'s, namely $\rho$ is a fair infinite 
run and we are done.

{\bf If}: Let \(\rho=\mmap_{0}\fire{t_0}\mmap_1\fire{t_1} \dots
\fire{t_{i-1}}\mmap_i\fire{t_i}\dots\) where \(\mmap_{0}=\mmap_{\imath}\) be a
infinite fair run of $(N_{\ap},\mmap_{\imath})$.
By definition we find
that $\rho$ is an infinite $\ap$-run and that for all $a\in\Sigma$,
if $t_i\in T^{d(a)}_{\ap}$ for finitely many $i$'s then $\mmap_j(a)=0$
for infinitely many $j$'s.
Define $S$ to be the set \(\set{a\in\Sigma\mid t_i\in
T^{d(a)}_{\ap} \text{ for finitely many $i$'s}}\).
Let $m$ denote a positive integer such that for all $n\geq m$
we have 
\(t_{n}\in T_{\ap}\setminus \bigcup_{a\in S}
T^{d(a)}_{\ap}\).  
Observe that, because the run is fair, for every
\(a\in S\) and for all \(n\geq m\), we have \(\mmap_{n}(a)=0\).

Let us now rewrite \(\rho\) as
\(\mmap_0\fire{t_0}\mmap_1\dots
\mmap_{m}\fire{t_m}\mmap_{i_0}\fire{w_{i_0}}\mmap_{i_1}\fire{w_{i_1}}\dots\)
such that \(\mmap_0=\mmap_{\imath}\) and for all $a\in\Sigma\setminus S$
some \(T^{d(a)}_{\ap}\) occurs in $w_{i_j}$ for
all \(j\geq 0\).

Now using Dickson's Lemma \cite{dic13} over the infinite sequence
$\mmap_{i_0},\mmap_{i_1},\dots,\mmap_{i_n},\dots$ of markings defined above we find
that there exists $\ell>k$ such that $\mmap_{i_k}\preceq\mmap_{i_{\ell}}$.

Define $\sigma_1=\varepsilon$, $\sigma_2=t_0\dots t_{m} w_{i_0}\dots
w_{i_{k-1}}$, $\sigma_3=w_{i_k}\dots w_{i_{\ell-1}}$, $\mu_1=\mmap_{\imath}$,
$\mu_2=\mmap_{i_k}$ and $\mu_3=\mmap_{i_{\ell}}$. Clearly \(\mu_2\leq \mu_3\)
and \(T_{\ap}\otimes\parikh(\sigma_1)\leq 0\).
Also, some transition of $T^d_{\ap}$ occurs in \(\sigma_3\) by
definition of \(w_{i_j}\), hence we find that
\(T^d_{\ap}\otimes\parikh(\sigma_3)\geq 1\).  

Let \(a\in\Sigma\). 
The implication
$\mathbf{c}_a\otimes\parikh(\sigma_3)=0
\rightarrow \bigl((\mathbf{p}_a-\mathbf{c}_a)\otimes\parikh(\sigma_2)=0 \land
\mathbf{p}_a\otimes\parikh(\sigma_3)=0 \bigr)$
is divided into two cases.

First, if \(a\in S\) then we find that no \(T^{d(a)}_{\ap}\) occurs
after \(t_m\). In particular no \(T^{d(a)}_{\ap}\) occurs in \(\sigma_3\)
and the left hand side of the implication holds.
We now show that so does the right hand side.
We showed above that \(\mmap_{n}(a)=0\) for every \(n\geq m\).  
By Rmk.~\ref{rmk-wlog}, initially
\(\mmap_{\imath}=\multi{a_0}\) and \(a_0\) never reappears in the
task buffer.
So, we find that
\((\mathbf{p}_a-\mathbf{c}_a)\otimes\parikh(\sigma_2)=0 \) holds.  Also
\(\mathbf{p}_a\otimes\parikh(\sigma_3)=0 \) holds because \(\mmap_n(a)=0\) for
each \(n\geq m\) and no \(T^{d(a)}_{\ap}\) occurs in \(\sigma_3\), hence no
post of \(a\) can occur in \(\sigma_3\).

Second, if \(a\in\Sigma\setminus S\) then
we find that some \(T^{d(a)}_{\ap}\) occurs along
\(\sigma_3\) by definition of the \(w_{i_j}\)'s.
Therefore the implication evaluates to true because its left hand side evaluates
to false.

This concludes the proof since every conjunction of
$\Lambda_{\mathrm{ft}}$ is satisfied.
\end{proof}

\begin{remark}
$\Lambda_{\mathit{ft}}$ is not an increasing path formula because we cannot
conclude it implies $\mu_1\leq\mu_3$. Since $\sigma_1=\varepsilon$, for
$\mu_1\leq\mu_3$ to hold we must have $\mmap_{\imath}\leq\mu_3$. 
Because of Rmk.~\ref{rmk-wlog} it is clearly the case that \(\mmap_{\imath}\nleq\mu_3\) since
\(\mmap_{\imath}=\multi{a_0}\) and \(a_0\) is first dispatched and never posted eventually.
\end{remark}

We now show a lower bound on the fair termination problem.
Given an initialized Boolean $\pn$ $(N=(S,T,F),\mmap_0)$ and a place
$p\in P$,
we reduce the problem of checking if there exists a
reachable marking with no token in place $p$ (which is recursively equivalent
to the reachability problem of a marking \cite{hack76}) if{}f an asynchronous
program constructed from the \(\pn\)
has a fair infinite run.
For the sake of clarity, let us index $S=\set{p_1,\dots, p_{\card{S}}}$ and 
assume that $p_1$ plays the role
of place $p$ in the above definition.

Fig.~\ref{fig-asyncFTisPNreach} shows an outline of the reduction from 
the reachability problem for $\pn$ to
the fair termination problem for asynchronous programs. 
The reduction is similar to the simulation shown in 
Fig.~\ref{fig-pnbound2aabound}. In particular,
we again define a global state ${\tt st}$, a procedure ${\tt runPN}$
to fire transitions, and $\card{S}$ procedures, one for
each $p_i\in S$.

The program has three global variables, two booleans ${\tt terminate}$ and
${\tt p\_1\_is\_null}$ and the variable ${\tt st}$ which ranges over a finite
subset of $(T\cup\set{\varepsilon})\times S^*$. 
The program has $\card{S}+3$ procedures: one procedure
for each $p_i\in S$, ${\tt main}$, ${\tt guess}$ and ${\tt runPN}$.  The role
of ${\tt main}$ is to initialize the global variables, and to post
${\tt runPN}$ and ${\tt guess}$.  
As before, the role of ${\tt runPN}$ is to simulate the transitions of
the \(\pn\).
The role of ${\tt guess}$ is related to 
checking whether there exists some marking
$\mmap\in\fire{\mmap_0}$ such that $\mmap(p_1)=0$, and is explained below.

The program of Fig.~\ref{fig-asyncFTisPNreach} preserves the same invariant as
the program of Fig.~\ref{fig-pnbound2aabound} and is as follows. 
Whenever the program state is such that ${\tt st}$ coincides
with $(t,\varepsilon)$ for some $t\in T\cup\set{\epsilon}$ we have that the
multiset $\mmap$ given by the pending instances to handler $p\in S$
is such that $\mmap\in\fire{\mmap_0}$ and there exists $w\in T^*$ such that
$\mmap_0\fire{w\cdot t}\mmap$.

We now explain the role played by procedure ${\tt guess}$ and the
variables ${\tt p\_1\_is\_null}$ and ${\tt terminate}$.  
After the dispatch of ${\tt main}$, ${\tt guess}$ is pending.  
As long as ${\tt guess}$ does not run the program behaves exactly like
the program of Fig.~\ref{fig-pnbound2aabound}.
That is, ${\tt runPN}$ selects a transition which, if enabled, fires. Once the
firing is complete ${\tt runPN}$ selects a transition, and so on. 
Now consider the dispatch of ${\tt guess}$ which must eventually occur
by fairness. 
It sets ${\tt p\_1\_is\_null}$ to true.  
This prevents ${\tt runPN}$ to repost itself, hence to select a
transition to fire. 
So the dispatch of ${\tt guess}$ stops the simulation. 
Now we will see that if the program has an infinite run then
the dispatch of ${\tt guess}$ has to occur in a configuration where $(i)$
\({\tt st}\in (T\cup\set{\varepsilon})\times \set{\varepsilon}\) and $(ii)$ the
marking \(\mmap\) corresponding to the current configuration is such that
$\mmap(p_1)=0$.  For $(i)$, we see that if the precondition of ${\tt st}$ does
not equal \(\varepsilon\) then ${\tt terminate}$ is set to true in \({\tt guess}\), hence every
dispatch that follows does not post, and the program eventually
terminates.  For $(ii)$, suppose that {\tt guess} runs and that in the current
configuration there is a pending instance to \(p_1\).  By fairness we find
that eventually \(p_1\) has to be dispatched.  Since ${\tt guess}$ has set
${\tt p\_1\_is\_null}$ to true we have that the dispatch of \(p_1\) sets
{\tt terminate} to true and the program will eventually terminate following
the same reasoning as above. So if the program has a fair infinite run then it
cannot have any pending instance of handler \(p_1\) after the dispatch of ${\tt
guess}$. The rest of the infinite run looks like this.  After the dispatch of
${\tt guess}$ we have that ${\tt runPN}$ is dispatched at most once.  Every
dispatch of a \(p_i\) for $i\in\set{2,\dots,\card{S}}$ will simply repost
itself since ${\tt st}$ has an empty precondition and the value of ${\tt
terminate}$ is false. This way we have a run \(\rho\) with infinitely many dispatches and no 
effect: \(\rho\) leaves the program in the exact same configuration that
corresponds to a marking $\mmap\in\fire{\mmap_0}$ such that $\mmap(p_1)=0$.
Notice that if current configuration of the program corresponds to the marking
$\mmap=\varnothing$ we have that $\mmap(p_1)=0$ but the program terminates. We
can avoid this undesirable situation by adding one more place $p^g$ to the
$\pn$ such that it is marked initially and no transition is connected to $p^g$.

Let us now turn to the other direction. 
 Suppose there exists $w\in T^*$ such that
$\mmap_{\imath}\fire{w}\mmap$ with $\mmap(p_1)=0$. The infinite fair run of the
asynchronous program has the following form. The invariant shows that the
program can simulate the firing of $w$ and ends up in a configuration with no
pending instance to handler \(p_1\) and such that the precondition of ${\tt st}$ is
$\varepsilon$.  Then ${\tt guess}$ is dispatched followed by a fair infinite 
sequence of dispatch for \(p_i\) where $i\in\set{2,\dots,\card{S}}$.
Because of ${\tt st}$ the dispatch of \(p_i\) has no effect but reposting
\(p_i\). So we have a fair infinite run.

This shows that the fair termination problem is polynomial-time equivalent to
the Petri net reachability problem.

The reduction also suggests that finding an increasing path formula
for fair termination will be non-trivial, 
since it would imply that Petri net reachability is in {\sc expspace}.

\begin{figure*}[h]
  \centering
  \hspace{\stretch{1}}
    \begin{minipage}{.5\textwidth}
      \begin{tabbing}
\ \ \ \ \=\ \ \ \=\ \ \ \=\\
{\it global} {\tt st}, {\tt p\_1\_is\_null}, {\tt terminate};\\
\\
{\tt main}() \{\\
\>  {\tt st}\(=(\varepsilon,\varepsilon)\);\\
\>  {\tt p\_1\_is\_null}={\tt false};\\
\>  {\tt terminate}={\tt false};\\
\\
\>  {\tt post} {\tt runPN}();\\
\>  {\tt post} {\tt guess}();\\
\}\\
\\
{\tt guess}() \{\\
\> {\tt p\_1\_is\_null}={\tt true};\\
\> if (\({\tt st}\notin (T\cup\set{\varepsilon})\times \set{\varepsilon} \))\{\\
\>\>    {\tt terminate}={\tt true};\\
\>  \}\\
\}\\
\\
{\tt runPN}() \{\\
\>  if {\tt p\_1\_is\_null}=={\tt false} \{\\
\>\>if (\texttt{st} \(\in (T\cup\set{\varepsilon})\times \set{\varepsilon}\)) \{\\
\>\>\>      pick \(t'\in T\) non det.;\\
\>\>\>			\texttt{st}=(\(t',\hat{I}(t')\));\\
\>\>    \}\\
\>\>		{\tt post} \texttt{runPN}();\\
\>  \}\\
\}\\
\\
Initially: \(\mmap_{0} \oplus \multi{{\tt main}}\)
      \end{tabbing}
    \end{minipage}%
  \hspace{\stretch{1}}
    \begin{minipage}{.45\textwidth}
			     \begin{tabbing}
\ \ \ \ \=\ \ \ \=\ \ \ \=\ \ \ \=\ \ \ \=\ \ \ \=\\
\(p_1\)() \{\\
\>if \texttt{p\_1\_is\_null}==\texttt{true} \{\\
\>\>\texttt{terminate}=\texttt{true};\\
\>  \} else \{\\
\>\>	if \texttt{st}==(\(t, p_1\cdot w'\)) \{\\
\>\>\>	\texttt{st}=(\(t, w'\));\\
\>\>\>  if \(w'\)==\(\varepsilon\) \{\\
\>\>\>\>      for each \(j \in \set{1,...,\card{S}}\) do \{ \\
\>\>\>\>\>       if \(O(t)(p_j)>0\) \{\\
\>\>\>\>\>\>        {\tt post} $p_j$();\\
\>\>\>\>\>       \}\\
\>\>\>\>       \}\\
\>\>\>    \} \\
\>\>   \} else \{\\
\>\>\>		if \texttt{terminate} == \texttt{false} \{\\
\>\>\>\>        {\tt post} \(p_1\)();\\
\>\>\>      \}\\
\>\>    \}\\
\>  \}\\
\}\\
\\
\(p_i\)() \{ // for \(i \in \set{2,...,\card{S}}\)\\
\>if \texttt{st}==(\(t, p_i\cdot w'\)) \{\\
\>\>  \texttt{st}=(\(t, w'\));\\
\>\>    if \(w'\)==\(\varepsilon\) \{\\
\>\>\>      for each \(j \in \set{1,...,\card{S}}\) do \{ \\
\>\>\>\>       if \(O(t)(p_j)>0\) \{\\
\>\>\>\>\>        {\tt post} $p_j$();\\
\>\>\>\>       \}\\
\>\>\>       \}\\
\>\>    \}\\
\>  \} else \{\\
\>\>	if \texttt{terminate} == \texttt{false} \{\\
\>\>\>      {\tt post} \(p_i\)();\\
\>\>    \}\\
\>  \}\\
\}
\end{tabbing}
    \end{minipage}%
  \hspace{\stretch{1}}
	\caption{Let \((N=(S,T,F),\mmap_0)\) be an initialized Boolean $\pn$ such that
	$p_1\in S$ and \(\forall t\in T\colon\card{I(t)}>0\).
	\(\exists \mmap\in\fire{\mmap_0}\colon \mmap(p_1)=0\) if{}f
	the asynchronous program has a fair infinite execution.}
	\label{fig-asyncFTisPNreach}
\end{figure*}

\begin{proposition}
Given an asynchronous program $\ap$, determining the existence of a
fair infinite run is polynomial-time equivalent to
the reachability problem for \(\pn\).
Hence, it is {\sc expspace}-hard
and can be solved in 
non-primitive recursive space. 
\label{prop-decisionprocedurefairtermi}
\end{proposition}
\begin{proof}
	As in Prop.~\ref{prop-expdecisionproceduretermi} our decision procedure
	relies on reductions to equivalent \(\pn\) problems.
	Define \(N_{\ap}\) to be the \(\pn\) given by \(N_{\ap}(\pnfamily)\).
	Lem.~\ref{lem:equivpn_infinite} shows that \(\ap\) has a fair infinite run
	if{}f so does \((N_{\ap},\mmap_{\imath})\).
	Next, Lem.~\ref{lem:pnfairtermi} shows that determining whether \( (N_{\ap},\mmap_{\imath})\)
	has a fair infinite run is equivalent to determining
	the satisfiability of \((N_{\ap},\mmap_{\imath})\models\Lambda_{\mathit{ft}}\) where
	\(\Lambda_{\mathit{ft}}\) is computable in time polynomial in the size of \(N_{\ap}\).
  Finally, Th.~\ref{thm-rp09} shows that the satisfiability
  of \((N_{\ap},\mmap_{\imath})\models\Lambda_{\mathit{ft}}\) is reducible
	to a reachability problem for $\pn$. The best known upper bounds for the
	reachability problem in $\pn$ take non-primitive recursive space. 
	Therefore, we conclude that determining the existence
	of a fair infinite run in a given \(\ap\) can be solved in non-primitive recursive space.

	The lower bound is a consequence of (1) the reduction from the reachability
	problem for \(\pn\) to the fair termination problem for asynchronous program
	given at Fig.~\ref{fig-asyncFTisPNreach} and (2) the {\sc expspace} lower
	bound for the reachability problem for \(\pn\). 
\end{proof}

\subsection{Fair starvation}

Recall that the fair starvation property states that there is no pending
handler instance that is starved (i.e. never leaves the task buffer) along any
fair infinite run.

In order to solve the fair starvation problem, we first define
Constr.~\ref{constr:pnfs} which modifies Constr.~\ref{constr:pn} by introducing
constructs specific to the starvation problem.  In what follows, we assume that
the assumption of Rmk.~\ref{rmk-wlog} holds.

We first give some intuition. A particular pending instance of handler \(a\) starves
if there exists a fair infinite execution such that from some point in time ---
call it \(\dag\) --- there exists an instance of handler \(a\) in the task
buffer and it never leaves it.
Because the run is fair and there exists at least one instance of handler \(a\) in the task buffer, we find that
\(a\) is going to be dispatched infinitely often. In this case, a particular instance of handler \(a\) never
leaves the task buffer if{}f each time a dispatch to \(a\) occurs the task buffer contains two or more instances of \(a\).

In order to capture infinite fair runs of an asynchronous program that starves
a specific handler \(a\), we modify the Petri net construction as follows.  The
\(\pn\) has two parts: the first part simulates the asynchronous program as
before, and the second part which also simulated the asynchronous program
ensures that an instance of handler \(a\) never leaves the task buffer. 
In order to ensure that condition, the Petri net simply requires that any dispatch of \(a\)
requires at least {\em two} pending instances of \(a\) rather than just one (as
in normal simulation), and the dispatch transition consumes one instance of
\(a\) and puts back the second instance. The Petri net non-deterministically
transitions from the first part of the simulation to the second. The transition
point serves as a guess of time point \(\dag\) from which the task buffer
always contains at least pending instance of handler \(a\).  We now formalize
the intuition.

\begin{construction}[Petri net for fair starvation]\label{constr:pnfs}
Let \(\ap=(D,\Sigma,G,R,d_0,\mmap_0)\) be an asynchronous program.  
Let \(\pnfamily^{\spadesuit}=\set{N^{\spadesuit}_c}_{c\in\mathfrak{C}}\) and
\(N^{\spadesuit}_c=(S^{\spadesuit}_c,T^{\spadesuit}_c,F^{\spadesuit}_c)\) be an adequate
family of \emph{widgets}. 

Let $a\in\Sigma$. 
Define \(\mathfrak{C}^{a}\) to be the set \(\mathfrak{C}\cap (D\times\set{a}\times D)\)
and \((N^{a}_{\ap}(\pnfamily^{\spadesuit}),\mmap'_{\imath})\)
to be an initialized \(\pn\) where
(1) $N^{a}_{\ap}(\pnfamily^{\spadesuit})=(S_{\ap}, T_{\ap}, F_{\ap})$ is given as follows:
		\begin{itemize}
			\item $S_{\ap}=D\cup\Sigma\cup \bigcup_{c\in\mathfrak{C}} S^{\spadesuit}_c\cup \set{p_f,p_{\infty}}$
			\item \(T_{\ap}=\set{t^{f/\infty}}\cup \bigcup_{c\in\mathfrak{C}} T^{\spadesuit}_c
				\cup \set{t_c^{<}}_{c\in \mathfrak{C}\setminus\mathfrak{C}^{a}} \cup\set{t_c^{<f}, t_c^{<\infty}}_{c\in \mathfrak{C}^{a}}\cup \set{t_c^{>}}_{c\in\mathfrak{C}} \)
			\item $F_{\ap}$ is given by
				\begin{align*}
						F_{\ap}(t^{f/\infty})&=\tuple{\multi{p_f},\multi{p_{\infty}}}\\
						F_{\ap}(t_c^{<})&=\tuple{\multi{d_1,b},\multi{(\mathit{begin},c)}} &c=(d_1,b,d_2)\in\mathfrak{C}\setminus\mathfrak{C}^{a}\\
						F_{\ap}(t_c^{<f})&=\tuple{\multi{d_1,a,p_f},\multi{(\mathit{begin},c),p_f}} & c=(d_1,a,d_2)\in\mathfrak{C}^{a}\\
						F_{\ap}(t_c^{<\infty})&=\tuple{\multi{d_1,a,a,p_{\infty}},\multi{(\mathit{begin},c),a,p_{\infty}}} & c=(d_1,a,d_2)\in\mathfrak{C}^{a}\\
						F_{\ap}(t)&=F^{\spadesuit}_c(t)
                                                & t\in T^{\spadesuit}_c\\
						F_{\ap}(t_c^{>})&=\tuple{\multi{(\mathit{end},c)}, \multi{d_2}} & c=(d_1,b,d_2)\in\mathfrak{C}
				\end{align*}
		\end{itemize}
		and (2) $\mmap'_{\imath}=\multi{d_0,p_f}\oplus \mmap_0$.
\end{construction}

In an execution of the \(\pn\), the occurence of transition \(t^{f/\infty}\)
corresponds to the Petri net's transition from the first mode of
simulation to the second, i.e., the guess of the point \(\dag\) in time 
from which an instance of \(a\) never leaves the task buffer.

In what follows we use the notation \(N^{a}_{\ap}\) to denote
an adequate family \(N^{a}_{\ap}(\pnfamily^{\spadesuit})\).

\begin{lemma}
	Let $\ap$ be an asynchronous program and let
	\(\mathcal{N}=\set{N_c}_{c\in\mathfrak{C}}\) be an adequate family.
	Define \((N_{\ap},\mmap_{\imath})\) to be the initialized \(\pn\) \( (N_{\ap}(\mathcal{N}),\mmap_{\imath})\) as in Constr.~\ref{constr:pn} and
	given \(a\in\Sigma\) define \( (N^{a}_{\ap},\mmap'_{\imath})\) to be the
	initialized \(\pn\) \( (N^{a}_{\ap}(\mathcal{N}),\mmap'_{\imath})\) 
	as in Constr.~\ref{constr:pnfs}
	Let the path formula
$\Lambda^{a}_{\mathit{fs}}$ given by 
\begin{gather*}
\exists \mu_1,\mu_2,\mu_3 \exists \sigma_1,\sigma_2,\sigma_3\colon \mmap_{\imath}\fire{\sigma_1}\mu_1\fire{\sigma_2}\mu_2\fire{\sigma_3}\mu_3\\
	T_{\ap}\otimes\parikh(\sigma_1)\leq 0 \land \mu_2\leq\mu_3\land
	T^{d(a)}_{\ap}\otimes\parikh(\sigma_3)\geq 1 \land
	\multi{t^{f/\infty}}\otimes\parikh(\sigma_2)>0\\
	{\textstyle \bigwedge_{b\in\Sigma}}\Bigl(\mathbf{c}_b\otimes\parikh(\sigma_3)=0
	\rightarrow
	\bigl((\mathbf{p}_b-\mathbf{c}_b)\otimes\parikh(\sigma_2)=0 \land \mathbf{p}_b\otimes\parikh(\sigma_3)=0 \bigr)\Bigr)
\end{gather*}
where \(\mathbf{c}_b(t)=I(t)(b)\) and \(\mathbf{p}_b(t)=O(t)(b)\) for every \(t\in T_{\ap}\).\\
We have
\[
(N_{\ap}^{a},\mmap'_{\imath})\models\Lambda^{a}_{\mathit{fs}}
\quad \mbox{if{}f} \quad  (N_{\ap},\mmap_{\imath} ) \mbox{ fairly starves }a
\]
	\label{lem:pnstarve}
\end{lemma}
\begin{proof}
{\bf If}:
\( (N_{\ap},\mmap_{\imath}) \) fairly starves \(a\) implies the existence of
a fair infinite run \(\rho=\mmap_{0}\fire{t_0}\mmap_1\fire{t_1}\dots\) and
an index \(J\geq 0\) such that for each \(j\geq J\) we have
\(\mmap_j(a)\geq 1\land (t_j\in T_{\ap}^{d(a)}\rightarrow \mmap_j(a)\geq 2)\).

To show \(\rho\) yields the existence of a run \(\rho'\) in \(
(N^a_{\ap},\mmap'_{\imath})\) which satisfies \(\Lambda^a_{\mathit{fs}}\), we
first define a set of positions in \(\rho\) as we did in
Lem.~\ref{lem:pnfairtermi} for fair termination.  Let \(b\in\Sigma\), we define
\(m_b\) such that if every transition in \(T^{d(b)}_{\ap}\) occur finitely
often then \(m_b\) is greater than the last such occurrence; else (some \(t\in
T^{d(b)}_{\ap}\) occur infinitely often) \(m_b=0\).  Define \(m\) to be the
maximum over \(\set{J}\cup\set{m_b\mid b\in\Sigma}\).

Let us now rewrite \(\rho\) as the following infinite run
\begin{equation}\label{rho-starvation}
\mmap_0\fire{t_0}\mmap_1\dots\mmap_{m}\fire{t_m}\mmap_{i_0}\fire{w_{i_0}}\mmap_{i_1}\fire{w_{i_1}}\dots
\end{equation}
such that for every \(b\in\Sigma\) if some \(t\in T^{d(b)}_{\ap}\) occurs infinitely
often then that \(t\) occurs in each \(w_{i_j}\) for \(j\geq 0\).

Our next step is to associate to \(\rho\) a counterpart \(\rho'\) in
\( (N_{\ap}^a,\mmap'_{\imath}) \).
The run $\rho$ from Eqn.~\ref{rho-starvation}  is associated with the trace \(\rho'\) given by
\begin{gather*}
	\mmap_{0}\!\oplus\!\multi{p_f}\fire{t'_0}\dots\mmap_{m}\!\oplus\!\multi{p_f}\fire{t'_m}\mmap_{i_0}\!\oplus\!\multi{p_f}\fire{t^{f/\infty}}\mmap_{i_0}\!\oplus\!\multi{p_{\infty}}\fire{w'_{i_0}}\mmap_{i_1}\!\oplus\!\multi{p_{\infty}}\dots
\end{gather*}
where \(\mmap_{\imath}=\mmap_{0}\),
\(\mmap'_{\imath}=\mmap_0\oplus\multi{p_f}\).  \(\rho'\) is such that before
the occurrence of \(t^{f/\infty}\), if \(t_i=t_c^{<}\) where \(c\in
\mathfrak{C}^a\) then \(t'_i=t_c^{<f}\); else \((c\in\mathfrak{C}\setminus\mathfrak{C}^{a})\) \(t'_i=t_i\).  Moreover after the
occurrence of \(t^{f/\infty}\), if \(t_i=t_c^{<}\) where \(c\in
\mathfrak{C}^a\) then \(t'_i=t_c^{<\infty}\); else \(t'_i=t_i\).

Since \(m\geq J\) and \(\rho\) fairly  starves \(a\), we deduce that for every
\(j\geq m\) we have \(\mmap_{j}(a)\geq 1\) and \(t_j\in
T^{d(a)}_{\ap}\rightarrow \mmap_j(a)\geq 2\). This implies that the transitions
of the form \(t_c^{<\infty}\) which occur after \(t^{f/\infty}\) only, hence
after \(m\), are enabled because their counterpart \(t_c\) in \(N_{\ap}\) is
enabled in \(\rho\). Hence we conclude that \(\rho'\) is a run of
\((N^a_{\ap},\mmap'_{\imath})\),

Now using Dickson's Lemma \cite{dic13} over the infinite sequence
$\mmap_{i_0},\mmap_{i_1},\dots,\mmap_{i_n},\dots$ of markings defined above we
find that there exists $\ell>k$ such that $\mmap_{i_k}\preceq\mmap_{i_{\ell}}$.

Finally, let $\sigma_1=\varepsilon$, $\sigma_2=t'_0\dots t'_{m}
t^{f/\infty}
w'_{i_0}\dots w'_{i_{k-1}}$, $\sigma_3=w'_{i_k}\dots w'_{i_{\ell-1}}$,
$\mu_1=\mmap_{\imath}$, $\mu_2=\mmap_{i_k}$ and $\mu_3=\mmap_{i_{\ell}}$.
Clearly \(\mu_2\leq \mu_3\), $T_{\ap}\otimes\parikh(\sigma_1)\leq 0$
and \(\multi{t^{f/\infty}}\otimes\parikh(\sigma_2)>0\).
We conclude from \(\mmap_{\ell}(a)\geq 1\) for all \(\ell\geq m\) and because
\(\rho\) is fair that some \(T^{d(a)}_{\mathfrak{P}}\) must occur infinitely
often, hence that it occurs in \(w'_{i_j}\) for all \(j\geq 0\), and finally
that \(T^{d(a)}_{\ap}\otimes\parikh(\sigma_3)\geq 1\) by definition of
\(\sigma_3\).
Finally let \(b\in\Sigma\), the implication
\(\mathbf{c}_b\otimes\parikh(\sigma_3)=0
\rightarrow \bigl((\mathbf{p}_b-\mathbf{c}_b)\otimes\parikh(\sigma_2)=0 \land
\mathbf{p}_b\otimes\parikh(\sigma_3)=0 \bigr)\) holds using arguments
similar to the proof of Lem.~\ref{lem:pnfairtermi}.
This concludes this part of the proof since every conjunction of
$\Lambda^a_{\mathrm{fs}}$ is satisfied.

{\bf Only if}:
The arguments used here are close to the ones of Lem.~\ref{lem:pnfairtermi}.
Let $\mmap_{\mu_2}$, $\mmap_{\mu_3}$, $w_2$ and $w_3$ be a valuation of
$\mu_2$, $\mu_3$, $\sigma_2$ and $\sigma_3$ respectively such that
$\Lambda^a_{\mathit{fs}}$ is satisfied.  \(T_{\ap}\otimes\parikh(\sigma_1)\leq
0\) shows that $\sigma_1=\varepsilon$.  Hence we find that
$\mmap_{\imath}\fire{w_2}\mmap_{\mu_2}\fire{w_3}\mmap_{\mu_3}$ where $w_3\in
(T_{\ap})^*\cdot T^{d(a)}_{\ap}\cdot (T_{\ap})^*$ because
\(T_{\ap}^{d(a)}\oplus\parikh(\sigma_3)\geq 1\) holds.  Then \(\mu_2\leq
\mu_3\) shows that the run \(\rho\) given by
\(\mmap_{\imath}\fire{w_2}\mmap_{\mu_2}\fire{w_3^{\omega}}\) is an infinite run
of \((N^{a}_{\ap},\mmap'_{\imath})\).
\(\multi{t^{f/\infty}}\oplus\parikh(\sigma_2)>0\) where
\(F_{\ap}(t^{f/\infty})=\tuple{\multi{p_f},\multi{p_{\infty}}}\) shows that the
token initially in \(p_f\) moves to \(p_{\infty}\) while \(w_2\) executes.

Our next step is to show that \(\rho\) matches a run \(\rho'\) in
\((N_{\ap},\mmap_{\imath})\) which fairly starves \(a\). By hypothesis,
$\Lambda^a_{\mathrm{fs}}$ holds and so does each implication.  Let
$b\in\Sigma$, we examine the satisfiability of the implication.

(a) Assume that the left hand side does not hold which means that $w_3$ fires
some $t\in T^{d(b)}_{\ap}$, that is some \(T^{d(b)}_{\ap}\)
occurs infinitely often along $w_3^{\omega}$.

(b) If the left hand side of the implication holds we
find that no \(T^{d(b)}_{\ap}\) is fired along $w_3$, hence \(t_i\in
T^{d(b)}_{\ap}\) holds for finitely many $i$'s in \(\rho\).
Observe that \(b\neq a\) because we showed some \(T^{d(a)}_{\ap}\)
fires infinitely often in \(\rho\).
Because the implication is satisfied, along $w_2$, $b$ is posted as many times
as it is dispatched.

Hence, using similar arguments as those of Lem.~\ref{lem:pnfairtermi} that we
will not repeat here, we find that \(\rho'\) is a fair infinite run.

Also since \(\multi{t^{f/\infty}}\otimes\parikh(\sigma_2)>0\) holds, we find that
\(t^{f/\infty}\) occurs in \(w_2\).
This together with the fact that some transition of \(T^{d(a)}_{\ap}\) fires infinitely often in
\(w_3^{\omega}\) implies that each time a token is removed from \(a\) (through
some \(t^{<\infty}_c\) for some \(c\)) at least one token remains, hence
\(\mmap_i(a)\geq 2\) before a token is removed from \(a\), hence \(\rho\)
fairly starves \(a\).

Our last step shows that $\rho$ has a counterpart \(\rho'\) in
\((N_{\ap},\mmap_{\imath})\) and \(\rho'\) is fairly starving \(a\).
Let us define \(\rho'\) by abstracting away from \(\rho\) the places
\(\set{p_f,p_{\infty}}\) and the occurrence of \(t^{f/\infty}\). Clearly
\(\rho'\) is an infinite run of \( (N_{\ap},\mmap_{\imath}) \) fairly starving
\(a\).
\end{proof}

\begin{proposition}
	Given an asynchronous program $\ap$, determining the existence of a
	run that fairly starves some $a\in\Sigma$ is polynomial-time equivalent to \(\pn\) reachability.
        Fair starvation for asynchronous programs is {\sc
          expspace}-hard and can be solved in non-primitive recursive space.
\label{prop-decisionprocedurestarvation}
\end{proposition}
\begin{proof}
	As in Prop.~\ref{prop-decisionprocedurefairtermi} our decision procedure
	relies on reductions to equivalent \(\pn\) problems.  Fix \(N_{\ap}^a\) to be
	the \(\pn\) given by \(N_{\ap}^a(\pnfamily)\).
	Lem.~\ref{lem:equivpn_infinite} shows that \(\ap\) has a run that fairly starves
	some \(a\in\Sigma\) if{}f so does \((N_{\ap},\mmap_{\imath})\).
	Next, Lem.~\ref{lem:pnstarve} shows that determining whether \( (N_{\ap},\mmap_{\imath})\)
	has a run that fairly starves a given \(a\in\Sigma\) is equivalent to
	determining the satisfiability of
	\((N^{a}_{\ap},\mmap'_{\imath})\models\Lambda^{a}_{\mathit{fs}}\) where
	\(N^{a}_{\ap}\) and \(\mmap'_{\imath}\) are given as in
	Constr.~\ref{constr:pnfs}. The reduction from the problem of determining if
	\(\ap\) fairly starves to the problem of checking whether
	\((N^{a}_{\ap},\mmap'_{\imath})\models\Lambda^{a}_{\mathit{fs}}\) holds can
	be carried out in polynomial time.

	Finally, Th.~\ref{thm-rp09} shows that the satisfiability of
	\((N^{a}_{\ap},\mmap'_{\imath})\models\Lambda^{a}_{\mathit{fs}}\) is reducible to
	a reachability problem for $\pn$ which can be solved using non-primitive
	recursive space. Therefore, we
	conclude that determining the existence of a run that fairly starves \(a\)
	for a given \(\ap\) and \(a\in\Sigma\) can be solved using non-primitive
	recursive space.

	The lower bound is established similarly to the reduction for 
	fair termination (see the asynchronous program \(\ap\) of
        Fig.~\ref{fig-asyncFTisPNreach}). 
        Let us recall some intuition. 
	After a finite amount of time, \(\ap\) guesses that the current state of the task buffer
	has no pending instance to \(p_1\). If the guess is wrong, \(\ap\) will eventually terminate.
	If the guess is correct then the program will enter into a fair infinite run
	\(\rho\).  We can massage \(\ap\) so that \(\rho\) is a fair infinite run
	starving a given handler \(p_{\spadesuit}\).
	Initially, the task buffer contains one pending instance to a special handler \(p_{\spadesuit}\).
	If {\tt terminate} is false, then \(p_{\spadesuit}\) posts itself twice; otherwise
	it does not do anything. This guarantess that if \(\ap\)
        incorrectly guesses when $p_1$ is empty, then
	the number of pending instance to \(p_{\spadesuit}\) will
        eventually be \(0\) and \(\ap\)
	will terminate as above. Otherwise, if \(\ap\) correctly
        guesses when $p_1$ is empty, the number
	of pending instances of \(p_{\spadesuit}\) will grow unboundedly, therefore preventing
	some pending \(p_{\spadesuit}\) to ever complete.
        The {\sc expspace}-hardness follows from the corresponding hardness for Petri net reachability.
\end{proof}

\section{Extensions: Asynchronous Programs with Cancellation}

The basic model for asynchronous programming considered so far allows {\em
posting} a handler, but not doing any other changes to the task buffer.
In practice, APIs or languages for asynchronous
programming provide additional capabilities, such as {\em canceling} one or all
pending instances of a given handler, and checking if there are pending instances
of a handler.
For example, the {\tt node.js} library for Javascript allows canceling all posted
handlers of a certain kind.
A model with cancellation can also be used to abstractly model asynchronous programs with
timeouts associated with handlers, i.e., where a handler should not be called after
a specific amount of time has passed since the post.

We now discuss extensions of asynchronous programs that model
cancellation of handlers.

\subsection{Formal model}

We now give a model for asynchronous programming in which the programmer can
perform asynchronous calls as before, but in addition can {\em cancel} pending
instances of a given handler. Informally, the command \({\tt cancel}~f()\)
immediately removes every pending handler instances for $f$ from the task
buffer.

To model this extension, we define an extension of asynchronous programs called
\emph{asynchronous programs with cancel}.  The first step is to associate to
every handler \(f\) an additional symbol $\bar{f}$, which intuitively represents
a cancellation of handler \(f\in\Sigma\).

Let $\Sigma$ be the set of handler names, we denote by $\Sigmabar$ a distinct
copy of $\Sigma$ such that for each $\sigma\in\Sigma$ we have
$\bar{\sigma}\in\Sigmabar$.  So in the settings with cancel, an asynchronous
program defines an extended alphabet \(\Gamma=\Sigma_i\cup\Sigma\cup\Sigmabar\)
which respectively model the statements, the posting and cancellation of
handler instances.  We thus have that an asynchronous program
with cancel \(\ap=(D,\Sigma\cup\Sigmabar,\Sigma_i,G,R,d_0,\mmap_0)\) 
consists of a finite set of
global states \(D\), an alphabet \(\Sigma\cup\Sigmabar\) of for handler calls and cancels, a \(\cfg\)
\(G=(\mathcal{X},\Gamma,\prod)\), a regular grammar \(R=(D,\Gamma,\delta)\), a
multiset \(\mmap_0\) of initial pending handler instances, and an initial state
\(d_0\in D\).

As with asynchronous programs without cancel, we model the (potentially
recursive) code of a handler using a context-free grammar. The code of a
handler does two things: first, it can change the global state (through $R$),
and second, it can add and remove pending handler instances from the task
buffer (through derivation of a word in $(\Sigma\cup\Sigmabar)^*$).
In fact, a symbol \(\sigma\in\Sigma\) is interpreted as a post of handler
\(\sigma\) and a symbol \(\bar{\sigma}\in\Sigmabar\) is interpreted as the
removal of all pending instances to handler \(\sigma\).

The set of configurations of \(\ap\) is given by  \(
D\times\multiset{\Sigma}\).  Observe it does not differ from asynchronous
programs without cancel.  
The transition relation $\rightarrow \subseteq (D\times\multiset{\Sigma})\times
(D\times\multiset{\Sigma})$ is defined as follows:
let \(\mmap,\mmap'\in\multiset{\Sigma}\), \(d,d'\in D\) and \(\sigma\in\Sigma\)
\begin{gather*}
	(d,\mmap\oplus\multi{\sigma})\overset{\sigma}\rightarrow (d',\mmap')\\
\text{ if{}f }\\
\exists w\in\Gamma^* \colon d \underset{R}\Rightarrow^* w \cdot d'
\land X_{\sigma} \underset{G}\Rightarrow^* w \land \forall
b\in\Sigma\colon \Psi_1(b) \lor \Psi_2(b) 
\end{gather*}
where $\Psi_1(b)$ is given by
\begin{gather*}
  \exists w_1\in\Gamma^* \exists w_2 \in (\Gamma\setminus
  \set{\bar{b}})^*\colon
	w=w_1\cdot \bar{b}\cdot w_2	
	\land \mmap'(b)=\parikh(w_2)(b)
\end{gather*}
and $\Psi_2(b)$ is given by
\begin{gather*}
	w\in(\Gamma\setminus\set{\bar{b}})^* \land \mmap'(b)=\mmap(b)+\parikh(w)(b)
\end{gather*}
The transition relation $\rightarrow$ states that there is a transition from
configuration $(d, \mmap\oplus\multi{\sigma})$ to $(d',\mmap')$ if there is an
execution of handler $\sigma$ that changes the global state from $d$ to $d'$
and operates a sequence of posts and cancel which leaves the task buffer in
state \(\mmap'\). 
A cancel immediately removes every pending instance of the handler
being canceled.
Note that contrary to the case without cancel the order in
which the handler instances are added to and removed from the task buffer does
matter.
%

Finally, let us observe that asynchronous programs with cancel \(
(D,\Sigma\cup\Sigmabar,\Sigma_i,G,R,d_0,\mmap_0)\) 
define a well-structured transition systems
\(( (D\times\multiset{\Sigma},\sqsubseteq),\rightarrow,c_0)\) where
\(\sqsubseteq\) is the ordering used for asynchronous programs:
\(\sqsubseteq\subseteq
(D\times\multiset{\Sigma})\times(D\times\multiset{\Sigma})\) is given by
\(c\sqsubseteq c'\) if{}f \(c.d=c'.d \land c.\mmap\preceq c'.\mmap\).

The {\em safety}, {\em boundedness}, {\em
configuration reachability} and {\em (fair) non termination} problems for
asynchronous programs with cancel are defined as for asynchronous programs
(without cancel).

\subsection{Construction of an equivalent asynchronous program}

Similarly to what we have done for Lem.~\ref{lem:clear} we now give a simpler
yet equivalent semantics to asynchronous programs with cancel. To compute the
task buffer content after the run \(\rho\) of a handler \(h\), the following information
is needed: \( (i) \) the current content of the task buffer, \( (ii) \) the set of
cancelled handlers along \(\rho\), and \( (iii) \) for each handler \(b\in\Sigma\)
the number of posts to \(b\) that are still pending after \(\rho\), that is the
number of posts to \(b\) that have not been subsequently neutralized by a
cancel to \(b\).



Intuitively, our construction uses the following steps.

First, using the construction of Def.~\ref{def:GR}, we eliminate the
need to carry around internal actions $\Sigma_i$ and the regular
grammar $R$.
We get a \(\cfg\) $G^R$ as a result
of this step, and for each context $c = (d, a, d')$, we get the
initialized \(\cfg\) $G^c$ using Def.~\ref{def:gc}.
Remember that in $G^R$ and $G^c$, the alphabet is $\Sigma \cup
\Sigmabar$,
that is, both posts and cancels are visible.

Now, consider a run of $G^c$.
For each handler $a$, we want to remember how many posts to $a$
were issued after the {\em last} call (if any) to cancel $a$,
and also to remember if a cancel to $a$ was issued in the handler along
the execution.
To update the task buffer, for each handler $a$ for which no
cancel was issued, we proceed as before and add all the new posts
of $a$ to the buffer.
For each handler $a$ for which a cancel was called, we first remove
all pending instances of $a$ from the task buffer, and then
add all instances of $a$ posted after the last issuance of a cancel. 
We now give a formal construction that takes any grammar $G$ and
computes a new grammar from which we can get these two pieces of information.


Let $G = (\mathcal{X}, \Sigma \cup\Sigmabar, \prod)$ be a \(\cfg\).
Define the {\em reverse} $r(G) = (\mathcal{X}, \Sigma\cup \Sigmabar,
\overline{\prod})$ as the \(\cfg\) where $\overline{\prod}$ is the
least set containing the production $X \rightarrow a$ for each
$X\rightarrow a$ in $\prod$ and the production $X \rightarrow BA$ for
each
production $X\rightarrow AB$ in $\prod$.
It is easy to see that for each $X\in\mathcal{X}$ and each $w\in
(\Sigma\cup\Sigmabar)^*$,
we have $X \underset{G}{\Rightarrow}^* w$ iff $X
\underset{r(G)}{\Rightarrow}^* w^r$,
where $w^r$ is the reverse of $w$.

Define the regular grammar $\mathcal{C} = (\mathcal{Y}, \Sigma \cup\Sigmabar,
\prod_{\mathcal{Y}})$, where
$\mathcal{Y} = \set{ Y_S \mid S \subseteq \Sigma}$,
and $\prod_{\mathcal{Y}}$ consists of production rules
$Y_S \rightarrow \bar{c} Y_{S \cup\set{c}}$ for each $S\subseteq
\Sigma$,
and $Y_S \rightarrow c Y_S$ for each $S \subseteq \Sigma$.
Intuitively, the regular grammar tracks the set of handlers for which a
cancel has been seen.
Formally, $Y_\emptyset \underset{\mathcal{C}}{\Rightarrow}^* w
Y_S$ implies that for each $\bar{b}\in \Sigmabar$, we have
$\parikh(w)(\bar{b}) > 0$ iff $b \in S$.

Now, we construct a grammar $r(G)\times \mathcal{C} = (\mathcal{Z},
\Sigma, \prod_{\mathcal{Z}})$, where
$\mathcal{Z} = \set{ [Y_{S_1} X Y_{S_2}] \mid Y_{S_1}, Y_{S_2} \in
  \mathcal{Y}, X \in \mathcal{X}}$,
and $\prod_{\mathcal{Z}}$ is the least set of rules such that
\begin{itemize}
	\item if \( (X\rightarrow \varepsilon)\in\prod\) then \([Y_S X Y_S]\rightarrow \varepsilon\) for all \(S\subseteq \Sigma\);
\item if $(X \rightarrow c) \in \prod$, \(c\in\Sigma\cup\Sigmabar\) and $(Y_S \rightarrow c Y_{S'})
  \in \prod_{\mathcal{Y}}$, then
  $([Y_{S} X Y_{S'}] \rightarrow \proj_{\Sigma\setminus S} (c)) \in
  \prod_{\mathcal{Z}}$;
\item if $(X \rightarrow A B) \in \prod$ and then $([Y_{S_0} X Y_{S_2}]
  \rightarrow [Y_{S_0} A Y_{S_1}] [Y_{S_1} B Y_{S_2}]) \in
  \prod_{\mathcal{Z}}$ for each $S_0 \subseteq S_1\subseteq S_2
  \subseteq \Sigma$.
\end{itemize}
Intuitively, a leftmost derivation of the grammar generates
derivations of words in $r(G)$ while tracking
which symbols from $\Sigmabar$ have been seen. 
Additionally, it suppresses all symbols in $\Sigmabar$ as well as all
symbols $c\in\Sigma$ such that $\bar{c}$ has been seen.
Formally, the grammar $r(G) \times \mathcal{C}$ has the following property.
The proof is by induction on the derivation of $w$, similar to Lem.~\ref{lem:gr_correctness}.

\begin{lemma}
For $w\in\Sigma^*$ and $S\subseteq \Sigma$, we have
$[Y_\emptyset X Y_S]
\underset{r(G)\times\mathcal{C}}{\Rightarrow}^* w$ 
iff
there exists $w'\in (\Sigma\cup \Sigmabar)^*$ such that
$X \underset{G}{\Rightarrow}^* w'$ and for each $b\in\Sigma$,
we have 
(1) either $w'\in (\Sigma\cup \Sigmabar \setminus \set{\bar{b}})^*$
and $\parikh(w)(b) = \parikh(w')(b)$ and $b\not \in S$, or
(2) there exists $w'_1\in (\Sigma\cup\Sigmabar)^*$, $w'_2 \in
(\Sigma\cup\Sigmabar\setminus\set{\bar{b}})^*$, $w' = w'_1 \bar{b}
w'_2$,
and $\parikh(w)(b) = \parikh(w'_2)(b)$ and $b\in S$.
\label{lem:grammar}
\end{lemma}

Lem.~\ref{lem:grammar}, when instantiated with the grammar $G^c$,
provides the following corollary.

\begin{corollary}
	Let $\ap$ be an asynchronous program with cancel, and let
        $\mmap, \mmap'\in \multiset{\Sigma}$.
	For \(c=(d_1,\sigma,d_2)\in\mathfrak{C}\), let $G^c$ be
        defined as in Def.~\ref{def:gc} (with \(\Sigma\) replaced by \(\Sigma\cup\Sigmabar\)). 
				The following statements are equivalent:
	\begin{enumerate}
		\item \((d_1,\mmap\oplus \multi{\sigma})\overset{\sigma}{\rightarrow} (d_2,\mmap')\)
                \item \( \exists w \in \Sigma^*\colon 
                  [Y_\emptyset [d_1  X_\sigma d_2] Y_S]
                  \underset{r(G^c) \times \mathcal{C}}{\Rightarrow}^*
                  w\) and for all $b\in \Sigma$, we have
                  \[ \mmap'(b) = \begin{cases} \mmap(b)
                    + \parikh(w)(b) & \mbox{if }b\not \in S\\
                    \parikh(w)(b) & \mbox{if }b \in S
                    \end{cases}
                    \]
%
	\end{enumerate}
	\label{lem:equivalentprogramwithcancel}
\end{corollary}
\begin{proof}
We have
\begin{align*}
& (d_1,\mmap\oplus\multi{\sigma}) \overset{\sigma}{\rightarrow} (d_1,
\mmap') & \\
\mbox{if{}f } & \exists w\in (\Sigma\cup\Sigmabar)^*\colon [d_1X_\sigma d_2]
\underset{G^c}{\Rightarrow^*} w \land
\forall b\in \Sigma\colon \Psi_1(b) \vee \Psi_2(b) & \mbox{ def.\ of
  $\overset{\sigma}{\rightarrow}$ and $G^c$}\\
\mbox{if{}f } & 
\left(\begin{array}{c}
\exists w\in \Sigma^*\exists S \subseteq \Sigma\colon [Y_\emptyset [d_1 X_\sigma d_2] Y_S]
\underset{r(G^c)\times\mathcal{C}}{\Rightarrow^*} w\\
\mbox{and }\\
\forall b\in\Sigma\colon \mmap'(b) = \begin{cases}
\parikh(w)(b) & \mbox{ if $b\in S$ and } \\
\mmap(b) + \parikh(w)(b)   & \mbox{ otherwise}
\end{cases}
\end{array}\right)& \mbox{Lem.~\ref{lem:equivalentprogramwithcancel}}
\end{align*}
\end{proof}

\subsection{$\pn$ with reset arcs}\label{sec:pnrt}

Let us now introduce an extension of the \(\pn\) model which will serve to
model the semantics of asynchronous programs with cancel.

\begin{definition}
	A Petri net with reset arcs, \(\pnr\) for short, is a tuple $(S,T,F=\tuple{I,O,Z},\mmap_0)$ where $S$, $T$
	and $F$ are defined as for $\pn$ except that $F$ is extended with a mapping
	\(Z\) such that \(Z(t)\subseteq S\) for each \(t\in T\).
        As for
	$\pn$, $\mmap_0\in\multiset{S}$ defines the initial marking.
\end{definition}

\noindent %
{\bf Semantics.} %
Given a tuple $(S,T,F,\mmap_0)$, and a marking $\mmap\in\multiset{S}$, a
transition $t\in T$ is \emph{enabled} at $\mmap$,
written $\mmap\fire{t}$, if $I(t)\preceq \mmap$.  We write
$\mmap\fire{t}\mmap'$ if transition $t$ is enabled at $\mmap$ and its
\emph{firing} yields to marking $\mmap'$ defined as follows:
\begin{enumerate}
	\item Let $\mmap_1$ be such that $\mmap_1\oplus I(t)=\mmap$.
	\item Let $\mmap_2$ be such that
	$\displaystyle\mmap_2(p)=\begin{cases}0 &\text{if $p\in Z(t)$}\\\mmap_1(p)
						&\text{else.}\end{cases}$
	\item $\mmap'$ is such that $\mmap'=\mmap_2\oplus O(t)$.
\end{enumerate}
The semantics as well as the boundedness and coverability problems naturally
follows from their counterpart for $\pn$.
Note that if $Z(t)=\emptyset$ for each $t\in T$, then $N$ reduces to a $\pn$.

\begin{theorem}{\cite{DFS98}}
The coverability problem for $\pnr$ is decidable.
The boundedness problem and the reachability problem for $\pnr$ are both undecidable.
	\label{thm:pntcoverdecid}
\end{theorem}

\subsection{$\pnr$ semantics of asynchronous programs with cancel}

\begin{definition}
Let $c=(d_1,a,d_2)\in\mathfrak{C}$, and let $r(G^c) \times \mathcal{C}
= (\mathcal{Z},\Sigma, \prod_{\mathcal{Z}})$.
Define $k=|\mathcal{Z}|$ and the
$\pnr$ $N^{\nsim}_c=(S^{\nsim}_c,T^{\nsim}_c,F^{\nsim}_c)$ such that:
		\begin{itemize}
			\item \(S^{\nsim}_c=\set{(\mathit{begin},c),(\mathit{end},c)}\cup \mathcal{Z}\cup \set{(\$,c)}
				\cup\Sigma\);
			\item the sets $T^{\nsim}_c$ and $F^{\nsim}_c$ are such that $t\in T^{\nsim}_c$ if{}f one of the following holds
				\begin{align*}
					F^{\nsim}_c(t)&=\tuple{\multi{(\mathit{begin},c)},\multi{[Y_{\emptyset} [d_1X_ad_2] Y_{S_1}]}\oplus\multi{(\$,c)^{k}}, S_1} &\text{for each \(S_1\subseteq \Sigma\)}\\
					F^{\nsim}_c(t)&=\tuple{\multi{X,(\$,c)},\multi{Z,Y},\emptyset}&\text{\((X \rightarrow Z\cdot Y)\in\prod_{\mathcal{Z}}\)}\\
					F^{\nsim}_c(t)&=\tuple{\multi{X},\parikh(\sigma)\oplus\multi{(\$,c)},\emptyset}
					&\text{\((X \rightarrow \sigma)\in\prod_{\mathcal{Z}}\)}\\
				F^{\nsim}_c(t)&=\tuple{\multi{(\$,c)^{k+1}},\multi{(\mathit{end},c)},\emptyset}
				\end{align*}
		\end{itemize}
		Finally, define \(\mathcal{N}^{\nsim}=\set{N_c^{\nsim}}_{c\in\mathfrak{C}}\).
\end{definition}

The following lemma is proved similar to Lem.~\ref{lem:construction_final}.

\begin{lemma}
	Let $\ap$ be an asynchronous program with cancel
	and let \(d,d'\in D\) and \(\mmap,\mmap'\in\multiset{\Sigma}\).
	Define \(c=(d,\sigma,d')\in\mathfrak{C}\), we have:
	\[(d,\mmap)\overset{\sigma}{\rightarrow} (d',\mmap') \text{ if{}f } \exists
	w\in (T_c^{\nsim})^* \colon
	\bigl(\multi{(\mathit{begin},c)}\oplus\mmap\bigr)\fire{w}_{N_c^{\nsim}}\bigl(\multi{(\mathit{end},c)}\oplus\mmap'
	\bigr)\enspace .\]
	\label{lem:constructioncancel}
\end{lemma}
%

\begin{construction}\label{constr:parikhiireset}
	Let \(\ap=(D,\Sigma\cup\Sigmabar,\Sigma_i,G,R,d_0,\mmap_0)\) be an asynchronous program with cancel.
	Define $(N_{\ap},\mmap_{\imath})$ to be an initialized $\pnr$  where
	(1) $N_{\ap}=(S_{\ap}, T_{\ap}, F_{\ap})$ is given as follows:
		\begin{itemize}
			\item the set $S_{\ap}$ is given by $D\cup\Sigma\cup \bigcup_{c\in\mathfrak{C}} S^{\nsim}_c$

		\item the set $T_{\ap}$ of transitions is given by
			\( \bigcup_{c\in\mathfrak{C}} \bigl(\set{t_c^{<}}\cup T^{\nsim}_c\cup \set{t_c^{>}}\bigr)\)

		\item $F_{\ap}$ is such that for each $c=(d_1,a,d_2)\in\mathfrak{C}$ we have
				\begin{align*}
					F_{\ap}(t_c^{<})&=\tuple{\multi{d_1, a},\multi{(\mathit{begin},c)},\emptyset}\\
					F_{\ap}(T^{\nsim}_c)&=F^{\nsim}_{c}(T^{\nsim}_c)\\
					F_{\ap}(t_c^{>})&=\tuple{\multi{(\mathit{end},c)}, \multi{d_2},\emptyset}
				\end{align*}
		\end{itemize}
		and (2) $\mmap_{\imath}=\multi{d_0}\oplus \mmap_0$.
\end{construction}

From the previous lemma, it follows that.
\begin{lemma}
	Let $\ap$ be an asynchronous program with cancel and let \((N_{\ap},\mmap_{\imath})\) be
	an initialized $\pn$ as given in Constr.~\ref{constr:parikhiireset}. We have
	\((d,\mmap)\) is reachable in $\ap$ if{}f $\multi{d}\oplus\mmap$ is reachable in $N_{\ap}$ from \(\mmap_{\imath}\).
	\label{lem:equivpncancel}
\end{lemma}
%

\subsection{Model checking}

We now summarize the status of model checking asynchronous programs with cancel.

\begin{theorem}{}
\label{thm:canceldecid}
\begin{enumerate}
\item
The safety (global state reachability) problem for asynchronous programs with cancel is decidable.

\item The configuration reachability problem for asynchronous programs with cancel is undecidable.

\item
The boundedness problem for asynchronous programs with cancel is undecidable.
\end{enumerate}
\end{theorem}
\begin{proof}
Part (1) of Theorem.~\ref{thm:canceldecid} follows from Thm.~\ref{thm:pntcoverdecid} and
Lem.~\ref{lem:equivpncancel}.

To show configuration reachability and boundedness are
undecidable, we use a reduction similar to what we have previously seen at
Fig.~\ref{fig-pnbound2aabound} for $\pn$.
We reduce the reachability and boundedness problems for $\pnr$ to the configuration
reachability and boundedness problems for asynchronous programs with cancel, respectively.
The reachability and the boundedness problems for
$\pnr$ are both undecidable \cite{DFS98}.
Our reduction from the boundedness of \(\pnr\) is given at
Fig.~\ref{fig-asynccancelboundednessisPNRboundedness}. We omit the details, which are
similar to the construction for \(\pn\). The reduction for configuration reachability 
is similar.
\end{proof}

\begin{figure*}[t]
  \centering
  \hspace{\stretch{1}}
    \begin{minipage}{.5\textwidth}
\begin{tabbing}
\ \ \ \ \=\ \ \ \=\ \ \ \=\\
{\it global} {\tt st} = $(\varepsilon,\varepsilon)$;\\
\\
{\tt runPN} () \{\\
\>	if {\tt st} \(\in (T\cup\set{\varepsilon})\times\set{\varepsilon}\) \{\\
\>\>    pick $t\in T$ non det.;\\
\>\>    {\tt st} {\tt =} $(t,\hat{I}(t))$;\\
\>  \}\\
\>  {\tt post} {\tt runPN}();\\
\}\\
\\
Initially: $\mmap_{\imath} \oplus\multi{\mathtt{runPN}}$
\end{tabbing}
    \end{minipage}%
  \hspace{\stretch{1}}
    \begin{minipage}{.45\textwidth}
\begin{tabbing}
\ \ \ \ \=\ \ \ \=\ \ \ \=\ \ \ \=\ \ \ \=\\
$p'$() \{ // for \(p'\in S\) \\
\>  if {\tt st} {\tt ==} $(t, p'\cdot w)$ \{\\
\>\>    {\tt st} {\tt =} $(t, w)$;\\
\>\>  if $w$ {\tt ==} $\varepsilon$ \{\\
\>\>\>      for each $p\in S$ do \{\\
\>\>\>\>       if \(p\in Z(t)\) \{\\
\>\>\>\>\>        {\tt cancel} $p$();\\
\>\>\>\>       \}\\
\>\>\>\>       if \(O(t)(p)>0\) \{\\
\>\>\>\>\>        {\tt post} $p$();\\
\>\>\>\>       \}\\
\>\>\>     \}\\
\>\>    \}\\
\>  \} else \{\\
\>\>    {\tt post} $p'$();\\
\>  \}\\
\}
\end{tabbing}
    \end{minipage}%
  \hspace{\stretch{1}}
	\caption{Let $N=(S,T,F=\tuple{I,O,Z},\mmap_{0})$ be an initialized \(\pnr\)
	such that \(\forall t\in T\colon\card{I(t)}>0\).
	$N$ is unbounded (that is $\fire{\mmap_{\imath}}$ is infinite) if{}f
	the asynchronous program is unbounded.}
  \label{fig-asynccancelboundednessisPNRboundedness}
\end{figure*}

We now show undecidability results when it comes to determine properties
related to infinite runs.
Our proofs use undecidability results for counter machines, which we now introduce.

\begin{definition}
	A \emph{$n$-counter machine} $C$ ($\ncm$ for short), is a tuple
	$\tuple{\set{c_i}_{1\leq i\leq n},L,\instr}$ where:
	\begin{itemize}
		\item each $c_i$ takes its values in $\mathbb{N}$;
		\item $L=\set{l_1,\ldots, l_u}$ is a finite non-empty
                  set of locations;
		\item $\instr$ is a function that labels each location $l\in L$ with an instruction that has one of the following forms:
			\begin{itemize}
				\item $l\colon c_j\assign c_j+1 ;
                                  \texttt{ goto } l'$ where $1\leq
                                  j\leq n$ and $l'\in L$, this is
                                  called an {\em increment}, and we
					define
                                        $\typeinst(l)=\tuple{\inc_j, l'}$;
				\item $l\colon c_j\assign c_j-1 ;
                                  \texttt{ goto } l'$ where $1\leq
                                  j\leq n$ and $l'\in L$, this is
                                  called a {\em decrement}, and we
					define $\typeinst(l)=\tuple{\dec_j,l'}$;
				\item $l\colon \texttt{if } c_j=0
                                  \texttt{ then goto } l' \texttt{
                                    else goto } l''$ where $1\leq
                                  j\leq n$ and $l',l''\in L$, this is
                                  called a {\em zero-test}, and we define $\typeinst(l)=\tuple{\zerotest_j,l',l''}$;
			\end{itemize}
	\end{itemize}
\end{definition}

We define $\twocm$ and $\threecm$ as the class of $2$-counter and $3$-counter
machines, respectively.

\noindent %
{\bf Semantics.} 
The instructions have their usual obvious semantics, in particular,
decrement can only be done if the value of the counter is strictly greater
than zero.

A \emph{configuration} of an $\ncm$ $\tuple{\set{c_1,\ldots,c_n},L,\instr}$ is a
tuple $\tuple{\loc,v_1,v_2,\ldots,v_n}$ where $\loc\in L$ is the value of the program
counter and, $v_1, \ldots, v_n$ are positive integers that gives the values of
counters $c_1,\ldots,c_n$, respectively.
We adopt the convention that every \(\ncm\) is such that \(L\) contains
a special location \(l_1\) called the \emph{initial location}.

A \emph{computation} $\gamma$ of an $\ncm$ is
a finite sequence of configurations $\tuple{\loc^1,v^1_1,\ldots,v^1_n},
\tuple{\loc^2,v^2_1,\ldots,v^2_n},\dots,\tuple{\loc^r,v^r_1,\ldots,v^r_n}$
such that the following conditions hold.
$(i)$ ``Initialization'': $\loc^1=l_1$ and for each $i\in\set{1,\ldots,n}$,
we have $v^1_i=0$. That is, a computation
starts in $l_1$ and all counters are initialized to $0$.
$(ii)$ ``Consecution'': for
each $i\in\nats$ such that $1\leq i\leq \card{\gamma}$ we have that
$\tuple{\loc^{i+1},v^{i+1}_1,\ldots,v^{i+1}_n}$ is the configuration obtained from
$\tuple{\loc^i,v^i_1,\ldots,v^i_n}$ by applying instruction $\instr(\loc_i)$.
A configuration $c$ is \emph{reachable} if there
exists a finite computation $\gamma$ whose last configuration $c$.
A location $\ell\in L$ is reachable if there exists a reachable
configuration
$\tuple{\ell, v_1,\ldots,v_n}$ for some $v_1,\ldots, v_n\in\nats$.

Given an $\ncm$ $C$ and $F\subseteq L$, the
\emph{reachability problem} asks if some $\ell\in F$ is reachable.
If so, we say \(C\) \emph{reaches} \(F\).

\begin{theorem}{\cite{Min67}}
	The reachability problem for $\ncm$ is undecidable for \(n\geq 2\).
	\label{thm:2cmundec}
\end{theorem}

\begin{figure*}[t]
  \centering
  \hspace{\stretch{1}}
    \begin{minipage}{.65\textwidth}
\begin{tabbing}
\ \ \ \ \=\ \ \ \=\ \ \ \=\ \ \ \=\\
global {\tt loc}=\(l_1\);\\
\\
{\tt main}() \{\\
\>  while(*)\\
\>\> {\tt post} \(I\)();\\
\}\\
\\
\(c_j\)() \{ // for \(j\in\set{1,2,3}\)\\
\> if \(\typeinst({\tt loc})==\tuple{\dec_j,l'} \) \{\\
\>\> {\tt loc}=\(l'\);\\
\>\> {\tt post} \(I\)();\\
\> \} else if \\
\>\>\(\typeinst({\tt loc})==\tuple{\zerotest_j,l',l''}\) \{\\
\>\> {\tt loc}=\(l''\);\\
\>\> {\tt post} \(c_j\)();\\
\> \} else\\
\>\> {\tt loc}=\(\bot\);\\
\}\\
\\
Initially: \(\multi{{\tt main}}\)
\end{tabbing}
    \end{minipage}%
  \hspace{\stretch{1}}
    \begin{minipage}{.35\textwidth}
\begin{tabbing}
\ \ \ \ \=\ \ \ \=\\
\(I\)() \{\\
\>  if \(\typeinst({\tt loc})==\tuple{\inc_j,l'}\) \{\\
\>\> {\tt loc}=\(l'\);\\
\>\> {\tt post} \(c_j\)();\\
\>  \} else if \\
\>\>\(\typeinst({\tt loc})==\tuple{\zerotest_j,l',l''}\) \{\\
\>\> {\tt loc}=\(l'\);\\
\>\> {\tt cancel} \(c_j\)();\\
\>\> {\tt post} \(I\)();\\
\>  \} else\\
\>\>  {\tt loc}=\(\bot\);\\
\}
\end{tabbing}
    \end{minipage}%
  \hspace{\stretch{1}}
	\caption{Let \(C'=(\set{c_1,c_2,c_3}, L,\instr)\) be the \(\threecm\) defined
	upon a reachability problem instance for \(\twocm\),
	the above asynchronous program with cancel has an infinite computation if{}f
	\(C'\) has an infinite bounded computation.
	In the above program, whenever \texttt{loc} equals \(\bot\) then every conditional fails.}
  \label{fig:simpleaacanceltermi}
\end{figure*}

\begin{theorem}{}
	Determining if an asynchronous program with cancel has an infinite run is
	undecidable.
	\label{thm:fairetermi}
\end{theorem}
\begin{proof}
Our proof follows the proof of \cite{EFM99} which reduces the termination
of broadcast protocols to the reachability problem for \(\ncm\).

We first start with some additional notions on counter machines.
A configuration $\tuple{\loc,v_1,v_2,\dots,v_n}$ of an $\ncm$ is 
$k$\emph{-bounded} if $\sum_{i=1}^{n}v_i\leq k$. A computation $\gamma$ is
$k$-bounded if all its configurations are $k$-bounded,
and \emph{bounded} if it is $k$-bounded for some positive integer $k$.

Consider an instance of the reachability problem of a $\twocm$ given 
by \(C=\tuple{\set{c_1,c_2},L,\instr}\) and \(F\subseteq L\).
Without loss of generality, we assume that $l_1$ does not have an
``incoming edge'' in $C$.
Define $C'$ to be a \(\threecm\) that behaves as follows. $C'$ simulates $C$ on
counters $c_1$ and $c_2$ and increases $c_3$ by \(1\) after each step of
simulation.  If $C$ reaches some location in $F$, then $C'$ goes back to its initial
configuration $\tuple{l_1,0,0,0}$.  We make the following two observations about
$C'$:

\begin{itemize}
	\item $C'$ has an infinite bounded computation if{}f $C$ reaches $F$. Because
		after each step $C'$ increments counter $c_3$, the only infinite bounded
		computation of $C'$, if any, corresponds to the infinite iteration of a
		run of $C$ that reaches $F$.
	\item In every infinite bounded computation of $C'$, the initial configuration
		$\tuple{l_1,0,0,0}$ occurs infinitely often.
\end{itemize}

We can simulate $C'=\tuple{\set{c_1,c_2,c_3},L,\instr}$ in a weak sense by
using an asynchronous program with cancel \(\ap\) given at
Fig.~\ref{fig:simpleaacanceltermi}.
The simulation uses procedures $c_1$, $c_2$, and $c_3$ to simulate
decrements of counters as well as zero-tests where the ``else
branch'' is taken.
It additionally uses a procedure $I$ to simulate increments to
variables,
as well as the ``then branch'' for a zero-test.
The location $\bot$ is a special ``halt'' location with no
instructions (so the simulation eventually terminates once the
location is set to $\bot$).

We call a simulation \emph{faithful} if whenever the then-branch of a zero
test is executed, there are no pending instances of handler $c_j$ (and
thus the cancel is a no-op).
A simulation may not be faithful because the dispatch of handler $I$ amounts to guess that the then-branch is
taken, and cancels any pending instances of handler \(c_j\).
If there were pending instances of $c_j$, this guess is wrong, but
these instances get removed anyway by the cancel.
In that case we say that $\ap$ \emph{cheats}.


We prove that if $C$ reaches $F$, then by the above observation $\ap$ has an infinite run.  
If $C$ reaches $F$, then $C'$ has an infinite bounded computation \(\gamma\), which
		iterates infinitely often a computation of $C$ that reaches $F$.  By
		definition of bounded computation, there exists \(b\geq 0\) such that
		\(\gamma\) is \(b\)-bounded.  Let \(\rho\) be a run of $\ap$ that initially
		executes ``{\tt post} \(I\)()'' $b$ times and then
		 faithfully simulates \(\gamma\).
		Since this is a faithful simulation, each time a ``{\tt cancel} \(c_i\)'' (for \(i\in\set{1,2}\))
		statement is executed, there is no pending instance of handler \(c_i\) to
		remove. 
		Since \(\rho\) can simulate every step of \(\gamma\), it is infinite.

We now prove that if \(\ap\) has an infinite run, then $C$ reaches
$F$. 
Here, we have to take into account possible cheating in the
simulation.
Let \(\rho\) be an
infinite run of \(\ap\). 
Notice that in this run, the variable $\mathtt{loc}$ can never be set
to $\bot$ (since any run of $\ap$ where $\mathtt{loc} = \bot$
eventually terminates.
Suppose in this run, the statement ``{\tt
  post}\(I\)()'' was executed $b$ times in {\tt main}.
After the execution of {\tt main}, the number of pending handlers is
always at most $b$, and thus the execution encodes a $b$-bounded run
of the counter machine.
Moreover, the number of pending handlers only decreases if there is a cheat (that
is, some pending handler $c_j$ is canceled).
Thus, the infinite execution \(\rho\) can have only finitely many cheats.
Take a suffix of $\rho$ containing no cheats. 
It corresponds to a bounded infinite simulation $\gamma$ of $C'$. 
Now recall that every infinite
bounded run of $C'$ contains infinitely many initial configurations. 
So some suffix $\gamma'$ of $\gamma$ is an infinite computation of $C'$. 
Thus, $C$ reaches $F$.
\end{proof}

It can also be shown that the fair non termination and fair starvation problem
for asynchronous program with cancel are also undecidable.  Let us sketch the
main intuitions here. For the fair non termination problem, it suffices to
modify the \(\threecm\) \(C'\) as follows. In the initial configuration
\(\tuple{l_1,0,0,0}\), instead of simulating \(C\), \(C'\) first increments and then
decrement each counter \(c_i\) for \(i\in\set{1,2,3}\). Then \(C'\) simulates
\(C\) as given above. Observe that this modification preserves the correctness of the above
proof. Let us now turn to the asynchronous program with cancel \(\ap\)
simulating this updated \(C'\). We conclude from the above modification that if
\(\ap\) simulates the bounded infinite run of \(C'\) faithfully then the run is
fair because a faithful simulation requires the dispatch of every handler
(i.e. \(c_1(), c_2(), c_3()\) and \(I()\)). Therefore the infinite run is fair.

For the fair starvation problem, let \(k\) denote the value such that there is a
\(k\)-bounded infinite computation in \(C'\).
We will now show there exists a fair infinite run that starves handler \(I()\).
In this run, {\tt main} posts at least \(k+2\) instances of handler \(I()\). 
This will ensure that after executing the {\tt main} procedure there are at least
\(2\) pending instances of \(I()\) along the fair infinite run and we are done.

\begin{theorem}
	Determining if an asynchronous program  with cancel \(\ap\) has a fair infinite
	run or determining if \(\ap\) fairly starves some \(a\in\Sigma\) is
	undecidable.
\end{theorem}

\subsection{Asynchronous Programs with Cancel and Test}

Our final results investigate the decidability of natural extensions to
asynchronous programs with cancel, where additionally, the program can test for the absence of
pending instances to a particular handler \(p\). 
We model an additional instruction 
{\tt assertnopending} \(p()\) that succeeds if there is no pending
instance of $p$.
Here, we show that safety verification becomes undecidable as well.
Our proof reduces the coverability problem for
an extension of $\pnr$ where we additionally allow one transition whose enabling
condition is augmented by requiring the absence of token in a given place.
We call this transition a transition with \emph{inhibitor arc}.

We first introduce an extension of \(\pnr\) with one transition with inhibitor arc.
\begin{definition}
	A reset net with one inhibitor arc $N$ ($\pnri$ for short) is a tuple
	$\tuple{S,T,F=\tuple{I,O,Z},!,\mmap_0}$ where
	$\tuple{S,T,F=\tuple{I,O,Z},\mmap_0}$ is a \(\pnr\) and $!\in (T\times S)$.
\end{definition}

We know define the semantics for $\pnri$ by extending the one for $\pnr$.

\noindent %
{\bf Semantics.} %
Given a $\pnri$ $N=\tuple{S,T,F,!,\mmap_0}$, and a marking $\mmap$ of $N$, a
transition $t\in T$ is \emph{enabled} at $\mmap$, written $\mmap\fire{t}$, if
(1) $I(t)\preceq\mmap$ and (2) \(!=(t,p)\) implies \(\mmap(p)=0\). We write
$\mmap\fire{t}\mmap'$ if transition $t$ is enabled at $\mmap$ and its
\emph{firing} yields to marking $\mmap'$ defined as in Sect.~\ref{sec:pnrt}.

The coverability problem for $\pnri$ naturally follows from the definition
for $\pnr$.
The following result, due to Laurent Van Begin,
shows that coverability is undecidable in this model.


\begin{theorem}{}
	The coverability problem for $\pnri$ is undecidable.
	\label{thm:coverpntiundec}
\end{theorem}
\begin{proof}
	Our proof reduces the reachability problem for $\twocm$ to the coverability
	problem for $\pnri$.  We consider here a particular case of the reachability
	problem which asks whether a particular control location, e.g. \(l_f\), with
	null counter values is reachable (Is \(\tuple{l_f,0,0}\) reachable?).  This
	problem is known to be undecidable. 

	Fix an instance \((C=\tuple{\set{c_1,c_2},L,\instr},l_f)\) 
	of that problem where \(C\) is the \(\twocm\) and \(l_f\in L\) is a control
	location of \(C\).

	We define the $\pnri$ $N=(S,T,F=\tuple{I,O,Z},!,\mmap_0)$ such that
	\(N\) simulates \(C\) in a weak sense we define below.
	\begin{itemize}
		\item $S= L\cup\set{c_1,c_2}\cup\set{\cnt,\tocov}$
		\item \(T\) and \(F\) are such that \(t\in T\) if{}f one of the following holds:
			\begin{itemize}
				\item 
					$F(t)=\tuple{\multi{l},\multi{c_j,l',\cnt},\emptyset}$
					where 
					\(\typeinst(l)=\tuple{\inc_j,l'}\);
				\item 
					$F(t)=\tuple{\multi{c_j,l,\cnt},\multi{l'},\emptyset}$
					where 
					\(\typeinst(l)=\tuple{\dec_j,l'}\);
				\item 
					$F(t)=\tuple{\multi{l},\multi{l'},\set{c_j}}$
					where 
					\(\typeinst(l)=\tuple{\zerotest_j,l',l''}\);
				\item $F(t)=\tuple{\multi{l,c_j},\multi{l'',c_j},\emptyset}$
					where 
					\(\typeinst(l)=\tuple{\zerotest_j,l',l''}\);
				\item \(F(t)=\tuple{\multi{l_f},\multi{\tocov},\emptyset}\).
			\end{itemize}
	\item $!=(t,\cnt)$ such that
		\(F(t)=\tuple{\multi{l_f},\multi{\tocov},\emptyset}\) namely a token
		is produced in \(\tocov\) provided \(l_f\) contains some token and $\cnt$
		does not;
	\item $\mmap_0=\multi{l_1}$.
	\end{itemize}

	Define $(N,\multi{\tocov})$ to be an instance of the coverability problem
	for $\pnri$. The rest of the proof shows that \(\multi{\tocov}\) is coverable
	if{}f \(C\) reaches the configuration \(\tuple{l_f,0,0}\).

	Intuitively, the following property is maintained by \(N\): 
	as long as \(N\) simulates faithfully \(C\) the place \(\cnt\) 
	holds as many tokens as the sum of tokens in \(c_1\) and \(c_2\); 
	once \(N\) does not
	faithfully simulate \(C\) we have that \(\cnt\) holds strictly more
	tokens than \(c_1\) and \(c_2\). 

	The definition of \(\mmap_0\) shows that initially
	\(\mmap_0(\cnt)=\mmap_0(c_1)+\mmap_0(c_2)=0\), that is \(\cnt\) holds as many
	tokens as \(c_1\) and \(c_2\).  Moreover the definition of \(N\) shows that
	whenever a transition which resets \(c_j\) \(j=1,2\) is fired and removes at
	least one token from \(c_j\) then \(\cnt\) holds more tokens than \(c_1\) and
	\(c_2\). This will reflect that \(N\) incorrectly simulated \(C\). In fact,
	if a transition resets \(c_j\) and removes at least one token from it then we
	find that some \(\zerotest\) instruction was inaccurately simulated because
	the ``then'' branch was taken while the counter tested for \(0\) contained a
	token.  Therefore a token was removed from \(c_j\). Observe that once a reset
	transition of \(N\) has removed a token from \(c_1\) or \(c_2\) then from
	this point on \(\cnt\) holds strictly more than the sum of tokens in \(c_1\)
	and \(c_2\).

	Therefore, given a sequence of transitions $w\in T^*$, such that
	$\mmap_0\fire{w}\mmap$, we have $\mmap(\cnt)=\mmap(c_1)+\mmap(c_2)$ if{}f
	each occurrence of a transition \(t\) such that \(Z(t)=\set{c_j}\) along
	$w$ removes no token from \(c_j\). We thus interpret \(w\) as an accurate
	simulation of \(C\).


	Now suppose \(\tuple{l_f,0,0}\) is reachable in \(C\) through some
	computation \(\gamma\). By accurately simulating \(\gamma\) in \(N\) we find
	that a marking with some tokens in \(l_f\) and no tokens elsewhere is
	reachable, hence that \(\multi{\tocov}\) is coverable.  The other direction is proven by contradiction.

	Assume that \(\tuple{l_f,0,0}\) is not reachable in \(C\) but
	\(\multi{\tocov}\) is coverable in \(N\). Hence there exists \(w\in T^*\)
	such that \(\mmap_0\fire{w}\mmap\), \(\mmap(l_f)\geq 1\) and
	\(\mmap(\cnt)=0\). It follows that
	\(\mmap(c_1)+\mmap(c_2)=0=\mmap(\cnt)\). But we showed above that in this
	case \(w\) is a precise simulation of a computation in \(C\), hence a
	contradiction. 
	
	In fact, whenever \(N\) does not faithfully
	simulate \(C\), every marking \(\mmap\) reachable from this point is such that
	\(\mmap(c_1)+\mmap(c_2)<\mmap(\cnt)\), hence that \(\mmap(\cnt)>0\) since the
	minimum value for \(\mmap(c_1)+\mmap(c_2)\) is 0. This means \(\cnt\) can
	never be emptied, hence that the enabling condition expressed by \(!\) can
	never be satisfied, and finally that \(\multi{\tocov}\) can never be marked.
\end{proof}

We finally obtain the following negative result for the safety problem of
asynchronous programs with cancel and a test for the absence of pending instances
to a particular hander \(p\). Recall that boundedness, configuration reachability, and liveness properties are
undecidable already for the more restricted class without testing for
the absence of a handler.

\begin{figure*}[h]
  \centering
    \begin{minipage}{.5\textwidth}
			\begin{tabbing}
			\ \ \ \ \=\ \ \ \=\ \ \ \=\\
			{\it global} {\tt st} = $(\varepsilon,\varepsilon)$;\\
			\\
			{\tt runPN} () \{\\
			\>	if {\tt st} \(\in (T\cup\set{\varepsilon})\times\set{\varepsilon}\) \{\\
			\>\>    pick $t\in T$ non det.;\\
			\>\>    {\tt st} {\tt =} $(t,\hat{I}(t))$;\\
			\>\>    if \(!=(t,p)\) \{\\
			\>\>\>    {\tt assertnopending} \(p()\);\\
			\>\>    \}\\
			\>  \}\\
			\>  {\tt post} {\tt runPN}();\\
			\}\\
			\\
			Initially: $\mmap_{\imath} \oplus\multi{\mathtt{runPN}}$
			\end{tabbing}
\end{minipage}%
	\caption{Let $N=(S,T,F=\tuple{I,O,Z},!,\mmap_{0})$ be an initialized \(\pnri\)
	such that \(\forall t\in T\colon\card{I(t)}>0\).
	$N$ enables some given \(t_f\) if{}f
	\({\tt st}=(t_c,\varepsilon)\) is reachable in \(\ap\).}
  \label{fig-apcanceltestsafetyISpnticover}
\end{figure*}

\begin{lemma}
The safety problem for asynchronous programs with cancel and test for absence of pending instances is undecidable.
	\label{lem:apcanceltestundecic}
\end{lemma}
\begin{proof}
	We reduce from coverability problem for $\pnri$ which has been shown to be
	undecidable at Thm.~\ref{thm:coverpntiundec}.  The reduction is similar to
	the one given at Fig.~\ref{fig-asynccancelboundednessisPNRboundedness} only
	that {\tt runPN} has to be slightly modified in order take the augmented
	enabling condition of \(\pnri\) into account. As in
	Sect.~\ref{sec:safetyandboundedness} we assume w.l.o.g.\ that instead of asking if
	some given marking \(\mmap\) is such that
	\(\ucl{\mmap}\in\fire{\mmap_{\imath}}_{N}\) where \(N\) is a \(\pnri\), we
	equivalently asks if there exists a marking
	\(\mmap\in\fire{\mmap_{\imath}}_{N}\) for a \(\pnri\) \(N\) such that
	\(\mmap\) enables some given transition \(t_c\), namely \(\mmap\fire{t_c}\).
	We thus obtain that there exists \(\mmap\in\fire{\mmap_{\imath}}\) such that
	\(\mmap\fire{t_f}\)	if{}f \({\tt st}=(t_c,\varepsilon)\) is reachable in
	\(\ap\).  The resulting code for {\tt runPN} is given at
	Fig.~\ref{fig-apcanceltestsafetyISpnticover}.
\end{proof}

\section{Conclusion}\label{sec:conclusion}

Asynchronous programming is ubiquitous in computing systems.
The results in this paper provide a fairly complete theoretical characterization of the safety and liveness
verification problems for this model.
Initial implementations for safety verification of asynchronous programs were reported in \cite{jmpopl07}.
One interesting direction will be to apply tools for coverability analysis of \(\pn\) to this problem,
using the reduction outlined in this paper.
For liveness verification, the \(\pn\) reachability lower bound is somewhat disappointing.
It will be interesting to see what heuristic approximations can work well in practice.

Since our initial work \cite{gmr09}, there have been
several other related results.  
The problem of whether an asynchronous program
is simulated by or simulates a finite state machine is shown to be decidable in
\cite{cv-tcs09}. The authors also show how to solve the control state
maintainability problem which asks whether an asynchronous program has
an infinite (or terminating) run such that each of its state belongs to a given
upward closed set of configurations.  Safety verification was shown to be
decidable for a model augmenting asynchronous programs with priorities (and
letting higher priority handlers interrupt lower priority ones) in
\cite{ABT08}.  Safety verification was shown to be undecidable for a natural
extension of asynchronous programs with timing \cite{gm09}.  A model of
asynchronous programs in which emptiness of a fixed subset of handlers can be
checked has been proposed in the Linux kernel (see
\url{http://lwn.net/Articles/314808/}).
For this model, safety and boundedness are decidable.
This follows from recent results in \cite{AbdullaM09} (for safety) and \cite{FinkelS10} (for boundedness).
As far as we known, the decidability of termination is still open.
When extended with cancellation of handlers, safety verification becomes undecidable, using Thm.~\ref{thm:coverpntiundec}.

\appendix
\section{APPENDIX: Construction of the grammar $G^R$}\label{sec:synsemproduct} 

\begin{definition}\label{def:Gr}
	Given a \(\cfg\) $G=(\mathcal{X},\Sigma\dotcup\Sigma_i,\prod)$ 
	and a regular grammar $R=(D,\Sigma\dotcup\Sigma_i,\delta)$, define $G^r=(\mathcal{X}^{r},\Sigma\dotcup\Sigma_i,\prod^{r})$
	where $\mathcal{X}^{r}=\set{ [d X d'] \mid X\in\mathcal{X}, d,d' \in D}$, and
	\(\prod^r\) is the least set such that each of the following holds:
	\begin{enumerate}
	\item if \( (X\rightarrow \varepsilon)\in\prod\) and \(d\in D\) then \(([dXd]\rightarrow \varepsilon)\in\prod^r\).
	\item if \( (X\rightarrow a)\in\prod\) and \( (d\rightarrow a\cdot d')\in\delta\) then \( ([dXd']\rightarrow a)\in\prod^r\).
	\item if \([d_0 A d_1],[d_1 B d_2]\in\mathcal{X}^r\) and \( (X\rightarrow AB)\in\prod\)
		then \( ([d_0 X d_2]\rightarrow [d_0 Ad_1][d_1 Bd_2])\in\prod^r \).
	\end{enumerate}
\end{definition}

\begin{lemma}
  Let \(\sigma\in (\Sigma\cup\Sigma_i\cup\set{\varepsilon})\), \(d,d'\in D\) and \(X\in\mathcal{X}\).
  \[ \text{if}\quad d\underset{R}\Rightarrow^* \sigma \cdot d' \land X\underset{G}\Rightarrow^* \sigma\quad \text{then}\quad
	[d X d']\underset{G^{r}}\Rightarrow^* \sigma\enspace .\] %
  \label{lem:length0or1}
\end{lemma}
\begin{proof}
  The proof is by induction on the length of the derivation  $X\underset{G}{\Rightarrow}^* \sigma$.

 {\(\mathbf{i=1}.\)} Then \(X\Rightarrow \sigma\). Moreover 
 \(d\underset{R}\Rightarrow^* \sigma\cdot d'\) shows that either
 \(d\underset{R}\Rightarrow \sigma\cdot d'\) or \(d=d'\) and \(\sigma=\varepsilon\)
 (i.e. \(d\underset{R}\Rightarrow^0 \sigma\cdot d'\)).

 In any case we have that \( ([dXd']\rightarrow \sigma)\in\prod^r\) by
 definition of \(G^r\), hence we find that \([dXd']\underset{G^r}{\Rightarrow} \sigma\).

 {\(\mathbf{i>1}.\)} We have  \(X\Rightarrow^{i} \sigma\).
 Then we necessarily have \(X\Rightarrow Z Y \Rightarrow^{j} w_1 Y\Rightarrow^{k} w_1 w_2=\sigma \) where \(j+k=i-1\).
 Two cases may arise: \(w_1=\sigma\) and \(w_2=\varepsilon\) or \(w_1=\varepsilon\) and
 \(w_2=\sigma\). Let us prove the case \(w_1=\sigma\) and \(w_2=\varepsilon\).
 The other one is treated similarly.

 We have \(Y\Rightarrow^{k} w_2(=\varepsilon)\) with \(k\leq i-1\). Moreover for
 each \(d\in D\), we have \(d\underset{R}{\Rightarrow}^* w_2\cdot d\).  Next,
 because \(k\leq i-1\) we can apply the induction hypothesis to conclude that
 \( [dYd] \underset{G^r}{\Rightarrow}^* \varepsilon\) for all \(d\in D\).

 Also \(Z\Rightarrow^{j} w_1(=\sigma)\) with \(j\leq i-1\).
 Moreover \(d\underset{R}{\Rightarrow}^* \sigma\cdot d'\) shows by induction
 that \([dZd'] \underset{G^r}{\Rightarrow}^* \sigma\).
 Finally, \( (X\rightarrow ZY)\in\prod \) and the definition of \(G^{r}\) shows that 
 \( ([dXd']\rightarrow [dYd][dZd'])\in\prod^r \), hence that
 \([dXd']\underset{G^r}{\Rightarrow}^* \sigma\) and we are done.
\end{proof}

\begin{lemma}
  Let \(X_0\underset{G}{\Rightarrow}^* w\) where \(|w|>1\). There exist
  \(X,X_1,X_2\in\mathcal{X}\) and \(w_1, w_2\in (\Sigma\cup\Sigma_i)^*\setminus\set{\varepsilon}\) such that each of the following holds:
  \begin{itemize}
    \item \(X\Rightarrow X_1 X_2 \Rightarrow^* w_1 X_2 \Rightarrow^* w_1 w_2 = w\) 
    \item \(X_0\Rightarrow^* X\)
  \end{itemize}
  \label{lem:findinghi}
\end{lemma}
\begin{proof}
  The proof is by induction of the length of the derivation \(X_0\underset{G}{\Rightarrow}^* w\). Since \(|w|>1\), the smallest derivation for \(w\) needs no less than three steps.  

  {\(\mathbf{i=3}.\)} Then \(X_0\underset{G}{\Rightarrow}^3 w\) is necessarily
  of the form \(X_0\Rightarrow X_1 X_2 \Rightarrow \sigma_1 X_2 \Rightarrow
  \sigma_1 \sigma_2 = w\) where \(\sigma_1\neq\varepsilon\neq\sigma_2\). By choosing \(X=X_0\) we have \(X_0\Rightarrow^* X\) which  concludes the proof of this case.

  {\(\mathbf{i>3}.\)} Then \(X_0\underset{G}{\Rightarrow}^i w\) is necessarily
  of the form \(X_0\Rightarrow X_1 X_2 \Rightarrow^j w_1 X_2 \Rightarrow^k
  w_1 w_2 = w\) with \(j+k=i-1\).

  Three cases may arise:
  \begin{description}
    \item[\(w_1=\varepsilon\) and \(w_2=w\)] Therefore we have that
      \(X_1\Rightarrow^* w_1=\varepsilon\) and \(X_2\Rightarrow^k w_2=w\)
      with \(k\leq i-1\). The induction hypothesis shows that there exists
      \(X',X'_1, X'_2\) and \(w'_1, w'_2\in (\Sigma\cup\Sigma_i)^*\setminus\set{\varepsilon}\) such that \(X'\Rightarrow X'_1 X'_2 \Rightarrow^* w'_1 X'_2 \Rightarrow^* w'_1 w'_2 (=w_2=w)\) and \(X_2 \Rightarrow^* X'\). Finally
      we find that \(X_0\Rightarrow^* X' \Rightarrow X'_1 X'_2 \Rightarrow^* w'_1 X'_2 \Rightarrow^* w'_1 w'_2 =w \) and we are done.
    \item[\(w_1=\varepsilon\) and \(w_2=w\)] This case is similar to the previous one.
    \item[\(w_1\neq\varepsilon\) and \(w_2\neq\varepsilon\)]
	By choosing \(X=X_0\) we have \(X_0\Rightarrow^* X\) which  concludes the proof of this case.
  \end{description}
\end{proof}

\begin{lemma}
  If  \(X_0\underset{G}{\Rightarrow}^* X\underset{G}{\Rightarrow}^* w\) 
  and \([d X d'] \underset{G^r}{\Rightarrow}^* w\) then
  \([d X_0 d'] \underset{G^r}{\Rightarrow}^* w\).
  \label{lem:bringingitup}
\end{lemma}
\begin{proof}
  The proof is by induction on the length of the derivation \(X_0\underset{G}{\Rightarrow}^* X\) 

  {\(\mathbf{i=0}\)}. So we have \(X_0=X\) and the result trivially holds.

  {\(\mathbf{i>0}\)}. We have \(X_0\Rightarrow^{i} X\Rightarrow^* w\). It follows
  that \(X_0\Rightarrow Y Z \Rightarrow^{i-1} X\Rightarrow^* w\).

  Two cases may arise: \(Y\Rightarrow^* \varepsilon\) and \(Z\Rightarrow^* X\)
  or \(Y\Rightarrow^* X\) and \(Z\Rightarrow^* \varepsilon\). We solve the former, the
  proof of the latter being similar.

  Applying Lem.~\ref{lem:length0or1} to \(Y\underset{G}{\Rightarrow}^* \varepsilon\)
  and \(d\underset{R}{\Rightarrow}^* d\) we find that \([dYd]\underset{G^r}{\Rightarrow}^*
  \varepsilon\).  Next since \(Z\Rightarrow^{k} X\) with \(k<i-1\) we find by induction
  hypothesis that \([d Z d']\underset{G^r}{\Rightarrow}^* w\), hence that \([d
  X_0 d']\underset{G^r}{\Rightarrow}^* w\) since \( ([d X_0 d']\rightarrow [d Y
  d][d Z d'])\in\prod^{r}\) and we are done.
\end{proof}

\begin{lemma}
	Let \(w\in(\Sigma\cup\Sigma_i)^*\), \(d,d'\in D\) and $X\in\mathcal{X}$.
	\[ [d X d']\underset{G^{r}}\Rightarrow^* w\quad \text{if{}f}\quad
	d\underset{R}\Rightarrow^* w \cdot d' \land X\underset{G}\Rightarrow^* w\enspace .\] %
	\label{lem:grammarconstruction}%
\end{lemma}%
\begin{proof}
   The proof for the only if direction is by induction on the length of the
   derivation of $[d X d']\Rightarrow^* w$.

  {$\mathbf{i=1}$.} So we conclude from $[d X d']\Rightarrow
  \sigma$ that \(([d X d']\rightarrow \sigma)\in\prod^r\), hence that
   $(X\rightarrow \sigma)\in\prod$ and $(d\rightarrow \sigma\cdot d')\in\delta$ or \(d=d'\) by definition of $G^r$, 
  and finally that $X\Rightarrow \sigma$
  and \(d\Rightarrow \sigma \cdot d'\) and we are done.

  {$\mathbf{i>1}$.} If the derivation of \(G^{r}\) has \(i\) steps with
  \(i>1\), it must be the case that:

  $[d X d'] \Rightarrow [d Z d_{\ell}][d_{\ell} Y d']\Rightarrow^{j} w_1\cdot
  [d_{\ell} Y d']\Rightarrow^{k} w_1w_2$ where \(w=w_1 w_2\) and \(j+k=i-1\).  By
  induction hypothesis, we have $d \Rightarrow^* w_1 \cdot d_{\ell}$ and
  $Z\Rightarrow^* w_1$. Also $d_{\ell}\Rightarrow^* w_2  \cdot d'$ and
  $Y\Rightarrow^* w_2$.  Hence we find that $d \Rightarrow^* w_1 w_2  \cdot d'$
  and $X\Rightarrow^* w_1 w_2$ since $(X\rightarrow Z Y)\in \prod$ and we are
  done since $w=w_1 w_2$.

  \medskip 
  For the if direction, let $w\in\Sigma^*$ such that $X\underset{G}{\Rightarrow}^* w$ and $d
  \underset{R}{\Rightarrow}^* w \cdot d'$. Then the proof goes by induction on the length
  \(i\) of \(w\).

	{\(\mathbf{i=0,1}\).} We have \(d\underset{R}\Rightarrow^* \sigma \cdot d' \land X\underset{G}\Rightarrow^* \sigma\) with \(\sigma\in(\Sigma\cup\Sigma_i\cup\set{\varepsilon})\).
  This coincides with the result of Lem.~\ref{lem:length0or1}.

  {\(\mathbf{i>1}\).}
  Lem.~\ref{lem:findinghi} shows that there exist \(X',X_1,X_2\in\mathcal{X}\)
  and \(w_1,w_2\in(\Sigma\cup\Sigma_i)^*\setminus\set{\varepsilon}\) such
  that \(X\Rightarrow^* X'\Rightarrow X_1X_2\Rightarrow^* w_1 X_2\Rightarrow w_1 w_2=w\).

  Since \(d\Rightarrow^* w\cdot d'\) and \(w_1w_2=w\), the definition of \(R\)
  shows that there exists \(d_{\ell}\in D\) such that
  \(d\Rightarrow^* w_1\cdot d_{\ell} \Rightarrow^* w_1 w_2\cdot d'\).

  Hence we can use that induction hypothesis for \(w_1\) and \(w_2\) which
  shows that \([d X_1 d_{\ell}]\Rightarrow^* w_1\) and \([d_{\ell} X_2
  d']\Rightarrow^* w_2\). Next, we conclude from \( (X'\rightarrow X_1
  X_2)\in\prod \) that \( ([dX'd']\rightarrow [dX_1
  d_{\ell}][d_{\ell}X_2d'])\in\prod^r\), hence that \([dX'd']\Rightarrow^*
  w_1w_2=w\).

  Finally \(X\Rightarrow^* X'\) and the result of
  Lem.~\ref{lem:bringingitup} shows that \([dXd']\Rightarrow^* w\).
\end{proof}

\begin{definition}
	Given $G^r=(\mathcal{X}^r,\Sigma\dotcup\Sigma_i,\prod^r)$ as given in Def.~\ref{def:GR}.
	Define $G^R=(\mathcal{X}^R,\Sigma,\prod^R)$ where $\mathcal{X}^R=\mathcal{X}^r$; and \(\prod^R\) is the smallest set such that if
	\( (X\rightarrow \alpha)\in\prod^r \) then \( (X\rightarrow \proj_{\Sigma\cup\mathcal{X}^R}(\alpha)\in\prod^R) \).
	\label{def:GRalt}
\end{definition}
It is routine to check that Def.~\ref{def:GRalt} is equivalent to
Def.~\ref{def:GR} p.~\pageref{def:GR}.  Finally, we conclude from
Lem.~\ref{lem:grammarconstruction} and Def.~\ref{def:GRalt} that for every
\(d,d'\in D\) and \(X\in\mathcal{X}\) we have:
\( (i) \) let \(w_1\in\Sigma^*\) such that 
\([dXd']{\underset{G^R}\Rightarrow^*}w_1\)
then 
there exists \(w_2\in (\Sigma\dotcup\Sigma_i)^*\)
such that 
\(d{\underset{R}\Rightarrow^*}w_2\cdot d'\), 
\(X{\underset{G}\Rightarrow^*} w_2\), and \(\proj_{\Sigma}(w_2)=w_1\);
\( (ii)\)
let \(w\in (\Sigma\dotcup\Sigma_i)^*\) such that
\(d{\underset{R}\Rightarrow^*}w\cdot d'\), 
\(X{\underset{G}\Rightarrow^*} w\) then
\([dXd']{\underset{G^R}\Rightarrow^*}\proj_{\Sigma}(w)\).
Hence Lem.~\ref{lem:gr_correctness} holds.

\subsection{Reduction from Petri Nets to Boolean Petri Nets}

\begin{lemma}
	(1) Let \((N,\mmap_{\imath})\) be an initialized \(\pn\).  There exists
        a Boolean initialized \(\pn\) 
	\((N',\mmap'_{\imath})\) computable in polynomial time in the size of $(N,\mmap_{\imath})$
        such that $(N,\mmap_{\imath})$ is bounded
        iff $(N',\mmap'_{\imath})$ is bounded. 

	(2) Let \((N,\mmap_{\imath},\mmap_f)\) be an instance of the
	reachability (respectively, coverability) problem.  There exists a Boolean initialized Petri net 
	\((N',\mmap'_{\imath})\) and a Boolean marking \(\mmap'_f\) computable in polynomial
	time such that $\mmap_f$ is reachable (respectively, coverable) in \((N,\mmap_{\imath})\) iff 
        $\mmap'_f$ is reachable (respectively, coverable) in \((N',\mmap'_{\imath})\).
	\label{lem:binaryisnotaproblemappendix}
\end{lemma}
\begin{proof}
	We prove the result in two steps. First, we transform the instances so that the initial
	marking and (in case of coverability and reachability) the target markings are Boolean.
	Second, we transform the instances so that $I(t)$ and $O(t)$ are Boolean for each transition $t$.

	Consider a boundedness problem instance \( (N=(S,T,F),\mmap_{\imath})\).
	In the first step, we define an equivalent instance \(
	(N^{\flat},\mmap^{\flat}_{\imath})\) where the marking $\mmap^{\flat}_{\imath}$ is Boolean
        (but transitions in $N^{\flat}$ need not be Boolean).
        We perform the transformation by adding a new place $p_i$ and a new transition $t_i$
        that consumes a token from $p_i$ and puts $\mmap_{\imath}$ tokens in the other places.
        Initially, $\mmap^{\flat}_{\imath}$ has one token in $p_i$ and zero tokens in all other places.
	Formally, \( N^{\flat}=(S\cup\set{p_i},T\cup\set{t_i}, F^{\flat} = \tuple{I^{\flat},O^{\flat}})\), where  
	\(I^{\flat}(t)=I(t)\) and \(O^{\flat}(t) = O(t)\) for all $t\in T$ and
	\(I^{\flat}(t_i)=\multi{p_i}\) and \(O^{\flat}(t_i) = \mmap_{\imath}\).

	Consider now a coverability problem instance
	\( (N,\mmap_{\imath},\mmap)\). 
        To replace $\mmap_{\imath}$ and $\mmap$ by Boolean markings, intuitively, we add two new
        places $p_i$ and $p_c$ to $N$. As in the case of boundedness, there is a single transition 
        out of $p_i$ that consumes one token and produces $\mmap_{\imath}$.
        Additionally, there is one transition that consumes $\mmap$ and produces a single token in
        $p_c$.
        Formally, define \(N^{\flat}=(S\cup\set{p_i,p_c},T\cup\set{t_i,t_c}, F^{\flat})\) with 
	\(F^{\flat}(T)=F(T)\), \(F^{\flat}(t_i)=\tuple{\multi{p_i},\mmap_{\imath}}\)
	and \(F^{\flat}(t_c)=\tuple{\mmap,\multi{p_c}}\).  The initial and target
	marking are respectively given by \(\multi{p_i}\) and \(\multi{p_c}\) each of
	which is Boolean. 

	Let us turn to a reachability problem instance  
        \((N,\mmap_{\imath},\mmap)\).
        The initial marking is made Boolean using the same trick: add a new place $p_i$ and
        add a transition that consumes one token from $p_i$ and produces $\mmap_{\imath}$ tokens.
        To get rid of $\mmap$, we use a construction from \cite{hack76} and additionally, we add a new place $p_r$.
        Then, we change each transition of $N$ to additionally consume a token from $p_r$ and produce
        a token back in $p_r$.
        Finally, we add a new transition that consumes $\mmap \oplus \multi{p_r}$ tokens and produces no tokens.
        The initial marking puts one token each at $p_i$ and $p_r$, and we ask if the marking where every place
        has zero tokens is reachable.
        Formally, define \(
	N^{\flat}=(S\cup\set{p_i,p_r},T\cup\set{t_i,t_r}, F^{\flat})\) such that
	\(F^{\flat}(t)=\tuple{\multi{p_r}\oplus I(t), \multi{p_r}\oplus O(t)}\) where
	\(F(t)=\tuple{I(t),O(t)}\),
	\(F^{\flat}(t_i)=\tuple{\multi{p_i},\mmap_{\imath}}\) and
	\(F^{\flat}(t_r)=\tuple{\mmap\oplus\multi{p_r},\varnothing}\).  The initial
	and target marking are respectively given by \(\multi{p_i,p_r}\) and
	\(\varnothing\) the empty marking each of those marking being a set.

	We now move to the second step of the construction.
	Given a \(\pn\) \(N=(S,T,F)\), we show how to compute in polynomial time a \(\pn\) \(N'=(S',T',F')\)
	such that for every transition \(t\in T'\) the multisets \(I(t)\) and
	\(O(t)\) are Boolean. 
        The construction is independent of the decision problem (boundedness, coverability, or reachability).

	\begin{figure}[t]
		\centering
			\includegraphics[scale=.6]{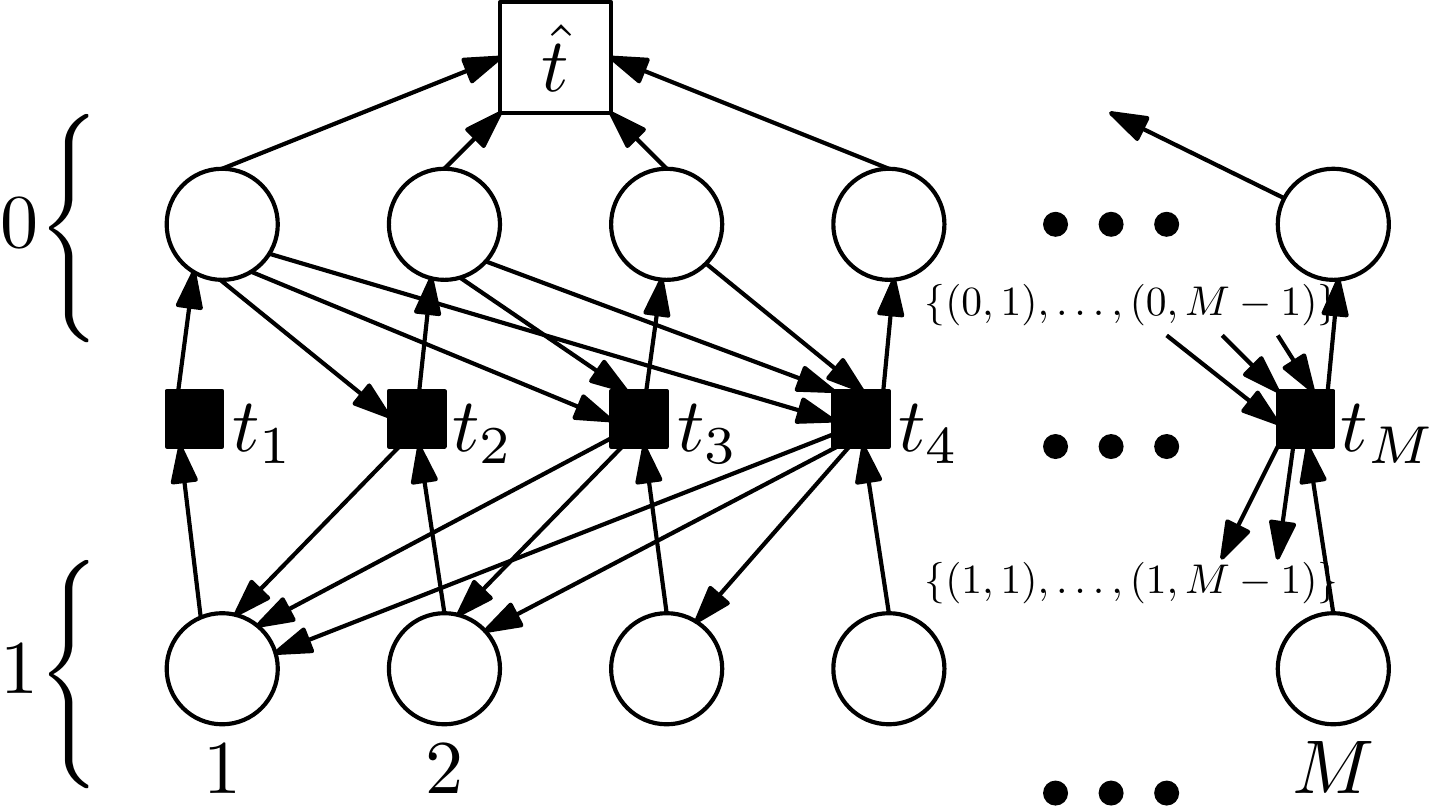}
		\caption{A Petri net widget, left to right is from the least significant
		bit to the most significant bit. 
                }
		\label{fig:pnwidget}
	\end{figure}

	Assume that \(S\) is given by \(\set{s_1,\ldots,s_n}\) and \(T\) is given by \(\set{t_1,\ldots,t_k}\).

        We convert \(N\) to a Boolean Petri net in five steps.
	First, we define the \(\pn\) \(N_1 = (S_1, T_1, F_1)\).
 	The set of places \(S_1 = S\).
	For each \(t\in T\), we define the transitions
 	\(t^{I}_1, t^{I}_2, \ldots, t^{I}_{n}\), \(t^{O}_1, t^{O}_2, \ldots, t^{O}_{n}\) in \(T_1\) such that:
	\begin{itemize}
		\item \(F_1(t^{I}_i)=\tuple{\proj_{\set{s_i}}(I(t)),\varnothing}\) and \(F_1(t^{O}_i)=\tuple{\varnothing,\proj_{\set{s_i}}(O(t))}\) for \(i\in\set{1,\ldots,n}\).
	\end{itemize}
	Intuitively, to each pair \((s_i,t)\) (\(i\in\set{1,\ldots,n}\), \(t\in T\)) 
	we associate two transitions \(t^I_{i}\) and \(t^O_{i}\) of \(T_1\) which we will use to simulate the 
        effect of \(t\) on \(s_i\).
	
	Second, we define the \(\pn\) \(N_2\) which is given by the synchronized product of \(N_1\) with the following
	regular language over alphabet \(T_1\):

	\noindent %
	\hspace*{\stretch{1}}\(L\stackrel{\mathit{def}}{=}(w_1 + \cdots + w_k)^* \) \hfill\vspace{0pt}\\
        where each \(w_i = t^I_{i1} t^O_{i2} \ldots t^I_{in} t^O_{in}\) is a 
        finite word that simulates the firing of transition \(t_i\in T\) for \(i\in\set{1,\ldots,k}\).
	Clearly, since each \(w_i\) corresponds to the firing of transition \(t_i\in T\) we find that \(N_2\) simulates \(N\) 
        (i.e., 
	\(\mmap\fire{t}\) does not hold in \(T\) if{}f \(\mmap\fire{t^I_i}\) does not hold from some \(i\in\set{1,\ldots,n}\);
	and \(\mmap\fire{t}\mmap'\) if{}f \(\mmap\fire{w}\mmap'\)).

	Observe that \(N_2\) is still not a Boolean \(\pn\). 
        In the third step, we 
	replace each transition \(t^{O}_i\) (resp. \(t^{I}_i\)) which 
	produce (resp. consume) \(\proj_{\set{s_i}}(O(t))\) (resp. \(\proj_{\set{s_i}}(I(t))\)) tokens to (resp. from) place \(s_i\)
	by a Boolean \(\pn\) \(N_{t^{O}_i}\) (resp. \(N_{t^{I}_i}\)).
        We do this by defining the following class of widgets.

	Let us consider a transition \(t^{O}_i\) which produces \(m\) tokens
	into \(s_i\), and let $M = \lceil \log_2{m}\rceil$. 
        We will substitute \(t^{O}_i\) with a Boolean \(\pn\) \(N_{t^{O}_i}\). 
        We call such a \(\pn\) a widget. A generic description of a
	widget is given in Fig.~\ref{fig:pnwidget}. 
	
	Intuitively, the widget behaves like a binary decrementer.  
        To begin with, we shall put a $(0,1)$-marking on the widget, where
        for each ``column'' labeled $1,\ldots, M$, we put a single token
        in either the 0th row or the 1st row.
	Each $(0,1)$-marking coincides with the binary representation of a number in the range
        $[0, 2^M - 1]$, obatined by $\sum_{i=1}^M \delta_i 2^i$, where $\delta_i = 1$ if the
        $(0,1)$-marking places a token in the 1st row of column $i$ and $\delta_i = 0$ if
        the $(0,1)$-marking places a token in the 0th row of column $i$.
	Conversely, every number in the range \([0,2^M-1]\) corresponds to
	exactly one $(0,1)$-marking of the widget. 
        Let \(f\) be the function which
	takes as input a number in the range \([0, 2^{M} -1]\) and returns the corresponding
	$(0,1)$-marking.

	One can check that the widget defines a Boolean \(\pn\).  
        Moreover, from every $(0,1)$-marking there exists exactly one enabled
	transition in the widget.
        Hence the widget behaves as follows: 
        starting from marking \(f(m)\) there exists a unique maximal
	sequence of enabled transitions which consists of \(m\) transitions in
	\(\set{t_1,\ldots,t_n}\) followed by \(\hat{t}\) enabled at the marking which
	represents \(0\) in binary (i.e., the $(0,1)$-marking that puts a single token each
        in the 0th row of each column).  
        Next, we add transition \(\check{t}\) whose
	role is to initialize the widget with marking \(f(m)\). 
        Therefore we have
	\(F_{t^O_i}(\check{t})=\tuple{\varnothing,f(m)}\).  Finally let us add an arc from every
	transition of the widget except \(\hat{t}\) and \(\check{t}\) into place \(s_i\).

	From the above construction, we observe that the firing of any sequence in
	the language \(\check{t} \cdot (\set{t_1,\ldots,t_n})^* \cdot \hat{t}\) has the 
	effect of producing exactly \(m\) tokens in place \(s_i\).

	Using a similar reasoning one can define a widget for \(t_i^{I}\). 

	In the fourth step, let us define \(N_3\) as the \(\pn\) which is given by the union of all
	the widgets (therefore \(S\) is contained in the places of \(N_3\)).  Given
	\(i\in \set{1,\ldots,n}\), let us denote by \(T_{t_i^I}\) and \(T_{t_i^{O}}\)
	the set of transitions of the widget corresponding to \(t_i^{I}\) and
	\(t_i^{O}\), respectively.  Also we have transitions
	\(\check{t}_i^{I},\hat{t}_i^{I},\check{t}_i^{O},\hat{t}_i^{O}\).  Observe
	that \(N_3\) is a Boolean \(\pn\).

	Finally, to conclude the construction of the Boolean \(\pn\) \(N'\), we define \(N'\)
	as the synchronized product of \(N_3\) with the language \(\tau(L)\) where
	\(\tau\) is a substitution which maps \(t_i^I\) onto the language \( (\check{t}_i^{I} \cdot
	(T_{t_i^I})^* \cdot \hat{t}_i^{I}) \) and \(t_i^O\) onto the language \( (\check{t}_i^{O} \cdot
	(T_{t_i^O})^* \cdot \hat{t}_i^{O}) \).

	It is routine to check that the obtained \(\pn\) is Boolean and it can be
	computed in polynomial time in the size of \(N\).
\end{proof}

\received{April 2011}{xxx XXXX}{xxx XXXX}

\begin{acks}
We thank Mohamed-Faouzi Atig, Andrey Rybalchenko, Bishesh Adhikari and Laurent
Van Begin for useful discussions, and the anonymous referees for many useful
comments.  We thank Laurent Van Begin for the proof of
Th.~\ref{thm:coverpntiundec}.
\end{acks}

\bibliographystyle{acmsmall}
\bibliography{ref}

\begin{thebibliography}{}

\bibitem[\protect\citeauthoryear{Abdulla, Cerans, Jonsson, and Tsay}{Abdulla
  et~al\mbox{.}}{1996}]{ACJT96}
{\sc Abdulla, P.~A.}, {\sc Cerans, K.}, {\sc Jonsson, B.}, {\sc and} {\sc Tsay,
  Y.-K.} 1996.
\newblock General decidability theorems for infinite-state systems.
\newblock In {\em LICS$\:$'96: Proc.\ 11th Annual IEEE Symp. on Logic in
  Computer Science}. IEEE Computer Society, 313--321.

\bibitem[\protect\citeauthoryear{Abdulla and Mayr}{Abdulla and
  Mayr}{2009}]{AbdullaM09}
{\sc Abdulla, P.~A.} {\sc and} {\sc Mayr, R.} 2009.
\newblock Minimal cost reachability/coverability in priced timed petri nets.
\newblock In {\em FOSSACS$\:$'09: Proc.\ 12th Int.\ Conf.\ Foundations of
  Software Science and Computational Structures}. LNCS Series, vol. 5504.
  Springer, 348--363.

\bibitem[\protect\citeauthoryear{Aho, Sethi, and Ullman}{Aho
  et~al\mbox{.}}{1986}]{ASU86}
{\sc Aho, A.}, {\sc Sethi, R.}, {\sc and} {\sc Ullman, J.~D.} 1986.
\newblock {\em Compilers: Principles, Techniques, and Tools}.
\newblock Addison-Wesley.

\bibitem[\protect\citeauthoryear{Atig, Bouajjani, and Touili}{Atig
  et~al\mbox{.}}{2008}]{ABT08}
{\sc Atig, M.~F.}, {\sc Bouajjani, A.}, {\sc and} {\sc Touili, T.} 2008.
\newblock Analyzing asynchronous programs with preemption.
\newblock In {\em FSTTCS$\:$'08: Proc.\ 28th Int.\ Conf.\ on Fondation of
  Software Technology and Theoretical Computer Science}. Leibniz International
  Proceedings in Informatics (LIPIcs) Series, vol.~2. Leibniz-Zentrum fuer
  Informatik, 37--48.

\bibitem[\protect\citeauthoryear{Atig and Habermehl}{Atig and
  Habermehl}{2009}]{RP09}
{\sc Atig, M.~F.} {\sc and} {\sc Habermehl, P.} 2009.
\newblock On {Y}en's path logic for {P}etri nets.
\newblock In {\em RP$\:$'09: Proc.\ 3rd Workshop on Reachability Problems}.
  LNCS Series, vol. 5797. Springer, 51--63.

\bibitem[\protect\citeauthoryear{Bouajjani, Esparza, and Maler}{Bouajjani
  et~al\mbox{.}}{1997}]{BEM97}
{\sc Bouajjani, A.}, {\sc Esparza, J.}, {\sc and} {\sc Maler, O.} 1997.
\newblock Reachability analysis of pushdown automata: {A}pplication to
  model-checking.
\newblock In {\em CONCUR$\:$'97: Proc.\ 8th Int.\ Conf.\ on Concurrency
  Theory}. LNCS Series, vol. 1243. Springer, 135--150.

\bibitem[\protect\citeauthoryear{Burkart and Steffen}{Burkart and
  Steffen}{1994}]{BurkartS94}
{\sc Burkart, O.} {\sc and} {\sc Steffen, B.} 1994.
\newblock Pushdown processes: Parallel composition and model checking.
\newblock In {\em CONCUR$\:$'94: Proc.\ 5th Int.\ Conf.\ on Concurrency
  Theory}. LNCS Series, vol. 836. Springer, 98--113.

\bibitem[\protect\citeauthoryear{Chadha and Viswanathan}{Chadha and
  Viswanathan}{2007}]{ChaV07}
{\sc Chadha, R.} {\sc and} {\sc Viswanathan, M.} 2007.
\newblock Decidability results for well-structured transition systems with
  auxiliary storage.
\newblock In {\em CONCUR$\:$'07: Proc.\ 18th Int.\ Conf.\ on Concurrency
  Theory}. LNCS Series, vol. 4703. Springer, 136--150.

\bibitem[\protect\citeauthoryear{Chadha and Viswanathan}{Chadha and
  Viswanathan}{2009}]{cv-tcs09}
{\sc Chadha, R.} {\sc and} {\sc Viswanathan, M.} 2009.
\newblock Deciding branching time properties for asynchronous programs.
\newblock {\em Theor. Comput. Sci.\/}~{\em 410,\/}~42, 4169--4179.

\bibitem[\protect\citeauthoryear{Dickson}{Dickson}{1913}]{dic13}
{\sc Dickson, L.~E.} 1913.
\newblock Finiteness of the odd perfect and primitive abundant numbers with $n$
  distinct prime factors.
\newblock {\em Amer. J. Math.\/}~{\em 35}, 413--422.

\bibitem[\protect\citeauthoryear{Dufourd, Finkel, and Schnoebelen}{Dufourd
  et~al\mbox{.}}{1998}]{DFS98}
{\sc Dufourd, C.}, {\sc Finkel, A.}, {\sc and} {\sc Schnoebelen, P.} 1998.
\newblock Reset nets between decidability and undecidability.
\newblock In {\em ICALP$\:$'98: Proc.\ of 25th Int.\ Colloquium on Automata,
  Languages and Programming}. LNCS Series, vol. 1443. Springer, 103--115.

\bibitem[\protect\citeauthoryear{Esparza}{Esparza}{1997}]{Esp97}
{\sc Esparza, J.} 1997.
\newblock Petri nets, commutative context-free grammars, and basic parallel
  processes.
\newblock {\em Fundamenta Informaticae\/}~{\em 31}, 13--26.

\bibitem[\protect\citeauthoryear{Esparza}{Esparza}{1998}]{esparza-course}
{\sc Esparza, J.} 1998.
\newblock Decidability and complexity of petri net problems -- an introduction.
\newblock In {\em Lectures on Petri Nets I: Basic Models}. LNCS Series, vol.
  1491. Springer, 374--428.

\bibitem[\protect\citeauthoryear{Esparza, Finkel, and Mayr}{Esparza
  et~al\mbox{.}}{1999}]{EFM99}
{\sc Esparza, J.}, {\sc Finkel, A.}, {\sc and} {\sc Mayr, R.} 1999.
\newblock On the verification of broadcast protocols.
\newblock In {\em LICS$\:$'99: Proc.\ 14th Annual IEEE Symp. on Logic in
  Computer Science}. IEEE Computer Society, 352--359.

\bibitem[\protect\citeauthoryear{Esparza, Ganty, Kiefer, and
  Luttenberger}{Esparza et~al\mbox{.}}{2011}]{egkl11-ipl}
{\sc Esparza, J.}, {\sc Ganty, P.}, {\sc Kiefer, S.}, {\sc and} {\sc
  Luttenberger, M.} 2011.
\newblock Parikh's theorem: A simple and direct automaton construction.
\newblock {\em Information Processing Letters\/}~{\em 111}, 614--619.

\bibitem[\protect\citeauthoryear{Esparza, Kiefer, and Luttenberger}{Esparza
  et~al\mbox{.}}{2010}]{EKL10:JACM}
{\sc Esparza, J.}, {\sc Kiefer, S.}, {\sc and} {\sc Luttenberger, M.} 2010.
\newblock Newtonian program analysis.
\newblock {\em Journal of the ACM\/}~{\em 57,\/}~6, 33:1--33:47.

\bibitem[\protect\citeauthoryear{Esparza and Nielsen}{Esparza and
  Nielsen}{1994}]{EN94}
{\sc Esparza, J.} {\sc and} {\sc Nielsen, M.} 1994.
\newblock Decibility issues for {P}etri nets - a survey.
\newblock {\em Journal of Informatik Processing and Cybernetics\/}~{\em
  30,\/}~3, 143--160.

\bibitem[\protect\citeauthoryear{Finkel and Sangnier}{Finkel and
  Sangnier}{2010}]{FinkelS10}
{\sc Finkel, A.} {\sc and} {\sc Sangnier, A.} 2010.
\newblock Mixing coverability and reachability to analyze vass with one
  zero-test.
\newblock In {\em SOFSEM$\:$'10: Proc.\ 36th Conf.\ on Current Trends in Theory
  and Practice of Computer Science}. LNCS Series, vol. 5901. Springer,
  394--406.

\bibitem[\protect\citeauthoryear{Finkel and Schnoebelen}{Finkel and
  Schnoebelen}{2001}]{FS01}
{\sc Finkel, A.} {\sc and} {\sc Schnoebelen, P.} 2001.
\newblock Well-structured transition systems everywhere!
\newblock {\em Theoretical Computer Science\/}~{\em 256,\/}~1-2, 63--92.

\bibitem[\protect\citeauthoryear{Ganty and Majumdar}{Ganty and
  Majumdar}{2009}]{gm09}
{\sc Ganty, P.} {\sc and} {\sc Majumdar, R.} 2009.
\newblock Analyzing real-time event-driven programs.
\newblock In {\em FORMATS$\:$'09: Proc.\ 7th Int.\ Conf.\ on Formal Modelling
  and Analysis of Timed Systems}. LNCS Series, vol. 5813. Springer, 164--178.

\bibitem[\protect\citeauthoryear{Ganty, Majumdar, and Rybalchenko}{Ganty
  et~al\mbox{.}}{2009}]{gmr09}
{\sc Ganty, P.}, {\sc Majumdar, R.}, {\sc and} {\sc Rybalchenko, A.} 2009.
\newblock Verifying liveness for asynchronous programs.
\newblock In {\em POPL$\:$'09: Proc.\ 36th ACM SIGACT-SIGPLAN Symp.\ on
  Principles of Programming Languages}. ACM Press, 102--113.

\bibitem[\protect\citeauthoryear{Hack}{Hack}{1976}]{hack76}
{\sc Hack, M. H.~T.} 1976.
\newblock Decidability questions for {P}etri nets.
\newblock Tech. Rep. 161, MIT. June.

\bibitem[\protect\citeauthoryear{Hill, Szewczyk, Woo, Hollar, Culler, and
  Pister}{Hill et~al\mbox{.}}{2000}]{TinyOS}
{\sc Hill, J.~L.}, {\sc Szewczyk, R.}, {\sc Woo, A.}, {\sc Hollar, S.}, {\sc
  Culler, D.~E.}, {\sc and} {\sc Pister, K. S.~J.} 2000.
\newblock System architecture directions for networked sensors.
\newblock In {\em ASPLOS$\:$'00 Proc.\ 9th Int.\ Conf.\ on Architectural
  Support for Programming Languages and Operating Systems}. ACM, 93--104.

\bibitem[\protect\citeauthoryear{Jhala and Majumdar}{Jhala and
  Majumdar}{2007}]{jmpopl07}
{\sc Jhala, R.} {\sc and} {\sc Majumdar, R.} 2007.
\newblock Interprocedural analysis of asynchronous programs.
\newblock In {\em POPL$\:$'07: Proc.\ 34th ACM SIGACT-SIGPLAN Symp.\ on
  Principles of Programming Languages}. ACM Press, 339--350.

\bibitem[\protect\citeauthoryear{Karp and Miller}{Karp and
  Miller}{1969}]{KarpM69}
{\sc Karp, R.~M.} {\sc and} {\sc Miller, R.~E.} 1969.
\newblock Parallel program schemata.
\newblock {\em Journal of Comput. Syst. Sci.\/}~{\em 3,\/}~2, 147--195.

\bibitem[\protect\citeauthoryear{Kohler, Morris, Chen, Jannotti, and
  Kaashoek}{Kohler et~al\mbox{.}}{2000}]{Click}
{\sc Kohler, E.}, {\sc Morris, R.}, {\sc Chen, B.}, {\sc Jannotti, J.}, {\sc
  and} {\sc Kaashoek, M.} 2000.
\newblock The {C}lick modular router.
\newblock {\em ACM TOCS\/}~{\em 18,\/}~3, 263--297.

\bibitem[\protect\citeauthoryear{Kosaraju}{Kosaraju}{1982}]{kosaraju82}
{\sc Kosaraju, S.~R.} 1982.
\newblock Decidability of reachability in vector addition systems (preliminary
  version).
\newblock In {\em STOC$\:$'82: Proc.\ of 14th ACM symp.\ on Theory of
  Computing}. ACM, 267--281.

\bibitem[\protect\citeauthoryear{Krohn, Kohler, and Kaashoek}{Krohn
  et~al\mbox{.}}{2007}]{KKK07}
{\sc Krohn, M.}, {\sc Kohler, E.}, {\sc and} {\sc Kaashoek, M.} 2007.
\newblock Events can make sense.
\newblock In {\em USENIX Annual Technical Conference}. USENIX Association.

\bibitem[\protect\citeauthoryear{Lambert}{Lambert}{1992}]{Lambert92}
{\sc Lambert, J.~L.} 1992.
\newblock A structure to decide reachability in petri nets.
\newblock {\em Theor. Comput. Sci.\/}~{\em 99,\/}~1, 79--104.

\bibitem[\protect\citeauthoryear{Lange and Lei{\ss}}{Lange and
  Lei{\ss}}{2010}]{LL10}
{\sc Lange, M.} {\sc and} {\sc Lei{\ss}, H.} 2008-2010.
\newblock To \uppercase{CNF} or not to \uppercase{CNF}? {A}n efficient yet
  presentable version of the \uppercase{CYK} algorithm.
\newblock {\em Informatica Didactica\/}~{\em 8}, 1--21.

\bibitem[\protect\citeauthoryear{Lipton}{Lipton}{1976}]{Lipton}
{\sc Lipton, R.} 1976.
\newblock The reachability problem is exponential-space hard.
\newblock Tech. Rep.~62, Department of Computer Science, Yale University. Jan.

\bibitem[\protect\citeauthoryear{Mayr}{Mayr}{1981}]{Mayr81}
{\sc Mayr, E.~W.} 1981.
\newblock An algorithm for the general petri net reachability problem.
\newblock In {\em STOC$\:$'81: Proc.\ of 13th ACM symp.\ on Theory of
  computing}. ACM, 238--246.

\bibitem[\protect\citeauthoryear{Mayr and Meyer}{Mayr and Meyer}{1981}]{Meyer}
{\sc Mayr, E.~W.} {\sc and} {\sc Meyer, A.~R.} 1981.
\newblock The complexity of the finite containment problem for petri nets.
\newblock {\em Journal of the ACM\/}~{\em 28,\/}~3, 561--576.

\bibitem[\protect\citeauthoryear{Minsky}{Minsky}{1967}]{Min67}
{\sc Minsky, M.} 1967.
\newblock {\em Finite and {I}nfinite {M}achines}.
\newblock Englewood Cliffs, N.J., Prentice-Hall.

\bibitem[\protect\citeauthoryear{Pai, Druschel, and Zwaenepoel}{Pai
  et~al\mbox{.}}{1999}]{Flash}
{\sc Pai, V.}, {\sc Druschel, P.}, {\sc and} {\sc Zwaenepoel, W.} 1999.
\newblock Flash: An efficient and portable web server.
\newblock In {\em Proc. USENIX Tech.\ Conf.} Usenix, 199--212.

\bibitem[\protect\citeauthoryear{Parikh}{Parikh}{1966}]{Parikh66}
{\sc Parikh, R.~J.} 1966.
\newblock On context-free languages.
\newblock {\em Journal of the ACM\/}~{\em 13,\/}~4, 570--581.

\bibitem[\protect\citeauthoryear{Rackoff}{Rackoff}{1978}]{Rackoff78}
{\sc Rackoff, C.} 1978.
\newblock The covering and boundedness problems for vector addition systems.
\newblock {\em Theoretical Computer Science\/}~{\em 6,\/}~2, 223--231.

\bibitem[\protect\citeauthoryear{Reisig}{Reisig}{1986}]{Rei86}
{\sc Reisig, W.} 1986.
\newblock {\em {P}etri {N}ets. {A}n introduction}.
\newblock Springer.

\bibitem[\protect\citeauthoryear{Reps, Horwitz, and Sagiv}{Reps
  et~al\mbox{.}}{1995}]{RHS95}
{\sc Reps, T.}, {\sc Horwitz, S.}, {\sc and} {\sc Sagiv, M.} 1995.
\newblock Precise interprocedural dataflow analysis via graph reachability.
\newblock In {\em POPL$\:$'95: Proc.\ 22nd ACM SIGACT-SIGPLAN Symp.\ on
  Principles of Programming Languages}. ACM, 49--61.

\bibitem[\protect\citeauthoryear{Sen and Viswanathan}{Sen and
  Viswanathan}{2006}]{SenV06}
{\sc Sen, K.} {\sc and} {\sc Viswanathan, M.} 2006.
\newblock Model checking multithreaded programs with asynchronous atomic
  methods.
\newblock In {\em CAV$\:$'06: Proc.\ 18th Int.\ Conf.\ on Computer Aided
  Verification}. LNCS Series, vol. 4144. Springer, 300--314.

\bibitem[\protect\citeauthoryear{Sharir and Pnueli}{Sharir and
  Pnueli}{1981}]{sp81}
{\sc Sharir, M.} {\sc and} {\sc Pnueli, A.} 1981.
\newblock Two approaches to interprocedural data flow analysis.
\newblock In {\em Program Flow Analysis: Theory and Applications}.
  Prentice-Hall, Inc{.}, Chapter~7, 189--233.

\bibitem[\protect\citeauthoryear{Vardi}{Vardi}{1991}]{Vardi91}
{\sc Vardi, M.~Y.} 1991.
\newblock Verification of concurrent programs --- the automata-theoretic
  approach.
\newblock {\em Annals of Pure and Applied Logic\/}~{\em 51}, 79--98.

\bibitem[\protect\citeauthoryear{Walukiewicz}{Walukiewicz}{2001}]{wal01}
{\sc Walukiewicz, I.} 2001.
\newblock {P}ushdown {P}rocesses: {G}ames and {M}odel-{C}hecking.
\newblock {\em Information and Computation\/}~{\em 164,\/}~2, 234--263.

\bibitem[\protect\citeauthoryear{Yen}{Yen}{1992}]{yen:paths}
{\sc Yen, H.-C.} 1992.
\newblock A unified approach for deciding the existence of certain petri net
  paths.
\newblock {\em Information and Computation\/}~{\em 96,\/}~1, 119--137.

\end{thebibliography}

\end{document}